\documentclass
[
    twoside,                 % The thesis is formatted like a book. That is, odd and even pages are handled differently.
    openright,               % Starts a new chapter on an odd page number (right side).
    cleardoublepage = empty, % Clear pages inserted in order to have new chapters appear on odd pages are formatted with an empty style.
    fontsize = 12 pt,        % The size of the font.
    american,                % Support for American English.
    captions = tableheading, % Places the correct amount of space when the caption of a table is above the table.
    numbers = noenddot,      % Does not use a period at the end of numbered titles, such as sections or figures.
    footheight = 35 pt,      % Defines the height of the foot. Due to the line, it needs extra height.
]
{scrbook}

\newif\ifprintVersion   % Defines a binary variable that signals whether the document is prepared for physical or digital print.
\newif\ifprofessionalPrint % Defines a binary variable that signals whether the print will be done by a professional printing service that requests extra margin for page cutting and is not bound to paper formats like A4.
\newif\iffancyTheorems  % Defines a binary variable that signals whether theorems are formatted in the classical style or in a new format that better suits the overall flavor of this thesis.
\newif\ifboldNumberSets % Defines a binary variable that signals whether the variables for number sets (like N or R) should be in bold. If not, they are in blackboard bold instead.
\newif\ifbachelorThesis % Defines a binary variable that signals whether this thesis is a bachelor thesis (true) or a master thesis (false).

\printVersionfalse
\professionalPrintfalse
\fancyTheoremstrue
\boldNumberSetstrue
\bachelorThesistrue

\newcommand*{\printTitle}{}
\newcommand*{\printGermanTitle}{}
\newcommand*{\myTitle}[2]{\renewcommand*{\printTitle}{#1}\renewcommand*{\printGermanTitle}{#2}}
\newcommand*{\printTitleBold}{\textbf{\printTitle}}

\newcommand*{\printAuthor}{}
\newcommand*{\myName}[1]{\renewcommand*{\printAuthor}{#1}}

\newcommand*{\printProgram}{}
\newcommand*{\myProgram}[1]{\renewcommand*{\printProgram}{#1}}

\newcommand*{\printDateReceived}{}
\newcommand*{\dateOfHandingIn}[1]{\renewcommand*{\printDateReceived}{#1}}

\newcommand*{\printSubject}{}
\newcommand*{\mySubject}[1]{\renewcommand*{\printSubject}{#1}}

\newcommand*{\printKeywords}{}
\newcommand*{\myKeywords}[1]{\renewcommand*{\printKeywords}{#1}}

\newcommand*{\printNameOfSupervisor}{}
\newcommand*{\nameOfMySupervisor}[1]{\renewcommand*{\printNameOfSupervisor}{#1}}

\newcommand*{\printAdditionalExaminers}{}
\newcommand*{\additionalExaminers}[1]{\renewcommand*{\printAdditionalExaminers}{#1}}

\newlength{\extraborderlength}
\newcommand*{\extraBorder}[1]{\setlength{\extraborderlength}{#1}}

\newlength{\mybindingcorrection}
\newcommand*{\bindingCorrection}[1]{\setlength{\mybindingcorrection}{#1}} % Contains commands that are used for certain information that is printed.

    \extraBorder{3 mm}

    \bindingCorrection{6 mm}

    \bachelorThesisfalse

    \myTitle{Graph Discovery and Source Detection in Temporal Graphs}{Grapherkundung und Quellensuche \\ in temporalen Graphen}

    \myName{Ben Bals}

    \myProgram{IT-Systems Engineering}

    \dateOfHandingIn{30. September 2024}

    \nameOfMySupervisor{Prof.\,Dr. Tobias Friedrich\newline Prof.\,Dr. Ralf Herbrich}

    \additionalExaminers{Dr. Georgios Skretas \newline Michelle Döring \newline Nicolas Klodt}

    \mySubject{Temporal Graphs}

    \myKeywords{master thesis | temporal graphs | source detection | network discovery}

\usepackage[utf8]{inputenc} % Allows to input UTF8 characters.
\usepackage[T1]{fontenc}    % Allows to print special characters correctly.
\usepackage
[
    ngerman,         % German is used for the German abstract.
    main = american, % This is the main language of the thesis.
]
{babel}                     % Is responsible for sensible hyphenations.

\input glyphtounicode
\usepackage{calc} % Makes it easer to do math with TeX measurements.

\newlength{\myparindent}
\newlength{\myparskip}
\setfootnoterule{0 cm}

\deffootnote[1.2 em]{1.2 em}{0 em}{\makebox[1.4 em][l]{\textbf{\thefootnotemark}}}

\makeatletter%
    \@removefromreset{footnote}{chapter}%
\makeatother

\usepackage[dvipsnames]{xcolor} % Allows it to define colors. The option says that common names can be used.

\definecolor{stroke1}{HTML}{2574A9} % This color is used as the standard color to highlight things.

\colorlet{captionlabel}{black}
\colorlet{footerpagenr}{black}
\colorlet{footerchapter}{stroke1}
\colorlet{footerchaptername}{black}
\colorlet{footersection}{stroke1}
\colorlet{footersectionname}{black}
\colorlet{chapternumber}{stroke1}

\newlength{\mypaperwidth}
\newlength{\mypaperheight}
\newlength{\mybodywidth}
\newlength{\mybodyheight}
\newlength{\myoutermargin}
\ifprintVersion
    \ifprofessionalPrint
    \else
    \fi
\else
\fi

\newlength{\mytopmargin}
\ifprintVersion
    \ifprofessionalPrint
    \fi
\fi

\newlength{\myinnermargin}
\ifprintVersion
    \ifprofessionalPrint
    \fi
\fi

\newlength{\mybottommargin}
\ifprintVersion
    \ifprofessionalPrint
    \fi
\fi

\newcommand{\goldenratio}{1.618}

\newlength{\myheadsep} % Distance from the header to the body.
\ifprintVersion
    \ifprofessionalPrint
    \fi
\fi

\newlength{\myfootskip} % Distance from the body to the footer.
\ifprintVersion
    \ifprofessionalPrint
    \fi
\fi

\newlength{\mymargininnersep} % Distance between the margin and the body.
\newlength{\mymarginoutersep} % Distance between the margin and the paper border.
\ifprintVersion
    \ifprofessionalPrint
    \fi
\fi

\newlength{\mymarginwidth} % Width of the margin.
\newlength{\mymarginwidthwithinnersep} % Width of the margin.
\usepackage
[
    \ifprintVersion
        \ifprofessionalPrint
            paperwidth = \mypaperwidth + 2\extraborderlength + \mybindingcorrection,
            paperheight = \mypaperheight + 2\extraborderlength,
        \else
            paperwidth = \mypaperwidth,
            paperheight = \mypaperheight,
        \fi
    \else
        paperwidth = \mypaperwidth,
        paperheight = \mypaperheight,
    \fi
    textwidth = \mybodywidth,
    textheight = \mybodyheight,
    outer = \myoutermargin,
    top = \mytopmargin,
    headsep = \myheadsep,
    footskip = \myfootskip,
    marginparsep = \mymargininnersep,
    marginparwidth = \mymarginwidth,
]
{geometry} % Used in order to define the dimensions of the page and its layout.

\usepackage
[
]
{scrlayer-scrpage} % Allows to adjust the definitions of the head and foot of a page.

\clearpairofpagestyles

\KOMAoptions
{%
    headwidth = \textwidth + \mymarginwidthwithinnersep,%
    footwidth = \myoutermargin : \textwidth,%
}

\automark[chapter]{chapter}
\automark*[section]{}

\lehead%
{%
    \begin{minipage}[b]{\mymarginwidth}%
        \small\raggedleft\normalfont\textsf{\textbf{\color{footerchapter}\chaptername\ \thechapter}}
    \end{minipage}
}
\cehead{\hspace*{\mymarginwidthwithinnersep}\parbox{\textwidth}{\raggedright\leftmark}}

\rohead%
{%
    \Ifstr{\rightmark}{\leftmark}%
    {%
        \begin{minipage}[b]{\mymarginwidth}%
            \small\raggedright\normalfont\textsf{\textbf{\color{footersection}Chapter\ \thechapter}}%
        \end{minipage}%
    }%
    {%
        \begin{minipage}[b]{\mymarginwidth}%
            \small\raggedright\normalfont\textsf{\textbf{\color{footersection}Section\ \thesection}}%
        \end{minipage}%
    }%
}
\cohead{\hspace*{-\mymarginwidthwithinnersep}\parbox{\textwidth}{\raggedleft\rightmark}}

\lefoot*%
{%
    \vspace*{1 ex}%
    {\color{stroke1}\rule{\myoutermargin - \mymargininnersep}{0.5 mm}}\\
    \begin{minipage}[b]{\myoutermargin - \mymargininnersep}%
        \raggedleft\normalfont\color{footerpagenr}\textbf{\thepage}%
    \end{minipage}%
}
\rofoot*%
{%
    \vspace*{1 ex}%
    {\color{stroke1}\rule{\myoutermargin - \mymargininnersep}{0.5 mm}}\\
    \begin{minipage}[b]{\myoutermargin - \mymargininnersep}%
        \raggedright\normalfont\color{footerpagenr}\textbf{\thepage}%
    \end{minipage}%
}

\usepackage{caption}
\usepackage{tikz} % Used in order to draw the stylistic elements.

\newlength{\mytmpa}
\newlength{\mytmpb}

\renewcommand*{\partlineswithprefixformat}[3]%
{%
    #2
    \thispagestyle{empty}
    \setlength{\mytmpa}{0.618\mypaperwidth}%
    \setlength{\mytmpb}{0.382\mypaperheight}%
    \ifprintVersion
        \ifprofessionalPrint
            \setlength{\mytmpa}{0.618\mypaperwidth + \mybindingcorrection + \extraborderlength}%
            \setlength{\mytmpb}{0.382\mypaperheight + \extraborderlength}%
        \fi
    \fi
    \begin{tikzpicture}[overlay, remember picture]%
        \node [inner sep = 0, outer sep = 0, anchor = north] at (current page.north west)%
        {%
            \begin{tikzpicture}[overlay, remember picture]%
            \draw[color = stroke1, line width = 0.7 mm] (\mytmpa, 0) -- (\mytmpa, -\mytmpb);%
            \end{tikzpicture}%
        };%
        \node (align) [align = right, below = \mytmpb - 2 ex, inner sep = 0, outer sep = 0, anchor = north west] at (current page.north west)%
        {%
            \hspace{\mytmpa}\hspace{0.5 em}\partname\ \thepart\\[1 ex]
            \color{stroke1}#3%
        };%
    \end{tikzpicture}%
}
\RedeclareSectionCommand%
[%
    font = \normalfont\Huge\sffamily,
    prefixfont = \normalfont\Huge\sffamily,
]
{part}

\usepackage{etoolbox}

\newbool{chapterHasANumber}
\newbool{chapterHasAStar}
\renewcommand*{\chapterlinesformat}[3]%
{%
    \Ifnumbered{#1}{\setbool{chapterHasANumber}{true}}{\setbool{chapterHasANumber}{false}}%
    \Ifstr{#2}{}{\setbool{chapterHasAStar}{true}}{\setbool{chapterHasAStar}{false}}%
    \ifboolexpr{bool{chapterHasANumber} and not bool{chapterHasAStar}}%
    {%
        \begin{tikzpicture}[overlay, remember picture]%
            \node [right = \myinnermargin, below = \mytopmargin, inner sep = 0, outer sep = 0, anchor = north west] (numbernode) at (current page.north west)%
            {%
                \hspace{\myinnermargin}%
                \sffamily\fontsize{40}{40}\selectfont%
                \color{chapternumber}%
                \thechapter%
            };%
            \node [inner sep = 0, outer sep = 0, anchor = north west] at (numbernode.south west)%
            {%
                \begin{tikzpicture}[overlay, remember picture]%
                    \draw[color = stroke1, line width = 0.7 mm] (\myinnermargin, -1 ex) -- (\paperwidth, -1 ex);%
                \end{tikzpicture}%
            };%
            \node (align) [text width = \textwidth - 2 cm, align = right, right = \myinnermargin + \mybodywidth, inner sep = 0, outer sep = 0, anchor = east] at (numbernode.west)%
            {%
                #3%
            };%
        \end{tikzpicture}%
    }%
    {%
        \begin{tikzpicture}[overlay, remember picture]%
            \node [right = \myinnermargin, below = \mytopmargin, inner sep = 0, outer sep = 0, anchor = north west] (numbernode) at (current page.north west)%
            {%
                \hspace{\myinnermargin}%
                \sffamily\fontsize{40}{40}\selectfont%
                \color{white}%
                \thechapter%
            };%
            \node [inner sep = 0, outer sep = 0, anchor = north west] at (numbernode.south west)%
            {%
                \begin{tikzpicture}[overlay, remember picture]%
                    \draw[color = stroke1, line width = 0.7 mm] (\myinnermargin, -1 ex) -- (\paperwidth, -1 ex);%
                \end{tikzpicture}%
            };%
            \node (align) [align = left, right = \myinnermargin, inner sep = 0, outer sep = 0, anchor = south west] at (numbernode.south west)%
            {%
                #3%
            };%
        \end{tikzpicture}%
    }%
}
\RedeclareSectionCommand%
[%
    font = \color{stroke1}\normalfont\fontsize{19}{19}\sffamily,
    afterskip = 20 pt,
]
{chapter}

\BeforeStartingTOC[toc]{\pagestyle{plain}}
\AfterStartingTOC{\thispagestyle{plain}}
\usepackage
[
    sortcites,              % Sort multiple references when citing them together.
    style = alphabetic,     % The style of a citation mark.
    defernumbers,           % Makes sure that references always have unique numbers. This is important if you use multiple bibliographies.
    safeinputenc,           % Allows to use UTF8 characters in the bibliography and tries to translate them into TeX automatically.
    backref = true,         % Creates back references in the bibliography.
    backend=biber,
    backrefstyle = three,   % Compresses three or more consecutive pages in the back references into a range.
    hyperref = true,        % Makes links generated by biblatex clickable. If hyperref is not used, a warning is issued.
    maxbibnames = 99,       % The maximum number of names displayed in the bibliography.
    maxcitenames = 2,       % The maximum number of names displayed when using commands like ›textcite‹. The default is 3. After that, ›et al.‹ is used.
]
{biblatex} % Used in order to format the bibliography.

\renewbibmacro{in:}%
{%
    \ifentrytype{article}{}{\printtext{\bibstring{in}\intitlepunct}}%
}

\renewbibmacro*{volume+number+eid}%
{%
    \printfield{volume}%
    \iffieldundef{number}{}{\addcolon}%
    \printfield{number}%
    \setunit*{\addcomma\space}%
    \printfield{eid}%
}

\DefineBibliographyStrings{english}%
{%
    backrefpage  = {\lowercase{s}ee page}, % For a single page number.
    backrefpages = {\lowercase{s}ee pages} % For multiple page numbers.
}

\DeclareFieldFormat[article]{title}{\textbf{\color{stroke1}#1}}
\DeclareFieldFormat[inproceedings]{title}{\textbf{\color{stroke1}#1}}
\DeclareFieldFormat[thesis]{title}{\textbf{\color{stroke1}#1}}
\DeclareFieldFormat[book]{title}{\textbf{\color{stroke1}#1}}
\DeclareFieldFormat[unpublished]{title}{\textbf{\color{stroke1}#1}}
\DeclareFieldFormat[report]{title}{\textbf{\color{stroke1}#1}}
\DeclareFieldFormat[inbook]{chapter}{\textbf{\color{stroke1}#1}}
\DeclareFieldFormat[inbook]{title}{#1}
\DeclareFieldFormat{pages}{#1}

\newtoggle{authorend}
\togglefalse{authorend}

\DeclareBibliographyDriver{article}%
{%
  \usebibmacro{bibindex}%
  \usebibmacro{begentry}%
  \iftoggle{authorend}{}{\usebibmacro{author/translator+others}}%
  \setunit{\labelnamepunct}\newblock
  \usebibmacro{title}%
  \newunit
  \printlist{language}%
  \newunit\newblock
  \usebibmacro{byauthor}%
  \newunit\newblock
  \usebibmacro{bytranslator+others}%
  \newunit\newblock
  \printfield{version}%
  \newunit\newblock
  \usebibmacro{in:}%
  \usebibmacro{journal+issuetitle}%
  \newunit
  \usebibmacro{byeditor+others}%
  \newunit
  \usebibmacro{note+pages}%
  \newunit\newblock
  \iftoggle{bbx:isbn}
  {\printfield{issn}}
  {}%
  \newunit\newblock
  \usebibmacro{doi+eprint+url}%
  \newunit\newblock
  \usebibmacro{addendum+pubstate}%
  \setunit{\bibpagerefpunct}\newblock
  \usebibmacro{pageref}%
  \newunit\newblock
  \iftoggle{bbx:related}
  {\usebibmacro{related:init}%
    \usebibmacro{related}}
  {}%
  \usebibmacro{finentry}%
  \iftoggle{authorend}{\usebibmacro{author/translator+others}}{}%
}

\DeclareBibliographyDriver{inbook}%
{%
  \usebibmacro{bibindex}%
  \usebibmacro{begentry}%
  \iftoggle{authorend}{}{\usebibmacro{author/translator+others}}%
  \setunit{\labelnamepunct}\newblock
  \usebibmacro{chapter+pages}%
  \newunit
  \printlist{language}%
  \newunit\newblock
  \usebibmacro{byauthor}%
  \newunit\newblock
  \usebibmacro{in:}%
  \usebibmacro{bybookauthor}%
  \newunit\newblock
  \usebibmacro{maintitle+booktitle}%
  \newunit\newblock
  \usebibmacro{byeditor+others}%
  \newunit\newblock
  \printfield{edition}%
  \newunit
  \iffieldundef{maintitle}
  {\printfield{volume}%
    \printfield{part}}
  {}%
  \newunit
  \printfield{volumes}%
  \newunit\newblock
  \usebibmacro{series+number}%
  \newunit\newblock
  \printfield{note}%
  \newunit\newblock
  \usebibmacro{publisher+location+date}%
  \newunit\newblock
  \newunit\newblock
  \iftoggle{bbx:isbn}
  {\printfield{isbn}}
  {}%
  \newunit\newblock
  \usebibmacro{doi+eprint+url}%
  \newunit\newblock
  \usebibmacro{addendum+pubstate}%
  \setunit{\bibpagerefpunct}\newblock
  \usebibmacro{pageref}%
  \newunit\newblock
  \iftoggle{bbx:related}
  {\usebibmacro{related:init}%
    \usebibmacro{related}}
  {}%
  \usebibmacro{finentry}%
  \iftoggle{authorend}{\usebibmacro{author/translator+others}}{}%
}

\DeclareBibliographyDriver{inproceedings}%
{%
  \usebibmacro{bibindex}%
  \usebibmacro{begentry}%
  \iftoggle{authorend}{}{\usebibmacro{author/translator+others}}%
  \setunit{\labelnamepunct}\newblock
  \usebibmacro{title}%
  \newunit
  \printlist{language}%
  \newunit\newblock
  \usebibmacro{byauthor}%
  \newunit\newblock
  \usebibmacro{in:}%
  \usebibmacro{maintitle+booktitle}%
  \newunit\newblock
  \usebibmacro{event+venue+date}%
  \newunit\newblock
  \usebibmacro{byeditor+others}%
  \newunit\newblock
  \iffieldundef{maintitle}
  {\printfield{volume}%
    \printfield{part}}
  {}%
  \newunit
  \printfield{volumes}%
  \newunit\newblock
  \usebibmacro{series+number}%
  \newunit\newblock
  \printfield{note}%
  \newunit\newblock
  \printlist{organization}%
  \newunit
  \usebibmacro{publisher+location+date}%
  \newunit\newblock
  \usebibmacro{chapter+pages}%
  \newunit\newblock
  \iftoggle{bbx:isbn}
  {\printfield{isbn}}
  {}%
  \newunit\newblock
  \usebibmacro{doi+eprint+url}%
  \newunit\newblock
  \usebibmacro{addendum+pubstate}%
  \setunit{\bibpagerefpunct}\newblock
  \usebibmacro{pageref}%
  \newunit\newblock
  \iftoggle{bbx:related}
  {\usebibmacro{related:init}%
    \usebibmacro{related}}
  {}%
  \usebibmacro{finentry}%
  \iftoggle{authorend}{\usebibmacro{author/translator+others}}{}%
}

\DeclareBibliographyDriver{thesis}%
{%
  \usebibmacro{bibindex}%
  \usebibmacro{begentry}%
  \iftoggle{authorend}{}{\usebibmacro{author}}%
  \setunit{\labelnamepunct}\newblock
  \usebibmacro{title}%
  \newunit
  \printlist{language}%
  \newunit\newblock
  \usebibmacro{byauthor}%
  \newunit\newblock
  \printfield{note}%
  \newunit\newblock
  \printfield{type}%
  \newunit
  \usebibmacro{institution+location+date}%
  \newunit\newblock
  \usebibmacro{chapter+pages}%
  \newunit
  \printfield{pagetotal}%
  \newunit\newblock
  \iftoggle{bbx:isbn}
  {\printfield{isbn}}
  {}%
  \newunit\newblock
  \usebibmacro{doi+eprint+url}%
  \newunit\newblock
  \usebibmacro{addendum+pubstate}%
  \setunit{\bibpagerefpunct}\newblock
  \usebibmacro{pageref}%
  \newunit\newblock
  \iftoggle{bbx:related}
  {\usebibmacro{related:init}%
    \usebibmacro{related}}
  {}%
  \usebibmacro{finentry}%
  \iftoggle{authorend}{\usebibmacro{author}}{}%
}

\providebool{bbx:subentry}
\newbibmacro*{citenum}%
{% Note: the original macro was called ›cite‹. I did not redefine ›cite‹ but instead defined a new macro ›citenum‹ because the author-year citations use the ›cite‹ macro too. Using ›\renewbibmacro*{cite}‹ would have caused all the author-year citations to become numeric too.
  \printtext[bibhyperref]{% If you ever want to use hyperref.
    \printfield{prefixnumber}%
    \printfield{labelnumber}%
    \ifbool{bbx:subentry}
    {\printfield{entrysetcount}}
    {}}%
}

\DeclareCiteCommand{\conline}[\mkbibbrackets]
{\usebibmacro{prenote}}
{\usebibmacro{citeindex}%
  \usebibmacro{citenum}}% Note: this was originally "cite" but I changed it to "citenum" to avoid clashes with the author-year style.
{\multicitedelim}
{\usebibmacro{postnote}}
\usepackage
[
    babel = true, % Enables language-specific tuning.
]
{microtype}           % Uses the text space more efficiently.
\usepackage{csquotes} % Uses the correct quotes according to the current language.

\usepackage{amsmath}
\usepackage{amssymb}
\usepackage{amsthm}
\usepackage{thmtools}
\usepackage{mathtools}
\usepackage{thm-restate}
\usepackage{dsfont}        % Yields far better blackboard-bold letters than \mathbb. Use \mathds in order to write such letters.
\usepackage{braceMnSymbol} % Adjusts overbraces and underbraces such that longer versions are put together seamlessly.

\usepackage
[
    ttscale = 0.85, % Scales the typewriter font.
]
{libertine} % The main font used in this thesis.
\usepackage
[
    libertine,    % Changes the math font to libertine (the main font).
    slantedGreek, % Makes all greek letters italic by default. If you want to use an upright greek letter, use ›\up‹ immediately followed by the letter’s name. For example, \upGamma displays an upright uppercase gamma.
    vvarbb,       % Changes the \mathbb font to another font. However, \mathbb remains ugly and should not be used. Use \mathds instead.
    libaltvw,     % Uses different characters for v und w that look far better than the default ones.
]
{newtxmath} % The main math font of this thesis. It fits well with the main font.
\usepackage{url} % Responsible for URL formatting.
\usepackage{bm}  % Allows to use sensible bold letters in math mode. This package has to go after the font packages. Otherwise it does not work correctly!

\usepackage{graphicx} % The standard package for including graphics into your document.
\usepackage
[
    subrefformat = simple, % Formats the label of the \subref command without parentheses.
    labelformat = simple,  % Formats the mark of a subfigure without parentheses.
]
{subcaption}         % Enables it to have subfigures inside of a single figure.
\usepackage{wrapfig} % Allows to put figures next to text.

\columnsep = \mymargininnersep
\usepackage{array}     % Improves the way that tables can be formatted.
\usepackage{booktabs}  % Adds lines (called ›rules‹) that can be used in tables and improves spacing.
\usepackage{longtable} % Allows to make tables that span multiple pages.
\usepackage{pdflscape} % Allows to change a page into landscape. This is handy if a table is very wide.

\usepackage[inline]{enumitem} % Adds tons of useful features to enumeration environments.
\setlist[itemize]{noitemsep}
\setlist[itemize,2]{noitemsep}
\setlist[itemize,3]{noitemsep}
\setlist[enumerate]{noitemsep,label={(\arabic*)},ref={(\arabic*)}}
\setlist[enumerate,2]{noitemsep,label={(\alph*)},ref={(\theenumi.\alph*)}}
\setlist[enumerate,3]{noitemsep,label={FIXME}}

\usepackage
[
    ruled,         % Creates lines at the top and at the bottom. Further, the caption is now above the algorithm.
    vlined,        % Shows the scope of a statement spanning multiple lines via a small vertical bar. Thus, no closing tags are needed.
    linesnumbered, % Shows line numbers.
]
{algorithm2e} % Allows to write pseudocode.

\usepackage{xspace}   % Adds the functionality that a space after a command will be shown as a space in the output.
\usepackage
[
    shortcuts, % Allows to use short symbols for non-breaking hyphens and dashes instead of lengthy commands.
]
{extdash}             % Adds non-breaking hyphens and dashes.
\usepackage{setspace} % Allows to easily chnage the spacing inside of the document.

\usepackage{xparse}    % Is used in order to define reasonable commands.
\usepackage{footnote}  % Allows it to extend the environments footnotes can be used in. It is said that this package is in conflict with ›hyperref‹. I did not note any troubles. However, if something is fishy, it is probably best to not use this package.
\usepackage{afterpage} % Adds the \afterpage command, which specifies that the provided argument shall be processed after the current page is finished.
\usepackage
[
    textsize = scriptsize, % Determines the text size of the TODO note.
    disable,
]
{todonotes}            % Adds TODO notes to the document. These are small text areas inside of the margin of a page.

\usepackage
[
    bookmarks = true,                 % Generates boodmarks for the PDF.
    bookmarksopen = false,            % The bookmarks are closed by default.
    bookmarksnumbered = true,         % The bookmarks use the numbers of the corresponding headline.
    pdfstartpage = 1,                 % The first page seen when opening the PDF.
    pdftitle = {{\printTitle}},       % The PDF’s title in the meta data.
    pdfauthor = {{\printAuthor}},     % The PDF’s author name in the meta data.
    pdfsubject = {{\printSubject}},   % The PDF’s subject in the meta data.
    pdfkeywords = {{\printKeywords}}, % The PDF’s keywords in the meta data.
    breaklinks = true,                % Allows it to break links.
    \ifprintVersion
        hidelinks,                    % In the printed version, links are not highlighted, as they are not clickable.
    \else
    colorlinks = true,            % The text of hyperlinks is colored instead of having a colored box around it.
    allcolors = stroke1,          % Every hyperlink uses the same color. If you want to change specific colors, use the commands below.
    \fi
]
{hyperref} % The standard package that is used for creating hyperlinks inside of a document.

\usepackage
[
        capitalise, % Capitalizes the words in front of the labels. This can also be done by simply using \Cref instead of \cref. In order to have a greater variety, this option is not used.
    noabbrev,   % The words in front of the labels are not abbreviated.
    nameinlink, % Extends the link of a reference to the word in front of it.
]
{cleveref} % This package must be included after ›hyperref‹. It creates clever references that know what they refer to.
\newcommand*{\colloquialDegreeName}{Master}
\newcommand*{\colloquialDegreeNameLowercase}{master}

\newcommand*{\degreeAbbreviation}{M.}

\ifbachelorThesis
    \renewcommand*{\colloquialDegreeName}{Bachelor}
    \renewcommand*{\colloquialDegreeNameLowercase}{bachelor}
    \renewcommand*{\degreeAbbreviation}{B.}
\fi

\makeatletter
    \def\IfEmptyTF#1%
    {%
        \if\relax\detokenize{#1}\relax%
            \expandafter\@firstoftwo%
        \else%
            \expandafter\@secondoftwo%
        \fi%
    }
\makeatother

\NewDocumentCommand{\mathOrText}{m}
{%
    \ensuremath{#1}\xspace%
}

\let\originalleft\left
\let\originalright\right
\renewcommand{\left}{\mathopen{}\mathclose\bgroup\originalleft}
\renewcommand{\right}{\aftergroup\egroup\originalright}

\makeatletter
    \DeclareRobustCommand{\bfseries}%
    {%
        \not@math@alphabet\bfseries\mathbf%
        \fontseries\bfdefault\selectfont%
        \boldmath%
    }
\makeatother

\xspaceaddexceptions{]\}}

\allowdisplaybreaks

\crefname{ineq}{inequality}{inequalities}
\creflabelformat{ineq}{#2{\upshape(#1)}#3} 

\crefname{term}{term}{terms}
\creflabelformat{term}{#2{\upshape(#1)}#3}

\let\oldfootnote\footnote

\newlength{\spaceBeforeFootnote} % Denotes the space before the footnote mark in em.
\newlength{\spaceAfterFootnote}  % Denotes the space after the footnote mark in em.

\RenewDocumentCommand{\footnote}{o o o m}%
{%
    \IfNoValueTF{#1}%
    {%
        \oldfootnote{#4}%
    }%
    {%
        \setlength{\spaceBeforeFootnote}{\IfEmptyTF{#1}{0}{#1} em}%
        \IfNoValueTF{#2}%
        {%
            \hspace*{\spaceBeforeFootnote}\oldfootnote{#4}%
        }%
        {%
            \setlength{\spaceAfterFootnote}{\IfEmptyTF{#2}{0}{#2} em}%
            \hspace*{\spaceBeforeFootnote}\IfNoValueTF{#3}{\oldfootnote{#4}}{\oldfootnote[#3]{#4}}\hspace*{\spaceAfterFootnote}%
        }%
    }%
}

\makesavenoteenv{figure}
\makesavenoteenv{table}
\makesavenoteenv{tabular}

\iffancyTheorems
    \declaretheoremstyle
    [
        spaceabove = \topsep,
        spacebelow = \topsep,
        headfont = \bfseries,
        headformat = \textcolor{stroke1}{$\blacktriangleright$} \NAME~\NUMBER \NOTE,
        notefont = \bfseries,
        notebraces = {(}{)},
        bodyfont = \normalfont,
        postheadspace = 0.5 em,
        qed = \textcolor{stroke1}{\bfseries$\blacktriangleleft$},
    ]
    {myTheoremStyle}
    
    \declaretheorem
    [
        style = myTheoremStyle,
        name = Conjecture,
        numberwithin = chapter,
    ]
    {conjecture}
    \declaretheorem
    [
        style = myTheoremStyle,
        name = Proposition,
        sharenumber = conjecture,
    ]
    {proposition}
    \declaretheorem
    [
        style = myTheoremStyle,
        name = Claim,
        sharenumber = conjecture,
    ]
    {claim}
    \declaretheorem
    [
        style = myTheoremStyle,
        name = Lemma,
        sharenumber = conjecture,
    ]
    {lemma}
    \declaretheorem
    [
        style = myTheoremStyle,
        name = Corollary,
        sharenumber = conjecture,
    ]
    {corollary}
    \declaretheorem
    [
        style = myTheoremStyle,
        name = Theorem,
        sharenumber = conjecture,
    ]
    {theorem}
    \declaretheorem
    [
        style = myTheoremStyle,
        name = Definition,
        sharenumber = conjecture,
    ]
    {definition}
    \declaretheorem
    [
        style = myTheoremStyle,
        name = Example,
        sharenumber = conjecture,
    ]
    {example}
    \declaretheorem
    [
        style = myTheoremStyle,
        name = Remark,
        sharenumber = conjecture,
    ]
    {remark}
    \declaretheorem
    [
        style = myTheoremStyle,
        name = Reduction Rule,
        sharenumber = conjecture,
    ]
    {reductionrule}
    \declaretheorem
    [
        style = myTheoremStyle,
        name = Hypothesis,
        numberwithin = chapter,
    ]
    {hypothesis}
\else
    \theoremstyle{plain}
    
    \newtheorem{conjecture}{Conjecture}[chapter]

    \newtheorem{lemma}[conjecture]{Lemma}
    \newtheorem{corollary}[conjecture]{Corollary}
    \newtheorem{theorem}[conjecture]{Theorem}
    \newtheorem{definition}[conjecture]{Definition}

\fi

\crefname{reductionrule}{Reduction Rule}{Reduction Rules}

\NewDocumentCommand{\functionTemplate}{m m m m o}%
{%
    \IfNoValueTF{#5}%
    {%
        \mathOrText{#1\left#2{#4}\right#3}%
    }%
    {%
        \mathOrText{#1#5#2{#4}#5#3}%
    }%
}

\newcommand*{\leftBracketType}{(}
\newcommand*{\rightBracketType}{)}

\NewDocumentCommand{\createFunction}{m m o o}%
{%
    \renewcommand*{\leftBracketType}{\IfNoValueTF{#3}{(}{#3}}%
    \renewcommand*{\rightBracketType}{\IfNoValueTF{#4}{)}{#4}}%
    \NewDocumentCommand{#1}{o o}%
    {%
        \IfNoValueTF{##1}%
        {%
            \mathOrText{#2}%
        }%
        {%
            \functionTemplate{#2}{\leftBracketType}{\rightBracketType}{##1}[##2]%
        }%
    }%
}

\DeclareDocumentCommand{\probabilisticFunctionTemplate}{m m O{} o}
{%
    \functionTemplate{#1}%
    {\lbrack}%
    {\rbrack}%
    {#2\IfEmptyTF{#3}{}{\ \IfNoValueTF{#4}{\left}{#4}\vert\ \vphantom{#2}#3\IfNoValueTF{#4}{\right.}{}}}%
    [#4]%
}

\ifboldNumberSets
    \newcommand*{\N}{\mathOrText{\mathbf{N}}}

    \newcommand*{\R}{\mathOrText{\mathbf{R}}}
    
    \newcommand*{\indicatorFunctionSymbol}{\mathbf{1}}
\else
    \newcommand*{\N}{\mathOrText{\mathds{N}}}

    \newcommand*{\R}{\mathOrText{\mathds{R}}}
    
    \newcommand*{\indicatorFunctionSymbol}{\mathds{1}}
\fi

\RenewDocumentCommand{\Pr}{m O{} o}%
{%
    \probabilisticFunctionTemplate{\mathrm{Pr}}{#1}[#2][#3]%
}

\NewDocumentCommand{\Var}{m O{} o}%
{%
    \probabilisticFunctionTemplate{\mathrm{Var}}{#1}[#2][#3]%
}

\DeclareDocumentCommand{\bigO}{m o}%
{%
    \functionTemplate{\mathrm{O}}{(}{)}{#1}[#2]%
}

\DeclareDocumentCommand{\smallO}{m o}%
{%
    \functionTemplate{\mathrm{o}}{(}{)}{#1}[#2]%
}

\DeclareDocumentCommand{\bigTheta}{m o}%
{%
    \functionTemplate{\upTheta}{(}{)}{#1}[#2]%
}

\DeclareDocumentCommand{\bigOmega}{m o}%
{%
    \functionTemplate{\upOmega}{(}{)}{#1}[#2]%
}

\DeclareDocumentCommand{\smallOmega}{m o}%
{%
    \functionTemplate{\upomega}{(}{)}{#1}[#2]%
}

\DeclareDocumentCommand{\eulerE}{o}%
{%
    \mathOrText{\mathrm{e}\IfNoValueTF{#1}{}{^{#1}}}%
}

\DeclareDocumentCommand{\poly}{m o}%
{%
    \functionTemplate{\mathrm{poly}}{(}{)}{#1}[#2]%
}

\createFunction{\id}{\mathrm{id}}

\NewDocumentCommand{\ind}{m o o}%
{%
    \IfNoValueTF{#2}%
    {%
        \mathOrText{\indicatorFunctionSymbol_{#1}}%
    }%
    {%
        \functionTemplate{\indicatorFunctionSymbol_{#1}}{(}{)}{#2}[#3]%
    }%
}

\DeclareDocumentCommand{\dom}{m o}%
{%
    \functionTemplate{\mathrm{dom}}{(}{)}{#1}[#2]%
}

\DeclareDocumentCommand{\rng}{m o}%
{%
    \functionTemplate{\mathrm{rng}}{(}{)}{#1}[#2]%
}

\DeclareDocumentCommand{\d}{o}%
{%
    \mathrm{d}\IfNoValueTF{#1}{}{^{#1}}%
}

\DeclareDocumentCommand{\set}{m m o}%
{
    \mathOrText{\IfNoValueTF{#3}{\left}{#3}\{#1\ \IfNoValueTF{#3}{\left}{#3}\vert\
    \vphantom{#1}#2\IfNoValueTF{#3}{\right.}{}\IfNoValueTF{#3}{\right}{#3}\}}
}
\definecolor[named]{lipicsGray}{rgb}{0.31,0.31,0.33}
\usepackage{todonotes}
\usepackage{placeins}
\usepackage[normalem]{ulem}
\usepackage{cryptocode}
\usepackage{bm}
\usepackage{tabularx}
\usepackage{makecell}
\ifdefined\N\else\DeclareMathOperator{\N}{\mathbb{N}}\fi
\ifdefined\R\else\DeclareMathOperator{\R}{\mathbb{R}}\fi

\newcommand{\BigO}{\mathcal{O}}
\newcommand{\iphase}{\delta}
\newcommand{\nodeset}[1]{V(#1)}
\newcommand{\edgeset}[1]{E(#1)}
\newcommand{\adjacentEdges}[1]{N(#1)}
\newcommand{\setsize}[1]{\left|#1\right|}
\newcommand{\nodesetsize}[1]{\setsize{V(#1)}}
\newcommand{\edgesetsize}[1]{\setsize{E(#1)}}
\newcommand{\decc}[1]{\setsize{EC_\iphase(#1)}}
\newcommand{\Tmax}{T_\mathrm{max}}
\DeclareMathOperator{\edgeLabelOp}{\lambda}
\newcommand{\edgeLabel}[1]{\edgeLabelOp(#1)}
\newcommand{\G}{\mathcal{G}}
\newcommand{\defaultGraphName}{\G}
\newcommand{\defaultTemporalGraph}{\G = (V, E, \edgeLabelOp)}
\newcommand{\defaultdecc}{\decc{\defaultGraphName}}
\newcommand{\tw}{\mathrm{tw}}
\newcommand{\funnyconstant}{K}
\newcommand{\EOp}{\mathbb{E}}
\newcommand{\E}[1]{\EOp\left[#1\right]}
\newcommand{\ProbOp}{\mathbb{P}}
\newcommand{\Prob}[1]{\ProbOp\left[#1\right]}
\DeclareMathOperator{\divOp}{div}
\DeclareMathOperator{\modOp}{mod}
\newcommand{\FMainHack}{\mathtt{Follow}}
\newcommand{\FMain}{\FMainHack}
\newcommand{\FDMainHack}{\mathtt{DiscoveryFollow}}
\newcommand{\FDMain}{\FDMainHack}
\newcommand{\exHack}{\mathtt{Explore}}
\newcommand{\ex}{\exHack}

\newif\ifpaper
\definecolor[named]{benGray}{rgb}{0.61,0.61,0.63}

\paperfalse

\begin{document}

    \frontmatter
    % This file contains the layout of the title page. It is a generous interpretation of the demo page provided by the MNF of the University of Potsdam.

% This page uses a different geometry, as the content will be centered (not including the binding correction).
\ifprintVersion
    \ifprofessionalPrint
        \newgeometry
        {
            textwidth = 134 mm,
            textheight = 220 mm,
            top = 38 mm + \extraborderlength,
            inner = 38 mm + \mybindingcorrection + \extraborderlength,
        }
    \else
        \newgeometry
        {
            textwidth = 134 mm,
            textheight = 220 mm,
            top = 38 mm,
            inner = 38 mm + \mybindingcorrection,
        }
    \fi
\else
    \newgeometry
    {
        textwidth = 134 mm,
        textheight = 220 mm,
        top = 38 mm,
        inner = 38 mm,
    }
\fi

% The format of the title page.
\begin{titlepage}
    \sffamily
    \begin{center}
        \includegraphics[height = 3.2 cm]{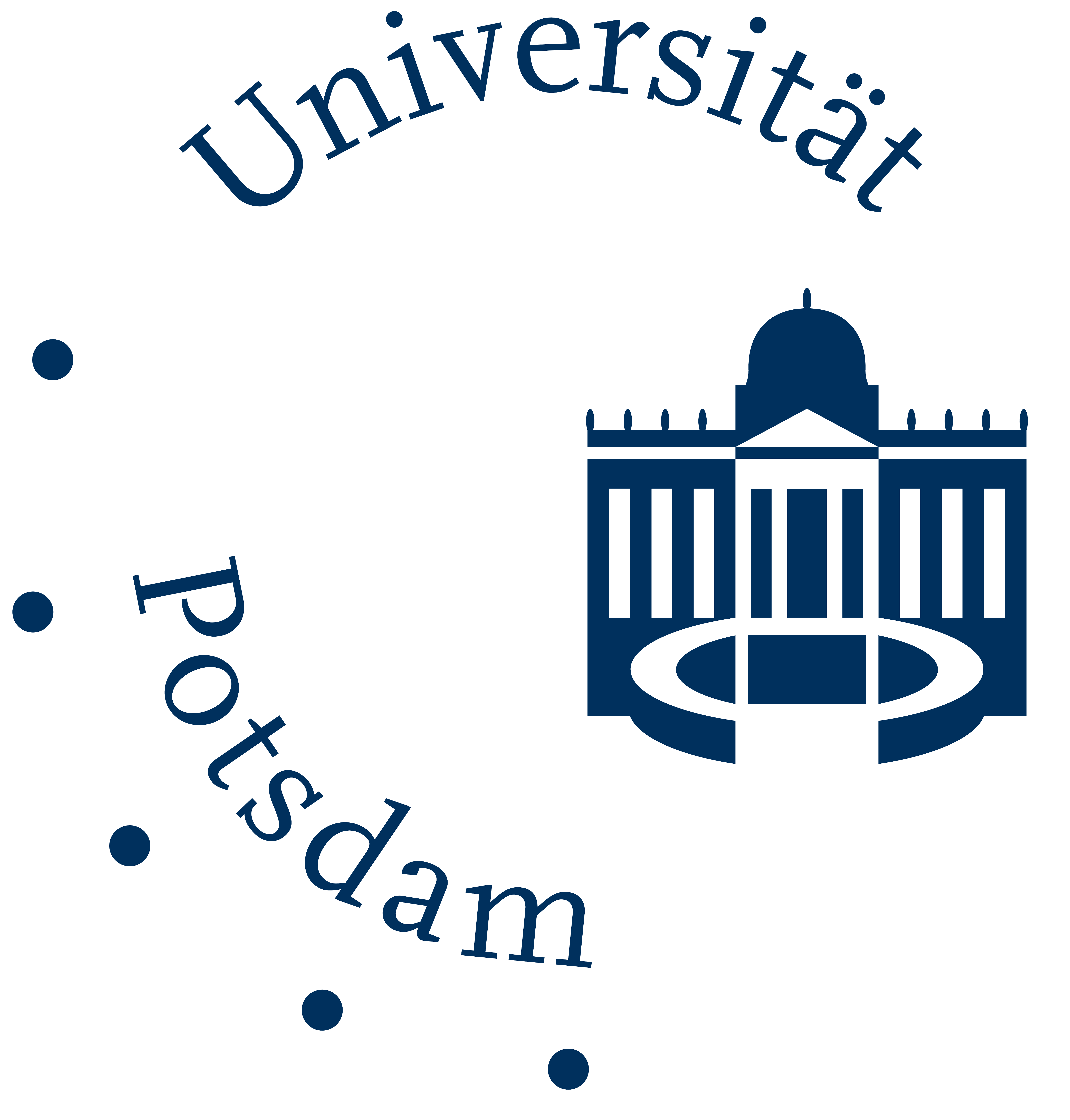} \hfill \includegraphics[height = 3 cm]{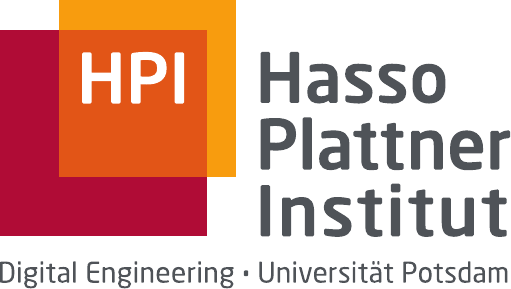}\\
        \vfil
        {\LARGE
            \rule[1 ex]{\textwidth}{1.5 pt}
            \onehalfspacing\printTitleBold\\[1 ex]
            {\vspace*{-1 ex}\Large \printGermanTitle}\\
            \rule[-1 ex]{\textwidth}{1.5 pt}
        }
        \vfil
        {\Large\textbf{\printAuthor}}
        \vfil
        {\large Universitäts\colloquialDegreeNameLowercase arbeit\\[0.25 ex]
        zur Erlangung des akademischen Grades}\\[0.25 ex]
        \bigskip
        {\Large \colloquialDegreeName{} of Science}\\[0.5 ex]
        {\large\emph{(\degreeAbbreviation\,Sc.)}}\\
        \bigskip
        {\large im Studiengang\\[0.25 ex]
        \printProgram}
        \vfil
        {\large eingereicht am \printDateReceived{} am\\[0.25 ex]
        Fachgebiet Algorithm Engineering der\\[0.25 ex]
        Digital-Engineering-Fakultät\\[0.25 ex]
        der Universität Potsdam}
    \end{center}
    
    \vfil
    \begin{table}[h]
        \centering
        \large
        \sffamily 
        {\def\arraystretch{1.2}
            \begin{tabular}{>{\bfseries}p{3.8 cm}p{5.3 cm}}
                Gutachter               & \printNameOfSupervisor\\
                Betreuer*innen                & \printAdditionalExaminers
            \end{tabular}
        }
    \end{table}
\end{titlepage}

\restoregeometry

    \pagestyle{plain}

    \addchap{Abstract}
    % This file should contain the English abstract.

Researchers, policy makers, and engineers need to make sense of data on spreading processes as diverse as viral infections, water contamination, and misinformation in social networks.
Classical questions include predicting infection behavior in a given network or deducing the structure of a network from infection data.
We study two central problems in this area.
In \emph{graph discovery}, we aim to fully reconstruct the structure of a graph from infection data.
In \emph{source detection}, we observe a limited subset of the infections and aim to deduce the source of the infection chain.
These questions have received considerable attention and have been analyzed in many settings (e.g.,~under different models of spreading processes), yet all previous work shares the assumption that the network has the same structure at every point in time.
For example, if we consider how a disease spreads, it is unrealistic to assume that two people can either never or always infect each other, rather such an infection is possible precisely when they meet.
\emph{Temporal graphs}, in which connections change over time, have recently been used as a more realistic graph model to study infections.
Despite this recent attention, we are the first to study graph discovery or source detection in temporal graphs.

We propose models for temporal graph discovery and source detection that are consistent with previous work on static graphs and extend it to embrace the stronger expressiveness of temporal graphs.
For this, we employ the standard susceptible-infected-resistant model of spreading processes, which is particularly often used to study diseases.

For graph discovery, our main contributions include the \(\FDMain \) algorithm, which we prove to be optimal in the cost of discovery per edge.
We also provide algorithms and lower bounds for a number of variations of the problem—such as where the Discoverer does not learn who infected whom but only when each node was infected.
These variations make our results more widely applicable and yield a clearer picture of which aspects make temporal graph discovery easy or difficult.
We complement our theoretical analysis with an experimental evaluation of the \(\FDMain \) algorithm on real-world interaction data from the Stanford Network Analysis Project and on temporal Erdős-Renyi graphs.
On Erdős-Renyi graphs, we discover an interesting threshold behavior, which can be explained by a novel connectivity parameter that captures which edges in the temporal graph can influence each other during the spreading process.

For the source detection problem, we crucially show that, under worst-case analysis, no algorithm can be faster than the trivial brute force algorithm.
This motivates us to propose faster, randomized algorithms for different variations of the problem.
For example, we examine general graphs and trees, as well as sources with dynamic or with consistent behavior.
We also prove matching lower bounds, demonstrating that our algorithms are asymptotically optimal among all algorithms that  succeed with constant probability.
We thus provide a thorough picture of the randomized complexity of the temporal source detection problem in a wealth of settings.

    \selectlanguage{ngerman}
    \addchap{Zusammenfassung}
    % This file should contain the German abstract.

Forschende, Entscheidungsträger*innen und Analyst*innen müssen jeden Tag Daten aus sich verbreitenden Prozessen verstehen, die alles von viralen Infekten über Wasserverunreinigungen bis hin zu Missinformationen in sozialen Netzwerken beschreiben.
Zu den klassischen Aufgaben gehören die Vorhersage von Infektionsverhalten in einem gegebenen Netzwerk und das Erschließen von Netzwerkstrukturen aus Infektionsdaten.
Mit dieser Arbeit betrachten wir zwei zentrale Probleme dieses Forschungsgebiets.
In der \emph{Graphenerkundung} ist das Ziel, die Struktur eines Netzwerks vollständig aus Infektionsdaten zu erschließen.
In der \emph{Quellensuche} möchten wir aus wenig Informationen über das Infektionsverhalten die Quelle einer Infektionskette ermitteln.
Diese Fragen wurden bereits eingehend und in vielen Variationen analysiert (z.~B.~unter verschiedenen Ausbreitungsmodellen).
Allerdings teilt alle aktuelle Forschung die Annahme, dass das zugrundeliegende Netzwerk zu jedem Zeitpunkt die gleiche Struktur aufweist.
% TODO: Maybe include the example here if there is space
In \emph{temporalen Graphen} verändern sich die Verbindungen des Netzwerks mit der Zeit, sodass mit diesem Netzwerkmodell Infektionsverhalten realistischer modelliert werden kann.
Trotz der jüngeren Forschung zu Infektionen in temporalen Graphen, sind wir die Ersten, die Graphenerkundung und Quellensuche auf temporalen Graphen untersuchen.

Wir stellen Modelle für temporale Graphenerkundung und Quellensuche vor, die konsistent mit den aktuellen Modellen auf statischen Graphen sind und das zeitabhängige Verhalten der Netzwerke realistischer beschreiben.
% Dafür nutzen wir das \emph{susceptible-infected-resistant} Ausbreitungsmodell.

Einer unserer wichtigsten Beiträge zur Graphenerkundung ist der \(\FDMain\)-Algorithmus, der beweisbar optimale Kosten pro erkundeter Kante erreicht.
Außerdem stellen wir Algorithmen und untere Schranken für viele Variationen des Problems vor.
So zum Beispiel für den Fall, dass der Erkunder nur lernt, welcher Knoten wann infiziert wird, aber nicht von welchem anderen Knoten.
Diese Variationen machen unsere Ergebnisse breit anwendbar und zeichnen ein klareres Bild davon, welche Aspekte Graphenerkundung einfach oder schwer machen.
Wir komplementieren unsere theoretische Analyse mit Experimenten mit dem \(\FDMain\)-Algorithmus auf Echtwelt-Interaktionsdaten aus dem Stanford Network Analysis Project und auf temporalen Erdős-Renyi-Graphen.
Auf Erdős-Renyi-Graphen beobachten wir ein interessantes Schwellenverhalten, dass wir durch einen neuartigen Zusammenhangsparameter erklären können, der einfängt, welche Kanten sich während des Infektionsverhaltens beeinflussen können.

Zur Quellenerkundung zeigen wir, dass im Worst-Case kein Algorithmus den trivialen Brute-Force-Algorithmus schlagen kann.
Davon motiviert stellen wir schnellere, randomisierte Algorithmen für mehrere Versionen des Problems vor.
So lösen unsere Algorithmen z.~B.~das Problem auf allgemeinen Graphen und Bäumen sowie mit konsistentem oder dynamischem Quellverhalten.
Außerdem beweisen wir untere Schranken, die zeigen, dass unsere Algorithmen asymptotisch optimal sind unter allen Algorithmen, die das Problem mit konstanter Erfolgswahrscheinlichkeit lösen.
Somit zeichnen wir ein detailliertes Bild der randomisierten Komplexität der temporalen Quellenerkundung.

    \selectlanguage{american}

    \addchap{Acknowledgments}
    % Here you can write whom you want to thank.

Above all, I owe a mountain of gratitude to my supervisors, who enabled me to move this personal mountain.
Without you,  Michelle, George, and Nico, this thesis would not exist.
Michelle has been the most Argus-eyed proofreader anyone could wish for, and her notes on my illustrations were invaluable for making the technical proofs understandable.
George prevented me from panicking too many times to count and made me write a coherent story instead of the mess that filled my head sometimes.
Nico did not only fulfill his destiny as the ultimate proof-breaker, but also gave me valuable advice on how to structure my arguments.
So, if you like the proofs in this thesis, think of Nico.

Jonas, Niko, Busty, Armin, and Paula were the best office mates.
You all ensured that I enjoyed coming to the office every day.
I will keep all the games we played, yarns we spun, and tea we spilled in fond memory.
And of course, you were great targets for nerd sniping whenever I needed someone to rubber duck this thesis with.

Sonia's advice on what to write and do was almost as valuable as her advice on what not to do and when not to write. You kept me sane in the last weeks of this thesis.

    \setuptoc{toc}{totoc}
    \tableofcontents

    \pagestyle{headings}
	\mainmatter

\chapter{Introduction}
\label{sec:orgd9802be}
Information, diseases, pollutants—all these things spread through the networks that connect everything from people and servers to our streets and sewers.
Despite their differences, all of these networks can be modeled in a similar fashion, and we can thus utilize a common algorithmic toolkit for their analysis, allowing us to answer many questions impacting our everyday lives.
The study of spreading processes in networks has received substantial attention from researchers as diverse as physicians, sociologists, and theoretical computer scientists \cite{independent-cascade,fake-news-survey,wang_predicting_2016}.
Most famously, the \emph{influence maximization problem}, introduced by
Kempe et al.~\cite{independent-cascade},
asks which node in a network should be infected to maximize the number of nodes infected by the spreading process.
Other problems include finding a sensor placement to detect outbreaks as quickly as possible \cite{sensor-placement}, and finding nodes or edges particularly important to the spread—so-called super-spreaders \cite{Enright2018DeletingET,venkataraman2005new,maji2021identifying}.

All of these problem formulations assume that the structure of the underlying network is already known.
This, of course, is not true in many real-world scenarios, and thus we require algorithms for discovering the structure of these networks to overcome this hurdle.
This could mean discovering the whole graph (which we call \emph{network discovery} or \emph{graph discovery}) or only determining particularly important structural properties, such as finding the source of the spreading process (which we call \emph{source detection}) \cite{sir-source-detection,shah2011rumors,sensor-placement,paluch_fast_2018}.
While these problems are fundamental to data mining, they can be approached from different angles \cite{estimate-diffusion-networks,park2016information}.
One natural approach is to discover the underlying network from the infection data itself (e.g.,~from a small set of disease infection times or from a number of sensors in a sewer network).
This idea has been extensively studied \cite{infer-from-cascades,lokhov2016reconstructing,daneshmand2014estimating,netrapalli2012learning}.
Chistikov et al.~\cite{chistikov2024learning},
similar to us, consider a model where the discovery algorithm may intervene in the spreading process.

Beyond their application as the foundation for spreading analysis tasks, graph discovery and source detection are interesting and relevant problems in their own right.
After we have discovered a network from infection data, we are free to abstract away from the spreading process itself and use the resulting network in a host of different ways.
This is especially relevant for the study of social networks, both real-world and online, where infection data can reveal underlying structures that are otherwise difficult to observe.

The current mainstream models of spreading processes, such as the independent cascade model proposed by
Kempe et al.~\cite{independent-cascade},
and the susceptible-infected-resistant (SIR) model \cite{Hethcote1989,a-contribution,shah2011rumors}, incorporate the inherent time-dependent behavior of the spreading process.
That is, infections take time, meaning a node first gets infected and only then is able to pass on the infection to other nodes.
Yet, they make the same simplifying assumption: they model the underlying network as \emph{static}.
That is, it does not change over time, which means these models fail to capture temporal dependencies in the underlying network itself, even though these changes heavily impact the behavior of spreading processes.
In most applications, this is not a realistic assumption as many real-world networks change over time (e.g.,~social networks in which diseases or information spread are highly temporal as the link between two people is only able to transmit diseases or information at specific points in time).

Over the last twenty years, \emph{temporal graphs} have provided a robust and vibrant graph theoretic framework in which many temporal analogs to problems previously studied on static graphs have been explored.
Temporal graphs are a model of dynamic networks where the edges exist only at specific time steps.
This model has received considerable attention from theoretical computer scientists for both foundational problems like computing shortest paths \cite{michail2016introduction,wu2014path,Danda} or flow \cite{akrida2019temporal}, and for a growing number of applications, for example, in logistics, transportation, mobile sensor networks, and neural networks \cite{casteigts_et_al:DagRep.11.3.16}.
For our purposes, a temporal graph \(\G=(V, E, \edgeLabelOp)\) with lifetime \(\Tmax \in \N\) is a static graph \(G=(V, E)\) together with a function \(\edgeLabelOp\colon E \to 2^{[\Tmax]}\) which indicates that an edge \(e \in E\) exists precisely at the time steps \(\edgeLabel{e}\).

In 2015,
Gayraud et al.~\cite{gayraud-evolving-social-networks}
initiated the study of infections in temporal graphs by generalizing the independent cascade model from static to temporal graphs.
Since then,
Deligkas et al.~\cite{influencers}
have studied which sources maximize infections given a known temporal graph in the susceptible-infected-susceptible (SIS) model.
In this model, each node in the network starts off as susceptible, becomes infected if one of its neighbors is in the infected state, and then stays in the infected state for a fixed period of time, after which the node turns susceptible again.
This is similar to our approach, where we study the SIR model  (which is a standard biological spreading model \cite{Hethcote1989}).
In this model, a node turns resistant after an infection and cannot be infected again.

Even though  graph discovery is a central task in the study of spreading processes, there is no prior work on temporal graphs.
Similarly, while the important problem of source detection in an unknown non-temporal graph has received considerable attention \cite{sir-source-detection,paluch_fast_2018}, its temporal analog has not yet been examined.

\section{Our Contribution}
\label{sec:org0dfed5b}
We provide the first rigorous definitions of graph discovery and source detection in temporal graphs, and give both algorithms and lower bounds for these problems in a variety of settings.
In \Cref{sec:prelims}, we introduce our preliminaries, including temporal graphs and the SIR infection model.
After this, the thesis is split into two parts.
The first deals with graph discovery, and the second with source detection.

In \Cref{part:graph-discovery}, we give both the first theoretical and the first empirical analysis of graph discovery on temporal graphs.
 \ifpaper In this work \fi
In \Cref{sec:game}, we  \ifpaper then \fi formally define temporal network discovery as a round-based, interactive, two-player game.
In each round of this game, the \emph{Discoverer} can cause infections and observe the resulting infection chains.
It aims to find the time labels of all edges in as few rounds as possible.
The \emph{Adversary} gets to pick the shape and temporal properties of the graph, with the aim of forcing the Discoverer to take as long as possible to accomplish their task.

In \Cref{sec:ideal-patient-zero}, we study a related problem, which will yield a valuable subroutine to solve the graph discovery problem later.
Concretely, in the \emph{ideal patient zero problem}, the Discoverer is asked to find a node \(v \in V\) and time step \(t \in [\Tmax]\) such that an infection at node \(v\) and time step \(t\) causes the whole graph to be infected, or to decide that no such pair \((v,t)\) exists.
We give the \(\FMain\) algorithm, which solves this problem in \(6\edgesetsize{\G} + \lceil \Tmax / \iphase \rceil\) rounds of our game, where \(\iphase\) is the number of time steps a node stays infectious before recovering.
The analysis of this algorithm naturally leads to the definition of \(\iphase\)-edge connected components, denoted by \(EC_\iphase(\G)\).
Intuitively, this is a grouping of the edges such that only edges from the same component may influence each other during infection chains.
The number and size of these components are a new edge-based connectivity parameter that captures the inherently temporal properties of the graph.
We believe this parameter to be of independent interest, and our experiments (in \Cref{sec:experiments}) on Erdős-Renyi graphs suggest a natural conjecture about how the parameters of the Erdős-Renyi model relate to the size of  the \(\iphase\)-edge connected components.

In \Cref{sec:fd-algorithm}, we extend our \(\FMain\) algorithm to the \(\FDMain\) algorithm, which solves the graph discovery problem in \(6\edgesetsize{\G} + \decc{\G} \left\lceil \Tmax/\iphase \right \rceil\) rounds.
We prove this is asymptotically tight in the number of edges in \Cref{sec:witnesses}.
Formally, we prove that there is an infinite family of graphs such that any algorithm winning the graph discovery game requires at least \(\Omega(\edgesetsize{\G})\) rounds.
Crucially, this cannot be improved even if the Discoverer is allowed to start multiple infection chains per round.
We also prove that there is an infinite family of graphs such that the minimum number of rounds required to win the graph discovery game grows in \(\Omega(n T_\mathrm{max} / (\iphase k))\), where \(k\) is the number of infection chains the Discoverer is allowed to start.
This shows that the algorithm is almost tight in the second summand.
These provable lower bounds are important as they show us the smallest number of rounds we can hope for, which allows us to closely analyze where our methods are open to further improvement and where they are asymptotically optimal.

\newcommand{\graphDiscoveryResultsTable}{%
\begin{table}[t]
\caption{\label{tbl:variations-overview}Overview of upper and lower bounds on the number of rounds for different variations of the graph discovery game. Let \(n = \nodesetsize{\G}\) and \(m = \edgesetsize{\G}\). A subscript to a Landau symbol indicates variables that the asymptotic growth is independent of.}
\centering
\begin{tabular}{llll}
\toprule
\textbf{Model} & \textbf{Lower Bound} & \textbf{Upper Bound}\\[0pt]
\midrule
Basic model & \(\Omega_k(m)\) & \(\BigO_k(m + \defaultdecc{} \Tmax / \iphase)\)\\[0pt]
Infection times only & \(\Omega_k(m)\) & \(\BigO_k(m + \defaultdecc{} \Tmax / \iphase)\)\\[0pt]
Unknown static graph & \(\Omega_{m, \decc{\G}}(n \Tmax / (\iphase k))\) & \(\BigO(n \Tmax )\)\\[0pt]
Multilabels & \(\Omega_{\decc{\G}}(\min\{n,m\} \Tmax / (\iphase k))\) & \(\BigO(n \Tmax )\)\\[0pt]
Multiedges & \(\Omega(n \Tmax / (\iphase k))\) & \(\BigO_k(m + \defaultdecc{} \Tmax / \iphase)\)\\[0pt]
\bottomrule
\end{tabular}
\end{table}
}
 \ifpaper \else {\graphDiscoveryResultsTable}\fi

We finish off our theoretical analysis in \Cref{sec:extending}, where we explore a number of variations of the graph discovery problem.
First, we analyze the case where the feedback the Discoverer receives about the infection chains is reduced to infection times.
That is, the Discoverer learns if and when a node was infected, but not by which other node.
Surprisingly, we are able to show that our \(\FDMain\) algorithm directly translates to this scenario, where we have significantly less information about the infections.
Second, we discuss what happens if the Discoverer has no information about the static graph in which the infections are taking place.
Third, we allow the temporal graph to now contain multiedges or more than one label per edge, and analyze the results.
 \ifpaper \else {See \cref{tbl:variations-overview} for an overview over our bounds on these variations.}\fi

Finally in \Cref{sec:experiments}, we empirically validate our theoretical results.
Using both synthetic and real-world data, we execute the \(\FDMain\) algorithm and observe its performance.
We utilize the natural temporal extension of Erdős-Renyi graphs \cite{casteigts_threshold} as well as the \texttt{comm-f2f-Resistance} data set from the Stanford Large Network Dataset Collection \cite{kumar2021deception}, a social network of face-to-face interactions.
Beyond the pure number of rounds, we closely analyze which factors affect the performance of the algorithm.
In particular, we see that the density of the graph affects the performance since, in dense graphs, it needs to spend less time finding new \(\iphase\)-edge connected components.
On Erdős-Renyi graphs, we provide evidence that this effect is mediated by the number of \(\iphase\)-edge connected components, which exhibits a threshold behavior in \(\Tmax/(\nodesetsize{\G} p)\), where \(p\) is the Erdős-Renyi density parameter.
This prompts us to give a conjecture on this threshold behavior, which closely mirrors the famous threshold behavior in the connected components of nodes in static Erdős-Renyi graphs \cite{erdos1960evolution}.

 \ifpaper In summary, we provide the first theoretical and empirical analysis of temporal network discovery, giving algorithms, lower bounds, and experimental evaluation. \fi

In \Cref{part:source-detection}, we turn our attention to source detection.
In \Cref{sec:game-def}, we formalize the source detection problem as a round-based game in which two players interact with each other.
The \emph{Adversary} decides the structure of the underlying temporal network, as well as which node is the source and when it starts an infection chain.
The \emph{Discoverer} initially has no information about the location, number, or label of the edges.
In each round, the Discoverer may watch a single node and is informed if, when, and by which neighbor this watched node is infected.
Note that the Adversary does not know which nodes the Discoverer watches.
Counting the rounds of this game until the Discoverer has found the source is an insufficient cost metric, as then the Adversary is incentivized to minimize infections in order to evade detection.
Instead, we study the number of infections until detection, which we call the \emph{price of detection}.
This cost measure combines an incentive to avoid detection (to prolong the period in which the Adversary may perform infections) and an incentive to infect as many nodes as possible in a single round.
Notice, that if \(n\) is the number of nodes in the graph, any Discoverer algorithm must always tolerate \(\Omega(n^2)\) infections in the worst case, and there is an algorithm that always wins the game within \(n^2\) infections (see \Cref{thm:dY1}).

Because in this game, the worst-case performance of any algorithm is trivially bad (i.e.,~its price of detection is in \(\Omega(n^2)\)), we explore the power of randomized techniques to overcome this limitation in \Cref{sec:rand}.
In particular, we give a Discoverer algorithm that wins the source detection game with constant probability (i.e.,~a probability that does not decrease for larger networks) while only tolerating \(\BigO(n \sqrt n)\) infections.
We also prove that this price of detection is asymptotically optimal for any algorithm that wins the game with a constant probability.

In \Cref{sec:known}, we explore how the difficulty of the problem changes if the Discoverer has knowledge about the underlying static graph (but not about when an edge exists in the temporal graph).
Surprisingly, we prove that this does not lead to a decreased price of detection.
To be precise, we prove that if there is an algorithm winning the game for all known static graphs with some constant probability within \(f(n)\) infections, then there is an algorithm winning the game for unknown static graphs with the same probability and price of detection.
Thus, the lower bound from the game on unknown static graphs transfers.
As the algorithm also transfers, we have an asymptotically tight \(\Theta(n \sqrt n)\) bound for the price of detection for algorithms winning the game with constant probability.
Yet, not all is lost.
We give an algorithm that, for graphs with bounded treewidth \(\tw\), wins the game on known static graphs with constant probability with a price of detection of \(\BigO(\tw \cdot n \log n)\).
This directly translates to a \(\BigO(n \log n)\) algorithm on trees.

\begin{table}[t]
\caption{\label{tbl:results}An overview of the price of detection in different settings. For lower bounds, a \emph{wcp} annotation signifies that the bound holds for all Discoverer algorithms winning the game with constant probability. For upper bounds, \emph{wcp} signifies that a Discoverer algorithm winning the game with constant probability and the noted price of detection exists. Upper bounds marked \emph{det} are achieved via a deterministic algorithm. Results marked \textsuperscript{a} are transferred between known and unknown settings (i.e.,~lower bounds from known to unknown and upper bounds in the other direction). Similarly, results transferred between the consistent and obliviously dynamic settings are marked \textsuperscript{b} and transfers between trees and general graphs are marked \textsuperscript{c}. The result marked * only holds when the Discoverer is allowed to watch two nodes.}
\centering
\begin{tabularx}{\linewidth}{Xllll}
\toprule
 & \multicolumn{2}{c}{\textbf{Trees}}  & \multicolumn{2}{c}{\textbf{General}} \\[0pt]
\cmidrule(r){2-3} \cmidrule(l){4-5}
 & Lower Bound & Upper Bound & Lower Bound & Upper Bound\\[0pt]
\midrule
\textbf{Consistent} &  &  &  & \\[0pt]
Known & \(\Omega(n \log n)\) wcp & \(\BigO(n \log n)\) wcp & \(\Omega(n \sqrt {n})\) wcp &  \textcolor{benGray}{\(\BigO(n \sqrt{n})\) wcp\textsuperscript{a}}\\[0pt]
 & \Cref{thm:aX1-lb-tree} & \Cref{cor:aX1-tree} &  \textcolor{benGray}{\Cref{thm:consistent-known-to-unknown}} & \\[0pt]
Unknown & \(\Omega(n \sqrt n)\) wcp &  \textcolor{benGray}{\(\BigO(n \sqrt n)\) wcp\textsuperscript{c}} &  \textcolor{benGray}{\(\Omega(n \sqrt n)\) wcp\textsuperscript{c}} & \(\BigO(n \sqrt {n})\) wcp\\[0pt]
 & \Cref{thm:aY1-lb} &  &  & \Cref{thm:aY1}\\[0pt]
\midrule
\textbf{Obliviously dynamic} &  &  &  & \\[0pt]
Known &  \textcolor{benGray}{\(\Omega(n \log n)\) wcp\textsuperscript{b}} & \(\BigO(n \log n)\) wcp* & \(\Omega(n^2)\) wcp &  \textcolor{benGray}{\(n^2\) det\textsuperscript{a}}\\[0pt]
 &  & \Cref{thm:dX2} & \Cref{thm:dX1-lb} & \\[0pt]
Unknown &  &  &  \textcolor{benGray}{\(\Omega(n^2)\) wcp\textsuperscript{a}} & \(n^2\) det\\[0pt]
 &  &  &  & \Cref{thm:dY1}\\[0pt]
\bottomrule
\end{tabularx}
\end{table}

The results up until now assume that the behavior of the source is consistent, that is, the source begins its infection chain at the same point in time in each round.
While this is not true in all real-world settings, it is equivalent to taking multiple measurements in one iteration of a process.
In \Cref{sec:dyn} we study what happens when we lift this restriction and thus give the Adversary stronger capabilities.
Note that we still assume the Adversary does not know which nodes the Discoverer chooses to watch.
That is, while the source behavior may be dynamic, it may not depend on the actions of the Discoverer, thus, we call this model \emph{obliviously dynamic}.
Tragically, we are now able to show that any Discoverer algorithm that wins the game with constant probability must have a price of detection in \(\Omega(n^2)\).
Thus, not even randomization may save us from what is asymptotically the worst case.
On a more positive note, we show that if we strengthen the Discoverer slightly by allowing them to watch two nodes each round (instead of just one), we are still able to achieve a price of detection of \(O(n \log n)\) with constant probability on trees (if the static graph is known).
We finish off this line of inquiry by giving a proof that allowing the Discoverer to watch \(k\) nodes may only decrease the price of detection by at most a factor of \(k\) if the source behavior is consistent, thus this generalization may only be asymptotically different in the obliviously dynamic setting.

See \Cref{tbl:results} for an overview of our results  \ifpaper \else {for the source detection problem}\fi.
Note that most of our results are tight.
That is, for most settings, we are able to provide an asymptotically optimal Discoverer algorithm that wins the game with constant probability.

After this thorough study of graph discovery and source detection, we finish this thesis by giving a conclusion in \Cref{sec:thesis-conclusion}.
We point out further research directions in the promising study of spreading processes in temporal networks, which we hope will result in both exciting theoretical insights and applications helping us deal with everything from viral infections over water pollution to misinformation in social networks.

\FloatBarrier
\section{Related Work}
\label{sec:orgb70d616}
The SIR model and related models—such as the SIS model—are standard in mathematical biology \cite{Hethcote1989} and have been extensively studied  from a statistical perspective \cite{BRITTON201024}.
The SIR most closely models how viral infections spread through populations.
For example, adaptations of the model have been used to study COVID-19 \cite{KUDRYASHOV2021466,COOPER2020110057,covid-sir}.
The model can be used both with explicit reference to a graph on which the infections take place and as a stochastic model that abstracts over the individual nodes and their links and instead aims to capture the spreading behavior in a small number of parameters \cite{BRITTON201024,montagnon_stochastic_2019,JI20145067,XUE2017434,wang_predicting_2016}.

More recently, researchers have started to ask more algorithmic questions about infection models, aiming to answer questions such as source detection in addition to the classical task of predicting spread.
In particular, with their seminal paper on influence maximization in social networks,
Kempe et al.~\cite{independent-cascade}
set off this new wave of algorithmic study of infections.
They introduce the \emph{influence maximization problem} (though under the independent cascade model instead of SIR), which already incorporates the inherently temporal nature of spreading infections while the underlying graph is modeled as static.
Their work finds broad resonance with applications as diverse as the study of misinformation \cite{budak-limiting,fake-news-survey}, infectious diseases \cite{covid-sir,computer-immunology}, and viral marketing \cite{marketing-im-billion,10.1145/2896377.2901462}.
Later work considers similar algorithmic questions, including both graph discovery and source detection, under this and other models of infection.

Graphs may be discovered from different sources of information.
In the context of spreading processes, the most natural choice is to discover the graph from infection data on that graph.
This approach has received considerable attention
 \cite{infer-from-cascades,lokhov2016reconstructing,daneshmand2014estimating,netrapalli2012learning}.
 In their work,
Amin et al.~\cite{Amin2014LearningFC}
study graph discovery in a framework where, similar to our graph discovery model, the party wishing to discover the network may choose a set of initial nodes to be infected.
More recently,
Chistikov et al.~\cite{chistikov2024learning}
also consider a model where the observer may intervene in the spreading process.

The source detection problem has a rich literature and has even been studied on the SIR model we employ.
Concretely, in 2011, Shah and Zaman \cite{shah2011rumors} study the source detection problem under the SIR model.
In 2016, Zhu and Ying \cite{sir-source-detection} extend their work.
Other notable algorithmic contributions include the PVTA algorithm proposed by
Pinto et al.~\cite{pinto2012locating},
which focuses on estimating the source of an infection from observations on a sparse set of nodes.
More recently, the independent cascade model has been extended to source detection \cite{pmlr-v216-berenbrink23a}.

Since their introduction by
Kempe et al.~\cite{kempe2000connectivity},
temporal graphs have become one of the standard models to capture graphs where the set of accessible edges changes over time.
Other related approaches include dynamic graphs \cite{ferone2017shortest,dynamicNetworksSurvey}, or settings where graphs are only provided as a stream of edges \cite{graph-streams}.
Casteigts et al.~\cite{casteigts2012time}
provide a comprehensive hierarchical classification over this category of models, which they call \emph{time-varying graphs}.
Temporal graphs in particular have received considerable attention and have been a rich research field for fundamental algorithmics.
Of particular interest are questions related to connectivity in temporal graphs, such as path finding \cite{Danda,wu2014path}, separators \cite{Zschoche2017TheCO}, or connected subgraphs \cite{angrick_et_al:LIPIcs.ESA.2024.11,axiotis_size_approx_spanner,casteigts_fireworks_jour}.
In 2022,
Casteigts et al.~\cite{casteigts2012time}
gave sharp thresholds for the connectivity of random temporal graphs that were generated by a generalization of Erdős-Renyi graphs.
Temporal graphs have also received interest in more applied settings, such as delay-robust routing \cite{Fchsle2022TemporalCC} or routing with waiting-time constraints \cite{Casteigts2021FindingTP}.

The problems of node and edge exploration ask us to find a walk traversing all nodes or edges, respectively.
These problems have been studied on temporal graphs \cite{Erlebach2015OnTG,Akrida2018TheTE,Bumpus2021EdgeEO} and are loosely related to the spreading processes we study.
The two central differences are that (a) infection chains form trees instead of walks and that (b) the infection chains are fully determined by their seed nodes and times, while in the exploration problems, the explorer may deliberately choose where to go next.

Unfortunately, until now, there is a limited amount of research on spreading processes on temporal graphs.
Gayraud et al.~\cite{gayraud-evolving-social-networks}
first extend the study of influence maximization on both the independent cascade and linear threshold models to temporal graphs.
Enright et al.~\cite{Enright2018DeletingET}
study the problem of finding edges or edge labels to delete in order to minimize the spread of infections in a known temporal graph.
Similarly, new work on firefighting in temporal graphs asks which nodes to protect from a spreading process in order to minimize its spread \cite{Gupta2021SpreadAD,DBLP:conf/fun/HandEM22}.
Most recently,
Deligkas et al.~\cite{influencers},
study a number of variations of the influence maximization problem on temporal graphs under the SIS model (which is closely related to SIR).

No previous work studies graph discovery or source detection in temporal graphs.

\chapter{Preliminaries \label{sec:prelims}}
\label{sec:org2d26cc2}
For \(n \in \N^+\) and \(x \in \N\), define \([x, n] \coloneqq \{x, \dots, n\}\) and \([n] \coloneqq [1,n]\).
Write \(2\N\) for the set of even natural numbers and \(2\N + 1\) for the set of odd natural numbers.
For \(a \in \N, m \in \N^+\), set \(a \divOp m \coloneqq \lfloor a / m \rfloor\).
Note that \(a = (a \divOp m)m + a \modOp m\) and thus (for a fixed \(m\)) every integer has precisely one unique representation by \(\divOp\) and \(\modOp\).

Define a \emph{temporal graph} \(\G = (V, E, \edgeLabelOp)\) with lifetime \(\Tmax\) as an undirected graph \((V, E)\) together with a \emph{labeling function} \(\edgeLabelOp \colon E \to 2^{[\Tmax]}\).
We interpret this as an edge \(e \in E\) being present precisely at the time steps \(\edgeLabel{e}\).
We restrict ourselves to \emph{simple} temporal graphs, where each edge has exactly one label.
Abusing  notation, we therefore also use \(\edgeLabelOp\) as if it were defined \(E \to [\Tmax]\).
In \Cref{sec:multigraphs}, we investigate the graph discovery problem if we allow an edge to be present at more than one time step and see that many of our results translate.
We also write \(\nodeset{\G}\) for the nodes of \(\G\) and \(\edgeset{\G}\) for its edges.
Every temporal graph has an \emph{(underlying) static graph} \(G=(V, E)\) obtained by simply forgetting the edge labels.
A sequence of nodes \(v_1, \dots, v_\ell \in V\) is a \emph{temporal path} iff it is a path in the underlying static graphs and the labels along the edges are increasing (i.e.,~for all \(i \in [\ell -2]\), we have \(\edgeLabel{v_i v_{i+1}} < \edgeLabel{v_{i+1} v_{i+2}}\)).
For the rest of this section, let \(\G=(V, E, \edgeLabelOp{})\) be a temporal graph with lifetime \(\Tmax\).

Our model of temporal infection behavior is based on \cite{influencers} (and more historically flows from \cite{a-contribution} and \cite{pastor2015epidemic}).
We use the simple susceptible-infected-resistant (SIR) model, where all nodes start in a \emph{susceptible} state, might become \emph{infected}, remain infectious (i.e.,~attempt to infect their neighbors) for a fixed period of time \(\iphase\), and then become \emph{resistant} (i.e.,~cannot be infected for the rest of the lifetime).
Concretely, the \emph{infection chain} unfolds as follows.
At most \(k \in \N\) nodes may be infected by the Discoverer instead of their neighbors.
We call these \emph{seed infections} and they may occur at arbitrary points in time.
Formally, we denote the set of seed infections as \(S \subseteq V \times [0,\Tmax]\).
Otherwise, a node \(u\) becomes \emph{infected at time step \(t\)} if and only if it is susceptible, some node \(v\) is infectious at time step \(t\), and there is an edge \(uv\) with label \(t\).
Then, \(u\) itself is infectious for the time steps \(t+1\) until \(t+\iphase{}\), after which \(u\) becomes resistant.
Note that if a susceptible node has two or more infected neighbors at the same time,  it may be infected by an arbitrary one, but only one of them.
Thus, for a given set of seed infections, there can be multiple possible infection chains.

The \emph{infection log} of such a chain now consists of the information showing which node becomes infected by which of its neighbors at which time.
Formally, the infection log is a set \(L \subseteq V^2 \times [0,\Tmax]\) where \((u,v,t) \in L\)  means that \(u\) infected \(v\) at time step \(t\). We denote a seed infection at node \(u\) at time step \(t\) by  \((u,u,t) \in L\).
Similarly, call the \emph{infection timetable} the list \(T \subseteq V \times [0,\Tmax]\) obtained by only including information on when a node became infected but not the neighbor who caused the infection.
For graph discovery, we will later discuss a variation where only this infection timetable available to the Discoverer, see \Cref{sec:infection-time-feedback}.
We call an infection log \(L\)  \emph{consistent} with a given set of seed infections \(S\) if there is an instance of the above-described infection process seeded with \(S\) and producing \(L\).
Let consistency for infection timetables be defined similarly.

While, for the same set of seeds, there can be multiple consistent infection logs, there is exactly one consistent infection timetable.

\begin{lemma}
Let \(S \subseteq V \times [0,\Tmax]\), be a set of seed infections, and \(L_1\) and \(L_2\) be two infection logs consistent with \(S\). Then the induced infection timetables \(T_i \coloneqq \{(u,t) \mid (u,v,t) \in L_i\}\) (for \(i=1,2\))  are the same, that is, \(T_1 = T_2\).
\label{lem:infection-log-to-time-table}
\end{lemma}

\begin{proof}
It is easy to see that, inductively, in each time step, the state of all nodes must be identical under \(L_1\) and \(L_2\).
In particular, for each node, the time step in which it first becomes infected (if it ever gets infected) is the same, and thus \(T_1 = T_2\).
\end{proof}

\part{Graph Discovery \label{part:graph-discovery}}
\chapter{Modeling Temporal Graph Discovery \label{sec:game}}
\label{game}
We model temporal graph discovery as an interactive game where player A (the Discoverer) aims to discover information about the graph (e.g.,~all edges and labels in the graph or a node with a specific property) while player B (the Adversary) gets to influence the infection behavior (e.g.,~by picking the edge labels) with the goal of stopping or stalling player A from achieving their goal.

Which operations and what information are available to each player are interesting parameters whose effect we will investigate in \Cref{sec:extending}.
But first, see \cref{game:graph-discovery} for a description of the common structure of all variations.
As the default model, we assume the Discoverer learns the set of static edges in step \ref{game:graph-discovery:step-info-sharing}.
While this is a strong assumption, it is pleasant to work with, and we will later show that many of our results translate to other variations.
In particular, we make such a restrictive choice to ensure our lower bound results translate.

In \Cref{sec:ideal-patient-zero}, we consider a simpler objective.
There, the Discoverer must not find out the entire graph but is tasked with finding an \emph{ideal patient zero}, that is, a combination of node and time such that an infection starting at that node and time, leads to an infection chain engulfing the whole graph.
The ideal patient zero game is a valuable stepping stone in the study of the graph discovery game.

\begin{figure}[tbph]
\rule{\textwidth}{0.4pt}
\textbf{Input:} \(\Tmax, \iphase, k, n \in \N\)
\begin{enumerate}
\item The Adversary picks the \(n\) nodes of the graph.
\item The Discoverer learns the set of nodes and possibly additional information. \label{game:graph-discovery:step-info-sharing}
\item In each round, the Discoverer may submit up to \(k\) seed infections (each specifying a node and a time) to the Adversary. The Adversary responds with an infection log consistent with these seed infections. \label{game:graph-discovery:step-rounds}
\item After a finite number of rounds, the Discoverer declares that it has discovered the whole graph. Now, the Discoverer submits a temporal graph to the Adversary. If the Adversary can output a different temporal graph consistent with all previous infection logs, the Adversary wins. Otherwise, the Discoverer wins. \label{game:graph-discovery:step-decision}
\end{enumerate}
\vspace{-0.5\baselineskip} \rule{\textwidth}{0.4pt}
\caption{The graph discovery game \label{game:graph-discovery}}
\end{figure}

We measure the quality of a Discoverer by the number of rounds it needs to win the game.
\begin{definition}
For parameters \(\Tmax{}\), \(\iphase\), \(k\) and a static graph \(G\), define the \emph{graph discovery complexity} as the minimum number of rounds required for the Discoverer to win all possible graph discovery games.
\end{definition}

We start off by giving the simplest possible algorithm: brute-forcing all combinations of nodes and time steps.
This gives us a baseline to compare both more intricate algorithms and lower bounds to.

\begin{theorem}
There is an algorithm that wins the graph discovery game in \(\nodesetsize{G}\Tmax\) rounds.
\label{thm:brute-force}
\end{theorem}

\begin{proof}
Consider the algorithm that performs the seed sets \(\{\{(v,t - 1)\} \mid v \in \nodeset{\G}, t \in [0, \Tmax-1]\}\).
Clearly, there is a successful infection along any edge at least once in these rounds.
Thus, the algorithm correctly discovers all edge labels.
\end{proof}
\chapter{Ideal Patient Zero \label{sec:ideal-patient-zero}}
\label{ideal-patient-zero}
Before we propose a better algorithm for the general graph discovery problem, we study a simpler problem: finding an ideal patient zero.

\begin{definition}
We call a node \(v\) an \emph{ideal patient zero with time \(t\)} if seed-infecting \(\{(v,t)\}\)  causes every node to be infected. In this case, we call \((v,t)\) an \emph{ideal patient zero pair}.
\label{def:ideal-patient-zero}
\end{definition}

In this section, we study a game similar to \cref{game:graph-discovery}. Define it identically except for a modification to the last step: Here the Discoverer submits a pair \((u,t) \in \nodeset{\G} \times [0,\Tmax]\) or some error \(\bot\) indicating it believes no ideal patient zero exists. The Adversary now replies with a graph \(\G=(V, E, \edgeLabelOp{})\) consistent with all infection logs.
If the Discoverer submitted a pair \((u,t)\), the Adversary wins if \(u\) is not an ideal patient zero at \(t\) in \(\G\). If the Discoverer submitted \(\bot\), the Adversary wins if there is some ideal patient zero in \(\G\).

A straightforward algorithmic approach to finding an ideal patient zero on a given temporal graph might be to check each node at the times when it has an adjacent edge.
To put it formally, for each combination \(u \in \nodeset{\defaultGraphName}\) and \(v \in \adjacentEdges{u}\), perform a round with the single seed infection \((u, \edgeLabel{u, v} - 1)\).
Surprisingly, there are instances where there is an ideal patient zero pair, but no ideal patient zero pair has this form.

\begin{lemma}
There is a temporal graph \(\G\) such that there is an ideal patient zero pair, but no ideal patient zero pair has the form \((u,\edgeLabel{e}-1)\) where  \(e\) is an edge adjacent to \(u\).
\end{lemma}

\begin{proof}
See \Cref{fig:ideal-patient-zero-non-edge-label-minus-delta-solution} for a graph with the claimed property.
\end{proof}

\begin{figure}[tbhp]
\centering
\includegraphics[width=3cm]{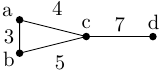}
\caption{\label{fig:ideal-patient-zero-non-edge-label-minus-delta-solution}In this temporal graph, for \(\iphase = 2\), no ideal patient zero pair has the form \((u, \edgeLabel{e} - 1)\) where \(e\) is an edge adjacent to \(u\). Also, \((a, 1)\) is an ideal patient zero pair (via the infection chain \(a \to b \to c \to d\)).}
\end{figure}

As we know from \Cref{thm:brute-force}, we can discover the entire temporal graph in \(\BigO(\nodesetsize{G} \cdot \Tmax)\).
Thus, to now obtain an interactive algorithm for the ideal patient zero problem, we need an algorithm to calculate an ideal patient zero (or determine its non-existence) from a given temporal graph.
To see that this is easily calculable, consider \Cref{lem:ideal-patient-zero-deterministic}.

\Cref{def:ideal-patient-zero} is not obviously precise, as our model allows for multiple different infection logs consistent with the same set of seed infections. Specifically, if a node \(u\) is susceptible at a time step where it has two infected neighbors, it could become infected via either edge.
Fortunately, these tie breaks are  essentially irrelevant for the ideal patient zero problem, and thus our problem is well-defined.

\begin{lemma}
Let \(\defaultGraphName\) be a graph and \((u,t)\) such that there is an infection log consistent with a seed infection only at \((u,t)\) in which all nodes become infected, then in all infection chains consistent with a seed infection only at \((u,t)\) all nodes become infected.
\label{lem:ideal-patient-zero-deterministic}
\end{lemma}

\begin{proof}
This follows directly from \Cref{lem:infection-log-to-time-table}.
\end{proof}

This leads us to the conclusion that simply enumerating and simulating all \(\Tmax \nodesetsize{G}\) possible seed node and time combinations yields a polynomial time offline algorithm to find an ideal patient zero pair (or conclude that there is none) given a known temporal graph.

\begin{corollary}
There is a Discoverer algorithm winning the ideal patient zero game using \(\Tmax \nodesetsize{G}\) rounds.
\end{corollary}

This Discoverer algorithm for the ideal patient zero game now works by applying the Discoverer algorithm that discovers the whole graph, and then computing the ideal patient zero pair offline from the discovered graph.
Unfortunately, the number of rounds is large (recall from section \Cref{thm:brute-force} that in the same number of rounds we can discover the whole graph).

This leads us to ask if we can find an algorithm that better uses the structure of the temporal graph.
In particular, observe that an ideal patient zero implies the existence of an infection chain spanning the whole graph.
By relaxing the notion of this spanning infection chain to only include information local to a single node, we observe that two edges can only ever partake in the same infection chain if they are connected via a series of edges whose time label differs by at most \(\iphase\).
From this idea, there naturally arises a new connectivity parameter for temporal graphs, which captures the constraint that nodes only stay infectious for a limited time, and thus infection chains must respect this type of waiting time constraint.

\begin{definition}
Let \(\defaultTemporalGraph\) be a temporal graph and \(\iphase \in \N^+\).
Consider the relation that links two edges if their time difference at a shared endpoint is at most \(\iphase\).
Let \(EC_\iphase\) be the partitioning of edges obtained by taking the transitive closure of this relation.
We call these the resulting equivalence classes the \emph{\(\iphase\)-edge connected components} of \(\G\).
\end{definition}

An infection originating at a node \(u\) can only infect a node \(v\) if they each have an edge in the same \(\iphase\)-edge connected component.
The opposite is not true, as you can see in \Cref{fig:delta-connected-not-infectious}.

\begin{figure}[tbhp]
\centering
\includegraphics[width=5cm]{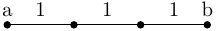}
\caption{\label{fig:delta-connected-not-infectious}While nodes \(a\) and \(b\) each have an edge in the same \(\iphase\)-edge connected component (regardless of the choice of \(\iphase\)), no infection originating at \(a\) can ever infect \(b\).}
\end{figure}

By definition, any infection chain caused by infecting a single seed node (thus also that of an ideal patient zero) must be contained in a single \(\iphase\)-edge connected component.
This motivates us to propose \Cref{alg:follow} which is faster as it only requires \(\BigO(\edgesetsize{G})\) rounds.
See \Cref{fig:follow-execution} for an illustration of an exemplary execution.

\begin{algorithm}[tbph]
    \DontPrintSemicolon
    \SetKwFunction{FMain}{Follow}
    \SetKwFunction{ex}{Explore}
    \SetKwProg{Fn}{fun}{:}{}
    \Fn{\FMain{$G$, $\iphase$}}{
        1. Pick node $v_0 \in \nodeset{G}$ arbitrarily.

        2. For \(i \in \left [0, \left \lceil \Tmax / \iphase \right \rceil \right]\), perform a round with seed infection \((v_0, i \iphase)\). Record successful infections.

        3. For each edge $v_0u$ that successfully infects: \ex{$e, \edgeLabel{e}$}.
    }
    \Fn{\ex{$u$, $t$}}{
    1. For each \(t' \in \{t-\iphase-1, t-1, t\}\):
        If there has been no round with seed infection $(u, t')$, seed an infection $(u, t')$. Store that this has been done.

    2. For each newly infected edge $uv$ with infection time $t$, do $\ex{v,t}$.
    }
    \caption{The $\FMainHack$ algorithm discovers the neighbors of $v_0$ and calls $\exHack$ on them. $\exHack$ then discovers the respective $\iphase$-edge connected component.}
    \label{alg:follow}
\end{algorithm} %

The next lemma is the crucial property that allows us to prove the correctness of \Cref{alg:follow} as it allows us to later argue that we do not miss relevant edges.

\begin{lemma}
Let \(v\) be a node that is seed-infected at time step \(t\) in the execution of the \(\FMain\) algorithm (\Cref{alg:follow}) and \(e\) an edge adjacent to \(v\) with label in \(t+1\) to \(t+\iphase\). Then there is a round such that there is a successful infection via \(e\) (from \(v\) or to \(v\)).
\label{lem:follow-edge}
\end{lemma}

\begin{figure}[t]
\centering
\includegraphics[width=\textwidth]{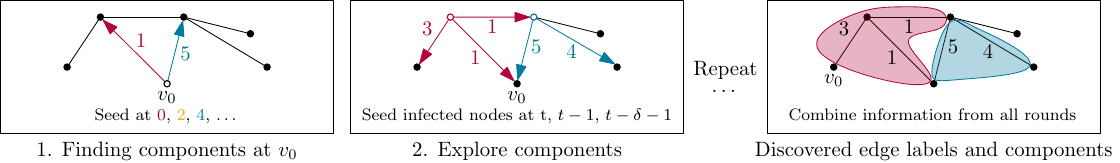}
\caption{\label{fig:follow-execution}An execution of the \(\FMainHack\) algorithm. In this example, \(\iphase = 2\). We perform the initial search for an edge at \(v_0\). Ringed nodes indicate seed infections. Next, we explore the two nodes we found in the initial search by performing seed infections at \(t\), \(t-1\), and \(t - \iphase - 1\) where \(t\) is the time step we observed the node being infected. We repeat that until we have no such seed infections to perform without repetitions. Since no \(\iphase\)-edge connected component spans all nodes, we conclude that no ideal patient zero pair exists.}
\end{figure}

\begin{proof}
Let \(e = uv \in \edgeset{G}\), \(t \in [0, \Tmax]\) such that \(v\) is seed-infected at \(t\) during the execution of the \(\FMain\) algorithm.
We argue the property via downwards induction over \(\edgeLabel{e} - t\) and call this the \emph{overtaking budget}.
Intuitively, it is the time in which \(u\) could be infected by some other path than via \(e\).
Observe that if \(\edgeLabel{e} - t = 1\), the infection attempt along \(e\) must be successful as any other path to \(u\) has hop-length at least two and thus takes at least two time steps because we have strictly increasing paths.

If the infection along \(e\) is successful in the currently considered round, we are done.
So assume the infection attempt along \(e\) is unsuccessful.
Then \(u\) must be infected via a different path \(p'\) at or before \(\edgeLabel{e}\).
Let \(e' \ne e\) be the last edge on \(p'\).
By line 1 in the \(\ex\) algorithm, there must be a round where \(u\) is seed-infected at \(\edgeLabel{e'} - 1\).
Note that since \(p'\) is strictly increasing and has at least hop-length two, we have \(\edgeLabel{e'} -1 > t\). Thus, the overtaking budget for that round is strictly smaller than \(\edgeLabel{e} - t\).
Note that \(v\) and \(u\) swap roles in this recursive application, but this is immaterial as the edges are undirected.

Since the overtaking budget strictly decreases with recursive applications, it must reach 1, at which point the infection attempt along \(e\) must be successful, as argued above.
\end{proof}

\begin{corollary}
If \(\FMain\)  discovers an edge from an \(\iphase{}\)-edge connected component, then it discovers all edges from that component.
\label{cor:follow-whole-component}
\end{corollary}

With this tool in hand, we can proceed to prove the correctness of the \(\FMain\) algorithm.

\begin{theorem}
The \(\FMain\) algorithm (\Cref{alg:follow}) correctly solves the ideal patient zero problem.
\label{thm:follow-correct}
\end{theorem}

\begin{proof}
Let \(v_0\) be the start vertex picked in the algorithm.
Observe that the loop in step 2 of \(\FMain{}\) discovers at least one edge from each \(\iphase{}\)-edge connected component adjacent to \(v_0\).
Applying \Cref{cor:follow-whole-component}, we can see that then the algorithm discovers all edges in these components.
To finish the proof we are left to show that if we know the labels of edges that are in the same \(\iphase\)-connected edge component as any edge at \(v_0\), we can find an ideal patient zero if it exists.

Assume that there is an ideal patient zero.
Now, take any infection chain caused by seed-infecting this ideal patient zero pair.
We construct the \emph{directed tree of successful infections} by only including edges along which there was a successful infection and by directing all edges in the direction along which the infection traveled.
By definition of the ideal patient zero pair, this tree spans the entire graph.
Thus, one of its edges is adjacent to \(v_0\).
By definition, the directed tree of infections must be in a single \(\iphase\)-connected edge component.
Therefore, all edge labels of the tree as well as all other edge labels that could affect the outcome of the ideal patient zero infection (edges in different components do not interact if there is only one seed infection) are known at the end of the algorithm.

Since we learn all relevant edges through the algorithm, we can decide the existence of an ideal patient zero.
\end{proof}

\begin{theorem}
The \(\FMain\) algorithm terminates after at most \(6 \edgesetsize{G} + \left \lceil T_\mathrm{max}/\iphase \right \rceil\) rounds.
\end{theorem}

\begin{proof}
The search of edges adjacent to \(v_0\) takes \(\left \lceil T_\mathrm{max}/\iphase \right \rceil\) rounds.
The \(\ex\) sub-algorithm uses at most 6 rounds per edge \(e\) (at most \(\edgeLabel{e}\), \(\edgeLabel{e}-1\), and \(\edgeLabel{e}-\iphase-1\) per endpoint).
This yields the desired bound.
\end{proof}

The algorithm can be more closely analyzed to yield an \(O(\min(m, n T_\mathrm{max}) + T_\mathrm{max}/\iphase)\) round bound, but this does not seem useful since, even as soon as \(nT_\mathrm{max} \le m\), we might as well use the trivial algorithm from \Cref{thm:brute-force} to discover the whole graph (not just this subset) in \(O(n T_\mathrm{max})\) rounds.
\chapter{Deriving a Graph Discovery Algorithm \label{sec:fd-algorithm}}
\label{fd-algorithm}
We now extend the ideas of the previous section to obtain a better graph discovery algorithm.
Interestingly, the \(\FMain\) algorithm (\cref{alg:follow}) explores precisely the \(\iphase\)-edge connected components adjacent to the start node \(v_0\), because calling the explore subroutine on a single edge explores precisely its \(\iphase\)-edge connected component.
Since we have discovered a whole graph precisely iff we have discovered all its \(\iphase\)-edge connected components, we can derive \cref{alg:follow-discover} which discovers the entire graph.

    \begin{algorithm}[tbhp]
        \DontPrintSemicolon
        \SetKwFunction{FDMain}{DiscoveryFollow}
        \SetKwFunction{ex}{Explore}
        \SetKwProg{Fn}{fun}{:}{}
        \Fn{\FDMain{$G$, $\iphase$}}{
            \While{there is a node $v_0$ with an adjacent edge for which the label is still unknown}{
                In steps of size $\iphase$ perform rounds such that $v_0$ is infection at every time step in $[\Tmax]$ in at least one round.

For each edge $e$ which successfully infects from $v_0$ to a neighbor: $\ex{$u$, $\edgeLabel{e}$}$.
            }
        }
    \caption{The graph discovery extension of the \(\FMainHack \) algorithm.}
    \label{alg:follow-discover}
    \end{algorithm}

\begin{theorem}
\Cref{alg:follow-discover} wins any instance of the graph discovery game on a graph \(\G\) in at most \(6\edgesetsize{\G} + \decc{\G} \left\lceil \Tmax/\iphase \right \rceil\) rounds.
\label{thm:follow-discover}
\end{theorem}

Note, this running time now depends on the edge labels and not only on the static part of the graph.
The algorithm crucially depends on the temporal connectivity structure of the graph to be discovered.

\begin{proof}
The correctness follows since, by \Cref{cor:follow-whole-component}, every call of the \(\ex\) algorithm discovers all edges which share a \(\iphase\)-edge connected component with one edge adjacent to the start vertex (i.e.,~the one picked as \(v_0\) in step 1).
Since we explore all \(\iphase\)-edge connected components, we discover all edges.

For the running time, observe that the loop in  \cref{alg:follow-discover} runs at most \(\decc{\G}\) times.
Since by the same argument as before, there are at most \(6\) infections per edge (in essence we avoid duplication), the factor only applies to the second summand, yielding the stated bound.
\end{proof}
\chapter{Lower Bounds for Temporal Graph Discovery \label{sec:witnesses}}
\label{witnesses}
Let us build a small toolkit to prove lower bounds for our graph discovery game.
Observe that at the beginning of the graph discovery game, every edge could have every possible label.
As the game progresses, the Adversary has to reveal information about the edge labels to the Discoverer, thus reducing the labels any edge could have.
In particular, if an infection attempt via an edge is successful, this completely fixes the label of the edge (the only label consistent with information the Adversary has revealed to the Discoverer is the time of the successful infection).
If an infection attempt is unsuccessful, but one of the endpoints stayed susceptible, this reduces the set of possible labels by at least one (precisely by the time stamp of the failed infection attempt).
We can use these observations to derive a potential argument technique for finding lower bounds on the graph discovery complexity.

\begin{theorem}
Let \(\defaultTemporalGraph\) be a temporal graph and \(T_\mathrm{max}\), \(\iphase\), \(k\) parameters for the graph discovery game. For a sequence of seed infection sets \(S_1, \dots, S_a\), we define \(\Phi(i)\) as the sum of the sizes of the sets of consistent labels over all edges after rounds 1 to \(i\). Then, \(S_1, \dots, S_a\) is a \emph{witnessing schedule} iff \(\Phi(a) = \edgesetsize{\G}\).
\label{thm:potential-argument}
\end{theorem}

\begin{proof}
Clearly, if there is more than one consistent label for some edge, the Adversary will always win step \ref{game:graph-discovery:step-decision} in the graph discovery game (\cref{game:graph-discovery}).
\end{proof}

Recall that we have an algorithm for graph discovery taking \(\BigO(n \Tmax)\) rounds (see \Cref{thm:brute-force}).
We will now show that, in the worst case, this can be improved by at most a factor of \(\iphase k\).

\begin{theorem}
For all \(\iphase, k \in \N^+\) and \(\Tmax \ge 4\), there is an infinite family of temporal graphs \(\{\G_n\}_{n \in \N}\) such that the graph \(\G_n\) has \(\Theta(n)\) nodes and the minimum number of rounds required to satisfy the graph discovery game grows in \(\Omega(n (T_\mathrm{max} - 3) / (\iphase k))\). Also, all graphs in the family are temporally connected.
\label{thm:temporally-connected-graph-discovery-lower-bound}
\end{theorem}

Consider \Cref{fig:temporally-connected-graph-discover-lower-bound} for an illustration of this family of graphs.

\begin{figure}[tbhp]
\centering
\includegraphics[width=6cm]{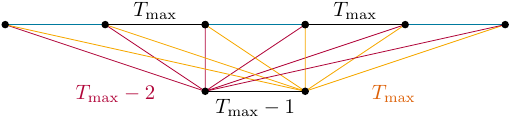}
\caption{\label{fig:temporally-connected-graph-discover-lower-bound}The construction from the proof of \Cref{thm:temporally-connected-graph-discovery-lower-bound}.}
\end{figure}

\begin{proof}
Let \(n, \iphase \in \N^+\) be arbitrary and let \(\Tmax \ge 4\). For simplicity, assume \(n\) is even. We will fix some edge labels, meaning the Adversary will always pick them to be the same value. Construct the host graph as follows:
\begin{enumerate}
\item \(v_1, \dots, v_{n-2}\) form a path of \(n-2\) nodes.
\item \(v_{n-1}\) has an edge with a fixed label \(T_\mathrm{max}-2\) to each of the  nodes \(v_1, \dots v_{n-2}\),
\item \(v_{n-1}\) and \(v_{n}\) share an edge with the fixed label \(T_\mathrm{max}-1\), and
\item \(v_{n}\) has an edge with a fixed label \(T_\mathrm{max}\) to each of the  nodes \(v_0, \dots v_{n-2}\).
\end{enumerate}

We assume the Discoverer acts arbitrarily and describe responses for the Adversary, which yields a lower bound as desired.

First observe that the graph is temporally connected since for all \(i,j \in [n-2]\), the path \(v_i \to v_{n-1} \to v_{n} \to v_j\) is temporal.
The nodes \(v_{n-1}\) and \(v_{n}\) are connected to all others in a similar fashion.

For an infection attempt along an edge whose time label we fixed during the construction, the Adversary responds according to that time label.
For all other infection attempts, the Adversary answers “infection failed” as long as the set of possible consistent labels for an edge is non-empty after that answer.
Otherwise, it replies “infection successful” with an arbitrary label from the set of consistent answers.

We apply the potential argument from \cref{thm:potential-argument} to all edges that do not have a fixed label.
Call these edges \emph{relevant}.
As the potential for a subset of edges is a lower bound for the potential of all edges, this restriction is sound.
Of these non-fixed edges there are \((n-3)/2\).
In the beginning, for every edge there are \(T_\mathrm{max}\) possible labels, thus clearly \(\Phi(0) \ge (n-3)/2 \cdot T_\mathrm{max}\).

Now, for every node which is infected, at most \(\iphase\) labels per adjacent relevant edge are removed from the set of consistent labels, because (1) if the Adversary answers “infection failed” for all \(\iphase\) infection attempts, either because the target is resistant or because the Adversary has decided against picking this label, at most \(\iphase\) labels are made impossible or (2) if the Adversary replies “infection successful” then there must have been \(\iphase\) or fewer consistent labels left.

Combining these arguments, we see that the potential decreases by at most \(\iphase k\) in each round. Dividing the starting potential by this decrease, yield that any Discoverer needs at least
$$
\left \lfloor \frac{n\cdot (T_\mathrm{max}-3)}{2 \iphase k} \right \rfloor
$$
rounds to win the game with the described Adversary.
\end{proof}

We are now ready to derive two lower bounds that almost match the upper bound of the \(\FDMain\) algorithm.

\section{Witness Complexity}
\label{sec:orgb2cfeb5}
Intuitively, the witness complexity of a temporal graph is the minimum number of rounds a Discoverer who knows all the labels needs to perform to convince an observer of the correctness of the claimed labeling.
We give a formal definition.

\begin{definition}
A length \(a\) \emph{witnessing schedule} for a temporal graph \(\defaultTemporalGraph\) is a sequence of seed infection sets \(S_1, \dots, S_a\) such that after performing \(a\) rounds with the respective seed infection sets, all labels in the graph are uniquely determined by the logs of these rounds.
The \emph{witness complexity} of a temporal graph \(\G\) is the length of the shortest witnessing schedule for it.
\end{definition}

Observe that the graphs constructed in the proof of \Cref{thm:temporally-connected-graph-discovery-lower-bound} have witness complexity \(O(n)\), which is significantly less than their complexity in the graph discovery game.
Yet, witness complexity proves a powerful tool for lower bounds on the graph discovery complexity, especially in cases where we look for a bound not including \(\Tmax\).
But first, let us state the formal relationship between the two complexities.

\begin{lemma}
Let \(\G\) be a temporal graph and \(T_\mathrm{max}\), \(\iphase\), \(k\) parameters as defined above. Then the witness complexity of this instance is at most as large as its graph discovery complexity.
\end{lemma}

\begin{proof}
This follows directly from the fact that recording the seed infection sets of any online algorithm yields a witnessing schedule upon termination.
\end{proof}

As stated above, the witness complexity technique is limited in its power, as it can only ever be at most the number of edges in a graph.

\begin{theorem}
For any instance \((\G, \Tmax{}, \iphase, k)\), the witness complexity is at most \(\edgesetsize{\G}\).
\end{theorem}

\begin{proof}
Let \(e_1, \dots, e_{\edgesetsize{\defaultGraphName}}\) be an arbitrary numbering of the edges, and set \(S_i \coloneqq \{(e_i, \edgeLabel{e_i}-1\}\).
This is a witnessing schedule of length \(\edgesetsize{\defaultGraphName}\).
\end{proof}

We now show that this worst case is actually tight, and there are graphs that asymptotically require about one round per edge to witness correctly.
Observe that this bound is irrespective of \(k\),  thus does not have one of the major shortfalls of our previous lower bounds for the graph discovery game.

Let us start by giving the main theorem of this section.
\begin{theorem}
There is an infinite family of temporal graphs \(\{\G_n\}_{n \in \N}\) such that their witness complexity grows in \(\Omega_k(\edgesetsize{\G_n})\).
\label{thm:omega-m-witness-complexity}
\end{theorem}

The construction of this family is a bit more involved.
First, we define the family of graphs we will study for the rest of this section.
We then define when we call an infection attempt \emph{relevant}, in the sense that it makes meaningful progress towards winning the witness complexity game.
Lastly, we show that there can be at most one such relevant infection attempt per round in the family of graphs we study.
From this we conclude the desired lower bound.

For the graph construction, let \(x \in \N\) be our size parameter. We construct a graph of size \(n \coloneqq 5x\). We define four subsets of nodes, let
\begin{align*}
L &\coloneqq \{\ell_1, \dots, \ell_x\}, \\
R &\coloneqq \{r_1, \dots, r_{2x}\}, \\
B &\coloneqq \{b_1, \dots, b_x\},\\
C &\coloneqq \{c_1, \dots, c_x\}.
\end{align*}

Set \(V(\G_n) \coloneqq L \sqcup R \sqcup B \sqcup C\). Notate by \(R_2\) the nodes in \(R\) with an even index.

Before we give a formal definition of the edges and their labels, let us briefly give an intuition for the construction: \(L\) and \(R\) form a complete bipartite graph.
These are the edges we are actually interested in, that is, for whom we will examine how quickly an algorithm can prove their label.
The correct labels will be such that the edges between a single node in \(L\) and all nodes in \(R\) are in their own \emph{phase} (a contiguous set of time steps, we make this rigorous later).
The nodes in \(R\) are connected carefully to ensure that once a node in \(R\) becomes infected, the infection spreads through \(R\) and no other node in \(R\) becomes infected from \(L\).

We add \(B\) and \(C\) as gadgets to ensure that at the end of a phase, all nodes in \(L\) and \(R\) become infected via edges not between \(L\) and \(R\) such that no further information can be gained about those edges.

First, let \(P = \{p_1, \dots, p_x\}\) be a set of \(x\) edge-disjoint Hamiltonian paths on \(R\).
The existence of such a path is proven in \cite{axiotis_approx}.
Note that by their construction, for every path, there are at most two nodes with an odd index in a row.
Assume without loss of generality that each \(p_j\) begins at \(r_{2j}\).
Furthermore, write \(r_{p(i,j)}\) for the \(j\)-th node from \(R_2\) on the path \(p_i\).
And write \(p^{-1}(i,j)\) for its index on \(p_i\).

Then, \(E(\G_n) \coloneqq L \times R_2 \sqcup R^2 \sqcup L\times B \sqcup B\times R\sqcup B^2 \sqcup C \times R \sqcup \{b_ic_i\mid i \in [x]\}\). Now, set
\begin{align*}
&\text{for }\ell_i \in L, j \in [x], &&\text{ set } \edgeLabel{\ell_i, r_{p(i,j)}} \coloneqq i(4x+1) + 4(p^{-1}(i,j))+1,  \\
&\text{for } p_q \in P, j\in[x], &&\text{ set } \edgeLabel{(p_q)_j, (p_q)_{j+1}} \coloneqq q(4x+1) + 2j+1,  \\
&\text{for } b_i \in B, j \in [x], &&\text{ set } \edgeLabel{b_i, r_{(p_i)_j)}} \coloneqq i(4x+1) + 2j + 2,  \\
&\text{for } b_i \in B, b_j \in B, i<j, &&\text{ set } \edgeLabel{b_i, b_j} \coloneqq (i+1)\cdot(4x+1) -2, \\
&\text{for } i \in [x] &&\text{ set } \edgeLabel{b_i, c_i} \coloneqq (i+1)\cdot(4x+1) - 1, \\
&\text{for } i > j \in [x] &&\text{ set } \edgeLabel{b_i, c_j} \coloneqq (i+1)\cdot(4x+1) - 2, \\
&\text{for } c_i \in C, r_j \in R &&\text{ set } \edgeLabel{c_i, r_j} \coloneqq (i+1)\cdot(4x+1), \\
&\text{for } \ell_i \in L, b_j \in B, &&\text{ set } \edgeLabel{\ell_i, b_j} \coloneqq (j+1)\cdot(4x+1).
\end{align*}
Let \(k\) be arbitrary, \(T_\mathrm{max} \coloneqq x(4x+1)\), and \(\iphase \coloneqq 4x+1\).

Before we start our argument, let us formalize the notion of phases.
Observe that the edges can be partitioned into \(x\) sets, each having labels in a fixed interval of size \(\iphase\).

\begin{definition}
Formally for \(i \in [x]\), let \(E_i = \{e \in E(\G_n) \mid \edgeLabel{e} \in [\iphase \cdot i, \iphase \cdot (i+1)-1]\}\).
We refer to these intervals and their edges as \emph{phases}.
\end{definition}

We are now ready to give a formal definition of when an infection attempt is relevant.
\begin{definition}
Call an infection attempt between some \(l_i \in L\) and \(r_{2j} \in R_2\) \emph{relevant} if (a) it happens at \(\edgeLabel{l_i r_{2j}}\) and is successful or (b) it happens at \(\edgeLabel{l_i r_{2j}} - 1\), was unsuccessful and exactly one of the endpoints was infected.
\end{definition}

With these definitions in hand, let us outline the central properties from which the result follows:
\begin{itemize}
\item we need at least one relevant infection attempt per edge in \(L \times R_2\) to win the witness complexity game,
\item we can have at most one relevant infection per phase, and
\item we can have at most one phase with a relevant infection per round.
\end{itemize}

\begin{lemma}
If, at the end of the witness complexity game, the Prover wins, there must have been one relevant infection attempt for every edge in \(L \times R_2\).
\label{lem:number-relevant-attempts}
\end{lemma}

\begin{proof}
The result follows since if there is an edge in \(L \times R_2\) for which there was no relevant infection attempt, then both \(\edgeLabel{e}\) and \(\edgeLabel{e} - 1\) are possible labels consistent with the infection log.
Therefore, the Prover must lose the game, since the Adversary can always pick the label the Prover does not pick out of these two.
\end{proof}

\begin{lemma}
There is at most one phase with relevant infection attempts per round.
\label{lem:active-phases-per-round}
\end{lemma}

The proof of this statement is fairly technical. We argue that all infection chains must essentially follow the pattern illustrated in \Cref{fig:construction-omega-m}.
Our construction contains gadgets to ensure this also happens in all edge cases, which we carefully consider in the proof to make this idea rigorous.

\begin{figure}[tb]
\centering
\includegraphics[width=6cm]{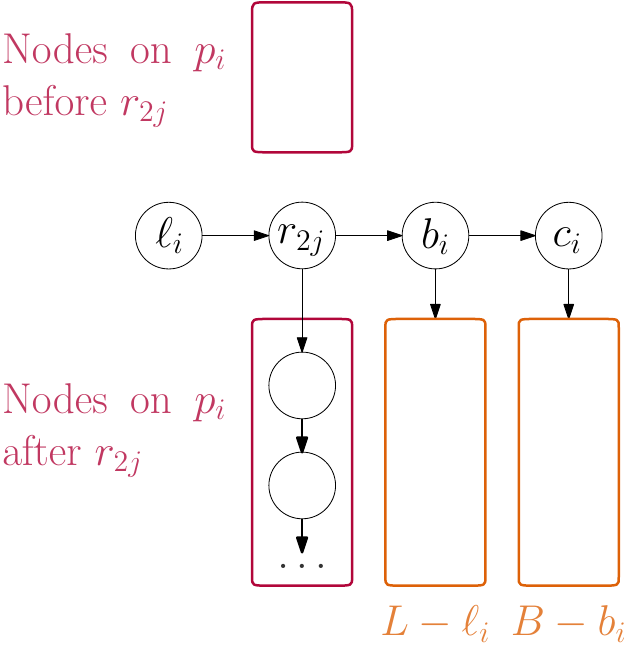}
\caption{\label{fig:construction-omega-m}As an illustrative example, consider the directed tree of infections (an edge \(u \to v\) means that \(v\) is infected via \(uv\)) that occurs when \(\ell_i\) is infected such that the relevant infection attempt is along \(\ell_ir_{2j}\). The other cases are similar. Colored arrows mean that all nodes in a set are infected via an edge from the origin node.}
\end{figure}

To prove this lemma, we need to formalize the notion of an active phase, which, in essence, is the phase where the relevant infection attempts must take place.

\begin{definition}
Let \(v\) be the first node infected in a given round, let \(t\) be the infection time, and let \(i\) be the phase of that infection time.
Then one of two things must be the case:
\begin{enumerate}
\item At time step \((i+1)\iphase\) all nodes in \(L\) and \(R_2\) are infected or resistant, and thus no relevant infection attempt may take place after that time. In this case, we call \(i\) the active phase.
\item At time step \((i + 1) \iphase\) not all nodes in \(L\) and \(R_2\) are infected or resistant. In that case, call \(i+1\) the active phase.
\end{enumerate}
\end{definition}

\begin{proof}
In what follows, we prove that there is at most one active phase and that any two relevant infection attempts must occur within that active phase.
For technical reasons, the first infection in a round is either in the active phase or right before the start of the active phase.

Towards that end, we prove that \(b_i\) is infectious at time steps \((i+1)\iphase\) and \((i+1)\iphase-1\) or \(b_{i+1}\) is infectious at \((i+2)\iphase\) and \((i+2)\iphase -1\), but in the second case, there are no infection attempts before \(i(4x+1) + 4x + 1\) (and thus no relevant infection attempts).

Note that for all infection chain descriptions below, we can ignore the fact that, when we claim a node infects another, there may have been additional seed infections which lead to infecting one of the endpoints earlier. Since we know that there are no infections before \(i\iphase\) and are only interested in infections until \((i+1)\iphase\), earlier infections still leave the nodes infectious at the claimed time steps anyways, except in those cases where we explicitly argue about scenario two.

\uline{Case 1:} \(v = b_i\). If \(t \le i(4x+1) + 4x + 1\), \(b_i\) infectious at both \((i+1)\iphase\) and \((i+1)\iphase -1\). If not, case two clearly holds. There will be no relevant infections (because it’s too late) and \(b_i\) infects \(b_{i+1}\) at time step \((i+2)\iphase - 2\) which is then infectious at \((i+2)\iphase\) and \((i+2)\iphase -1\).

\uline{Case 2:} \(v \in B-b_i\). If \(t \le i(4x+1) + 4x + 1\), then \(v\) infects \(b_i\) at time step \((i+1)\iphase - 2\) and the statement holds. Otherwise, argue scenario two analogously to case 1.

\uline{Case 3:} \(v \in L\). If \(v=\ell_i\) and the infection occurs before \(i(4x+1) + 4x + 1\), then \(\ell_i\) infects some node in \(R_{2}\) within the next four time steps, which then infects \(b_i\) right after. In all other cases, \(v\) infects \(b_i\) at time step \((i+1)\iphase\) which then infects \(b_{i+1}\) at time step \((i+2) \iphase -2\) which then infects \(c_{i+1}\) at time step \((i+2)\iphase -1\), so that at time step \((i+2) \iphase\) all nodes in \(L\) and \(R_2\) become either infected or are already resistant or infected.

\uline{Case 4:} \(v \in R\). Assume \(v = r_{(p_i)_j}\). Then, if the seed infection occurs before \(i(4x+1) + 2j + 2\), \(b_i\) gets infected at that time via the edge \(r_jb_i\). If the infection occurs after \(i(4x+1) + 4j + 2\), we have to analyze a bit more carefully. Now, \(c_i\) becomes infected at time step \((i+1)\iphase\).  Then, \(c_i\) infects \(b_{i+1}\) at \((i+2)\iphase-2\), fulfilling our condition.

\uline{Case 5:} \(v \in C\). Let \(v = c_j\) with \(j \in [x]\). If the infection occurs before \((i+1)\iphase -1\), then \(b_i\) becomes infected at \((i+1)\iphase - 1\). If the infection occurs after \((i+1)\iphase -1\), then \(b_{i+1}\) becomes infected at \((i+2)\iphase -1\). If the infection is this late, clearly, no relevant infection can occur before \((i+1)\iphase\), that is the end of phase \(i\).

We have seen that if the seed infection(s) occur before \(i(4x+1) + 4x + 1\), \(b_i\) is infectious at time steps \((i+1)\iphase - 1\) and \((i+1)\iphase\). In this case we call \(i\) the first active phase. In the other case, \(b_{i+1}\) is infectious at \((i+2)\iphase - 1\) and \((i+2)\iphase\) and we call \(i+1\) the first active phase.

Let \(i'\) be the first active phase, then at time step \((i'+1)\iphase -1\) the node \(b_{i'}\) infects \(c_{i'}\) (or \(c_{i'}\) is already infectious). Then at time step \((i'+1)\iphase\), all nodes in \(L\) and \(R\) become infected (if they are not infected or resistant before). Thus, after that time step, no relevant infection attempt may occur anymore.
\end{proof}

\begin{lemma}
There is at most one relevant infection attempt per phase.
\label{lem:relevant-attempts-per-phase}
\end{lemma}

\begin{proof}
Let \(i\) be the first active phase.
Then any relevant infection attempt must be made in connection with \(\ell_{i}\) since the time steps related to relevant infection attempts for other nodes in \(L\) are all either in earlier (inactive) phases or in later phase when, as proven above, no relevant infection attempt may take place.

Now assume there are two relevant infection attempts along two different edges \(e\) and \(e'\).
Write \(e = \ell_i r_{2j}\) and \(e' = \ell_i r_{2j'}\) and assume without loss of generality that \(\edgeLabel{e} < \edgeLabel{e'}\).

Observe that since \(\edgeLabel{e} < \edgeLabel{e'}\), it follows that \(r_{2j'}\) must be further back on \(p_i\).
If the infection attempt along \(e\) is relevant because it’s successful, then \(r_{2j'}\) is infected at the latest via \(p_i\) one time steps before a relevant infection attempt can occur via \(e'\).

Now assume that there is a relevant infection attempt along \(e\) but there is no successful infection via \(e\).
This could only happen if one of two things were the case: (1) \(\ell_i\) is infectious at \(\edgeLabel{e} -1\), but not at \(\edgeLabel{e}\) or (2) \(r_{2j}\) is already infected or resistant at \(\edgeLabel{e}\).

(1) cannot happen because then \(\ell_i\) must have been infectious before \((i-1)\iphase + 4x + 1\), which contradicts our assumption.
Regarding, (2) if \(r_{2j}\) is infectious long enough to infect their next node on \(p_i\), the argument proceeds as above.
If \(r_{2j}\) is not infectious anymore by that time, it must have been infectious before \((i-1)\iphase + 4x + 1\), again contradicting our assumption.
\end{proof}

\begin{proof}[Proof of \cref{thm:omega-m-witness-complexity}]
Putting together \cref{lem:number-relevant-attempts,lem:active-phases-per-round,lem:relevant-attempts-per-phase}, we get the desired bound.
\end{proof}
\chapter{Extending the Model \label{sec:extending}}
\label{extending}
So far, we have looked at a simple version of our graph discovery game.
In particular, we have made a few assumptions about our infection behavior and Discoverer knowledge that might seem restrictive and less close to the real-world processes.
For example, we assumed that we not only learn at what time step a node was infected, but also by whom, a property that is, for example, not usually given in disease infection tracing applications.
Similarly, we assumed that the static graph is known to the Discoverer or that an edge has at most one label.
In this section, we will show that lifting these restrictions either leads to the same or equivalent theoretical behavior or trivializes the problem so that it is not solvable significantly faster than by the trivial algorithm.
This way, we legitimize our previous restrictions, as they yield relevant theoretical results while being easier to work with.
See \Cref{tbl:variations-overview} for an overview of the lower bounds and algorithms for these variations.
\todo[color=pink]{@supervisors: Do I need to explain the variations here?}

 \ifpaper \graphDiscoveryResultsTable \fi

\section{Infection Times Only–Feedback \label{sec:infection-time-feedback}}
\label{sec:org0759acf}
First, let us look at a variation of the graph discovery game where the player only gets feedback on when a node is infected, but not by which other node.
Call this \emph{infection times only–feedback}.
This makes the Discoverer's job harder.
In some cases (e.g.,~when a node only has one neighbor), the Discoverer can still deduce who infected a node, but generally that is not the case anymore.
To see this, consider a case where a node becomes infected at some time when it has two infectious neighbors. Then, one of the edges must have the label of the infection time, but the Discoverer cannot directly deduce which of the edges, as it could in our basic model.
Since, up until now, we have looked at an easier (from the Discoverer's perspective) version of the graph discovery game, lower bounds directly translate.
We will see this pattern for the other variations as well, though we will not state it with this level of formality from now on.

\begin{theorem}
Let \((G_n, (\Tmax)_n, \iphase_n, k_n)_{n \in \N}\) be a family of instances of the graph discovery game under the basic model and \(f: \N^5 \to \N\) such that any Discoverer must use at least
\(\Omega(f(\nodesetsize{G_n}, \edgesetsize{G_n}, (\Tmax)_n, \iphase_n, k_n))\)
rounds, then the same lower bound holds if the game is played under the infection times only–variation.
\end{theorem}

\begin{proof}
The result follows since the Discoverer gets strictly less information, and the rest of the process is exactly the same.
Thus, every algorithm winning the game under the infection times only–variation must also win the game under the basic model.
\end{proof}

Crucially, applying this to \Cref{thm:temporally-connected-graph-discovery-lower-bound} and \Cref{thm:omega-m-witness-complexity} yields the following two results.

\begin{corollary}
For all \(\iphase, k \in \N^+\) and \(\Tmax \ge 4\), there is an infinite family of temporal graphs \(\{\G_n\}_{n \in \N}\) such that the graph \(\G_n\) has \(\Theta(n)\) nodes and the minimum number of rounds required to satisfy the graph discovery game under the infection times only–variation grows in \(\Omega(n (T_\mathrm{max} - 3) / (\iphase k))\). Also, all graphs in the family are temporally connected.
\label{thm:temporally-connected-graph-discovery-lower-bound-infection-times-only}
\end{corollary}

\begin{corollary}
There is an infinite family of instances \(\{\G_n\}_{n \in \N}\) such that the witness complexity  under the infection times only–variation grows in \(\Omega_k(\edgesetsize{\G_n})\).
\end{corollary}

What is not quite as obvious is that both the trivial and the \(\FDMain\) graph discovery algorithm translate, giving us the following two results.

\begin{theorem}
\Cref{alg:follow-discover} wins any instance of the graph discovery game under the infection times only–variation on a  graph \(\G\) in at most \(6\edgesetsize{\G} + \decc{\G} \left\lceil \Tmax/\iphase \right \rceil\) rounds.
\end{theorem}

\begin{proof}
First, observe that the algorithm does not require the source of an infection to be known. Calling the \(\ex\) subroutine simply needs the fact that a node was infected and at what time.

Secondly, observe that at the end, for each edge \(e = \{u,v\}\) there has been a seed infection at \(u\) and \(v\) at \(\edgeLabel{e} - 1\), thus the infection must have been successful and taken place in the first possible time step (after the seed infection).
As any infection chain not along the single edge between \(u\) and \(v\) takes at least two time steps, we can soundly conclude the label of the edge from the infection logs.

Together, these two properties mean that the \(\FDMain\) algorithm translates to the infection times only variation.
\end{proof}

\section{Unknown Static Graph \label{sec:unknown-static-graph}}
\label{sec:orgbfedc01}
It is a fair criticism that it is not always realistic to assume the Discoverer knows the static graph at the start of the game.
This motivates us to investigate a version of the game where only the node set, but not the edges, are known to the Discoverer at the start of the game.
Here too, the lower bounds from the basic model translate as the Discoverer gets strictly less information, but we can also achieve new, stricter bounds, which show that not only does our \(\FDMain\) not translate, no comparably fast algorithm can even exist.

\begin{theorem}
Consider the variation of the graph discovery game where the Discoverer does not learn the static graph. Let \(n, T_\mathrm{max} \in \N^+\) and \(k \in [n], m \in [{n \choose 2}-n], \iphase \in [T_\mathrm{max}]\). Then there is an Adversary such that any algorithm winning this game variation on graphs with \(n\) nodes must take at least \(\lfloor nT_\mathrm{max}/(2\iphase k) \rfloor\) rounds. This Adversary picks a graph with at most two \(\iphase\)-edge connected components.
\label{thm:unknown-static-graph-lb}
\end{theorem}

Note that this result holds regardless of the number of edges in the graph and for just two \(\iphase\)-edge connected components.
Therefore, we cannot hope for an algorithm only dependent on the number of edges and \(\iphase\)-edge connected components (as we have in the basic model).
In the basic model, nodes without edges are the best case, as we can simply ignore them.
Here, the opposite is true, since we perform tests to ensure the non-existence of the edge at every time step.

\begin{proof}
As the underlying graph, we fix all \(m\) edges but one, which we pick dynamically.
Choose the \(m -1\) fixed edges arbitrarily such that every node has at most \(n-2\) adjacent edges and that their subgraph is connected.
This is possible by a simple greed strategy starting at a single node \(v\), connecting it to arbitrary nodes that have less than \(n-2\) neighbors until \(v\) has  \(n-2\) neighbors, and then moving on to one of these neighbors to do the same.
Repeat this process until \(m-1\) edges have been added.
Give those edges the time label 1.
For these, the Adversary responds to infection attempts by simply simulating the correct behavior (i.e.,~an attempt is successful precisely at time step 1).

For all other possible edges, respond to all infection attempts with ‘infection failed’ as long as after that response there is a still are \(u,v \in \nodeset{G}\) and \(t \in [\Tmax]\) such that
\begin{itemize}
\item \(uv\) is not one of the edges fixed before, and
\item there has been no unsuccessful infection attempt from \(u\) to \(v\) or \(v\) to \(u\) at \(t\).
\end{itemize}
Otherwise, respond with `infection successful`.
Clearly, no Discoverer can terminate and win until one round before that has happened, since there are still multiple consistent options for the unfixed edge remaining.

A similar argument to \Cref{thm:temporally-connected-graph-discovery-lower-bound} shows that, since no infection can ever spread, the online algorithm essentially has to search through all nodes at all time steps, working on at most \(k\) nodes per round and covering at most \(\iphase\) time steps per node.
This yields the \(\lfloor n \Tmax / (2\iphase k) \rfloor\) bound.
\end{proof}

Therefore, while the goal of this variation seems natural, it makes the problem difficult to efficiently solve.
In particular, exploiting the structure of the static graph is hard, as both edges and non-edges need to be proven.
Naturally, this means that any improvement over the brute-force method is limited.
On the positive side, the brute force method still works.

\begin{theorem}
There is an algorithm that wins the graph discovery game in the unknown graphs variation in \(\nodesetsize{G}\Tmax\) rounds.
\end{theorem}

\begin{proof}
As for \Cref{thm:brute-force}, the algorithm that performs the seed sets \(\{\{(v,t - 1)\} \mid v \in \nodeset{G}, t \in [0, \Tmax-1]\}\) trivially wins the game.
\end{proof}

While, in that model, we were able to find a better algorithm for sparse graphs (in particular the \(\FDMain\) algorithm), we have no such hope here by \Cref{thm:unknown-static-graph-lb}.
In fact, these results show that this naive algorithm is close to optimal.

\section{Multilabels \& Multiedges \label{sec:multigraphs}}
\label{sec:org0b1c5a2}
The last extension we investigate is what happens if an edge between the same two nodes may exist at multiple time steps.
There are two ways to model this: we could allow each edge to have a set of labels instead of just one (we call these \emph{multilabels}) or we could allow multiple distinct edges (with different labels) to exist between the same two nodes (we call these multiedges).
Note that, while for most problems in the literature these notions are equivalent, that is not the case here.
With multiedges, the Discoverer learns the multiplicity of every edge, but with multilabels it does not.
If we only tell the Discoverer where an edge is, but not how many time labels it has, then we run into the same issue as with unknown static graphs in \Cref{sec:unknown-static-graph}.
In essence, any Discoverer would be forced to check all combinations of nodes and seed times, only allowing for the trivial algorithm.
We first formalize this notion by giving the appropriate lower bound for multilabels before giving more positive results for multiedges.

\begin{theorem}
Consider the variation of the graph discovery game where a single edge might have an arbitrary subset of \([\Tmax]\) as labels. Let \(n, T_\mathrm{max} \in \N\) and \(k \in [n], m \in [{n \choose 2}-n], \iphase \in [T_\mathrm{max}]\). Then there is an Adversary such that any algorithm winning this game variation on graphs with \(n\) nodes must take at least \(\lfloor \min\{n/2, m\}\Tmax/(\iphase k) \rfloor\) rounds. This Adversary picks a graph with at most two \(\iphase\)-edge connected components.
\label{thm:multilabels-lb}
\end{theorem}

Again, this lower bound tells us we may not hope for an algorithm taking advantage of a small number of \(\iphase\)-edge connected components.
Similarly, any algorithm can only take advantage of the knowledge of the static edges when there are few of them (precisely if the number of edges is sublinear in the number of nodes).

\begin{proof}
The proof is analogous to the one of \Cref{thm:unknown-static-graph-lb} with three minor modifications: (1) we must fix all edges and let the Discoverer search for possibly existing labels instead of edges, (2) we must assure that any node has few adjacent edges, and (3) the analysis of the potential now analyzes the search for the multilabels and respects this smallest vertex cover.
We need to bound the number of edges adjacent to any one node to ensure that we cannot check all the edges using a small number of nodes to perform seed infections at.
Observe that this condition implies that the smallest vertex cover is large (i.e.,~we need many nodes to cover all edges).

To address (1) and (2), construct the edges of the graph greedily by starting at an arbitrary node and connecting it to a different arbitrary node.
Now, until you have added \(m\) edges, repeatedly consider the node that has the smallest positive number of adjacent edges and add an edge to the node with the smallest number of adjacent edges (0 if possible).
Clearly, in the resulting graph, all nodes that have at least one edge are in the same connected component, and thus, after labeling them \(1\),  all these edges are in the same \(\iphase\)-edge connected component.
There will be at most one additional label assigned by the Adversary, again yielding at most two \(\iphase\)-edge connected components in total.
Also, any node has degree at most \(\lceil m /n \rceil\), thus any vertex cover must have at least \(\lceil m / (\lceil m / n \rceil) \rceil\) nodes.
Note that this implies that if \(m \ge n\), the size of the minimum vertex cover is at least \(\lfloor n/2 \rfloor\).

Regarding (2), observe that, initially, there are \(m \Tmax\) possible edge labels (call that the initial potential).
In this variation, a successful infection only tells the Discoverer that the edge has that label but not that it does not have other labels, as is the case in the basic game.
Also, by the construction of our graph, the smallest vertex cover is large, so any seed infection can only test a small number of edges.
Concretely, any node has at most \(\lceil m / n \rceil\) adjacent edges.
Therefore, any round can decrease the potential by at most \(\lceil m / n \rceil \delta k\).
Dividing \(m \Tmax\) by \(\lceil m / n \rceil \iphase k\) yields the desired bound.
\end{proof}

The situation looks much better if we consider temporal multiedges (i.e.,~temporal multigraphs where there can be multiple distinct edges between the same two nodes) instead.
Here, the core advantage we have is that we know the multiplicity of the edges.
In the proof of \Cref{thm:multilabels-lb} we saw that the crucial property that forces the Discoverer to spend so many rounds is that discovering a label on an edge does not preclude it from having to check all other possible time steps for more labels.
Note that we disallow multiple edges with the same endpoints and the same label.

First, notice that any temporal graph is also a temporal multigraph.
Therefore, we have the following lower bound as a corollary to \Cref{thm:temporally-connected-graph-discovery-lower-bound}.

\begin{corollary}
Consider the variation of the graph discovery game with multigraphs. For all \(\iphase, k \in \N^+\) and \(\Tmax \ge 4\), there is an infinite family of temporal graphs \(\{\G_n\}_{n \in \N}\) such that the graph \(G_n\) has \(\Theta(n)\) nodes and the minimum number of rounds required to satisfy the graph discovery game grows in \(\Omega(n (T_\mathrm{max} - 3) / (\iphase k))\). Also, all graphs in the family are temporally connected.
\label{cor:multigraph-lower-bound}
\end{corollary}

On a more positive note, the \(\FDMain\) algorithm (\Cref{alg:follow-discover}) translates as well.

\begin{theorem}
Consider the variation of the graph discovery game with multigraphs.
\Cref{alg:follow-discover} wins any instance of the graph discovery game on a graph \(\G\) in at most \(6\edgesetsize{\G} + \decc{\G} \left\lceil \Tmax/\iphase \right \rceil\) rounds.
\end{theorem}

Note that here, we count every multiedge individually.
In essence, this means that we only pay additional rounds for the additional multiplicity of the edges.

\begin{proof}
This result follows since the proofs of \Cref{thm:follow-discover}, \Cref{lem:follow-edge} and thus \Cref{cor:follow-whole-component} hold analogously on multigraphs.
\end{proof}

\section{Discussion}
\label{sec:org4a67c74}
In conclusion, our \(\FDMain\) algorithm translates to situations where we lift some of the restrictions we assumed for the rest of this paper.
In particular, it still works if the Discoverer only gets infection-time feedback or if we allow multiedges, lifting the two most restrictive assumptions we previously made.
We also prove that variations where the Discoverer needs to ensure the non-existence of a high number of edges or labels (such as the unknown static graph or multilabel variations) do not allow for significant improvements over the trivial algorithm.
In particular, our \(\FDMain\) algorithm is not applicable.
Again, see \Cref{tbl:variations-overview} for a summary of these results.
\chapter{Experimental Evaluation \label{sec:experiments}}
\label{experiments}
The gap between the lower and upper bounds for the graph discovery problem is relatively small, but only tells us about the worst-case performance, motivating us to investigate the performance of our graph discovery algorithm on common synthetic graphs and real-world data.
\section{Hypotheses}
\label{sec:orgadaa33f}
We formulate three hypotheses we aim to test with our experiments.

\begin{hypothesis}
The number of rounds required to discover a temporal graph is linear in the number of edge labels.
\end{hypothesis}
This first hypothesis is motivated by the worst-case analysis from \Cref{thm:follow-discover}.
We aim to test how closely the performance of our algorithm in practice matches this theoretical bound.
Note, this also takes into account the effect of the additional optimization described in the setup (\Cref{sec:setup}).

\begin{hypothesis}
Graphs with higher density spend fewer rounds on component discovery.
\label{hypo:component-discovery}
\end{hypothesis}
The \(\FDMain\) algorithm works in two distinct phases: (1) the main loop in \Cref{alg:follow-discover} discovers new \(\iphase\)-edge connected components (we thus call it the \emph{component discovery phase}) and (2) the \(\ex\) routine from explores the found components (we thus call it the \emph{component exploration phase}).
Both of these stages require the Discoverer to spend rounds, and their respective cost is dependent on the structure of the graph to be explored.
As we prove in \Cref{thm:follow-discover}, the component discovery phase is triggered at most \(\defaultdecc\) times.
We hypothesize that \(\decc{\G}\) tends to be lower in denser graphs as the components tend to merge as more edges are inserted, which would lead to a relatively lower cost for component discovery as compared to component exploration.

\begin{hypothesis}
In Erdős-Renyi graphs, the percentage of rounds spend on component discovery shows a threshold behavior in \(\Tmax / (\nodesetsize{\defaultGraphName} \cdot p)\). This effect is mediated by \(\defaultdecc\).
\label{hypo:threshold}
\end{hypothesis}
\Cref{hypo:threshold} is a specification of \Cref{hypo:component-discovery} for Erdős-Renyi graphs.
In particular, we can view this hypothesis as the temporal extension of the typical threshold behavior that static Erdős-Renyi graphs exhibit in \(np\) \cite{erdos1960evolution}.
We also conjecture this behavior holds provably in \Cref{conj:erdos}.

\section{Setup \label{sec:setup}}
\label{sec:org88bc0a4}
We perform our evaluation on two different data sets.
One is synthetically generated, while the other is real-world data from the Stanford Network Analysis Project collection.

Firstly, we evaluate our algorithm on Erdős-Renyi graphs. While the classical model has been developed for static graphs, it is commonly extended to temporal graphs \cite{casteigts_threshold,Angel_Ferber_Sudakov_Tassion_2020}.
To generate a simple temporal graph from an Erdős-Renyi graph, we simply choose each edge label uniformly at random from the set \([\Tmax]\).
This means that \(\Tmax\) is now the third parameter to generate these graphs, in addition to the classical ones \(n\) (the number of nodes) and \(p\) (the density parameter).
Denote such a graph as \(ERT(n, p, \Tmax)\).

Secondly, to evaluate our algorithm on real-world data, we employ a data set from the Stanford Large Network Dataset Collection \cite{kumar2021deception}.
The \texttt{comm-f2f-Resistance} collection is described by the project as a set of 62 “dynamic face-to-face interaction network[s] between a group of participants”.

In our implementation, we employ a small optimization upon \Cref{alg:follow-discover}.
Concretely, we skip what we call \emph{redundant seed infections}.
A seed infection \((v, t) \in \nodeset{\G} \times [0, \Tmax]\) is \emph{redundant} if we already know the labels of all edges adjacent to \(v\) at the time in the algorithm we would perform this seed infection.
Note, this does not depend on \(t\).
For an illustration, consider the example from \Cref{fig:follow-execution}.
There, the \(\FDMain\) algorithm would, after discovering the only edge adjacent to the leftmost node by a seed from its other adjacent node, perform seed infections at the leftmost node even though we already know the label of the only edge we could possibly discover.

We run the \(\FDMain\) algorithm on all graphs from the SNAP dataset and for a wide range of parameters for the temporal Erdős-Renyi graphs.
Concretely, we test with 1 to 100 nodes in steps of size 5, for probabilities \[p \in \{0.01, 0.05, 0.1, 0.15, 0.2, 0.25, 0.3, 0.35, 0.4, 0.5, 0.7, 0.9\}.\]
We pick \(\Tmax\) as a factor of \(n\), testing with \[\Tmax  / n \in \{0.05, 0.1, 0.2, 0.3, 0.4, 0.5, 0.7, 0.9, 1.0, 2.0, 3.0, 5.0, 7.0, 10.0\}.\]
Similarly, \(\iphase\) is 1 or a multiple of \(\Tmax\), namely \(\iphase / \Tmax \in \{0.01, 0.05, 0.1, 0.3, 0.5\}\)  
We record both the number of rounds needed to complete the discovery and how many of these rounds are spent in the component discovery versus the component exploration phases of the algorithm.

Both our implementation and analysis code is available under a permissive open-source license and can be used to replicate our findings.\footnote{See \url{https://github.com/BenBals/dynamic-network-discovery}.}

\section{Results}
\label{sec:orgd74da0e}
We now critically evaluate our hypotheses against the data thus obtained and compare the effects between the different data sets and parameters.
Particularly, we pay attention to evidence on how the different hypotheses interact.
Finally, we give \Cref{conj:erdos} as a result of our analysis of \Cref{hypo:threshold}.

\subsection{Hypothesis 1}
\label{sec:org088cf45}
\begin{figure}[t]
\begin{subfigure}{11cm}
    \centering
    \includegraphics[width=\linewidth]{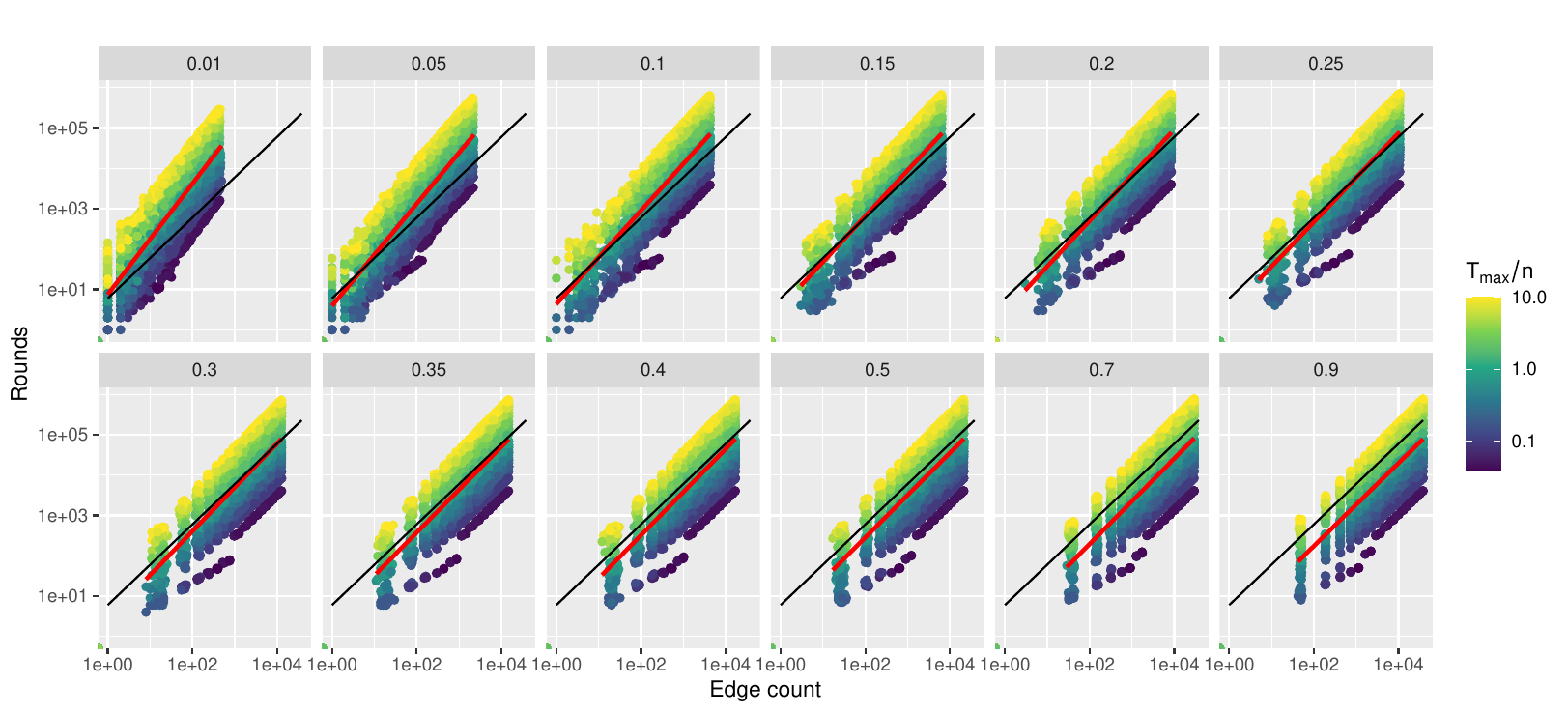}
    \caption{Erdős-Renyi graphs for different densities \(p\)}
    \label{fig:rounds_by_edge_count:erdos}
\end{subfigure}%
\hfill%
\begin{subfigure}{3cm}
    \centering
    \includegraphics[width=\linewidth]{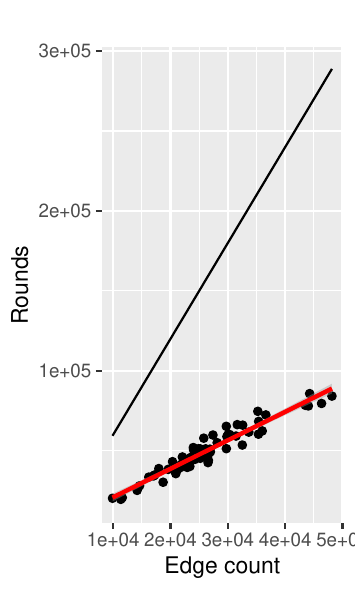}
    \caption{SNAP dataset}
    \label{fig:rounds_by_edge_count:snap}
\end{subfigure}
\caption{Number of rounds by number of edges for the \(\FDMainHack \) algorithm. The black line indicates the \(6\edgesetsize{\G}\) line (which is our theoretical bound for the component exploration phase), and the red line is the trend line (i.e.,~the linear regression).}
\label{fig:rounds_by_edge_count}
\end{figure}

To investigate our first hypothesis, we compare how the number of rounds relates to the number of edges.
See \cref{fig:rounds_by_edge_count} for the results on our datasets.
We can see that in Erdős-Renyi graphs, the relationship closely follows the \(6\edgesetsize{\G}\) line predicted by our theoretical results.
Clearly, this effect is influenced by other parameters such as \(p\), \(\Tmax\), and \(n\), whose roles we will examine in the discussion of the other hypotheses.
However, if we consider graphs with the same \(p\) and ratio \(\Tmax / n\) (i.e.,~one facet and one color in \cref{fig:rounds_by_edge_count}), we see that the relationship is strictly linear—the points form a tightly distributed straight line.
 We see that the trend for \(p \le 0.25\) is slightly above \(6\edgesetsize{\G}\) and slightly under it for larger values of \(p\).
This occurs since, when \(p\) is small, more time is spent on component discovery than on component exploration.
We will explore this more thoroughly in the discussion of the results regarding \Cref{hypo:component-discovery},

In the SNAP data set, we see that the trend is linear in the number of edges, but we require significantly less than \(6\edgesetsize{\G}\) rounds to discover the graph.
In fact, the gradient of the regression line is only 1.78.
This can be explained by the optimization skipping redundant infections as outlined in our experimental setup, as this enables the algorithm to require less than the \(6\) infections per edge, which would otherwise be strictly required.

In summary, while there is a strong linear relationship following the \(6\edgesetsize{\G}\) line, the number of rounds also significantly depends on other properties of the graph.
We can see these effects in both synthetic and real-world data.

\subsection{Hypothesis 2}
\label{sec:orgefb0d05}
\begin{figure}[tbph]
\centering
\includegraphics[width=9cm]{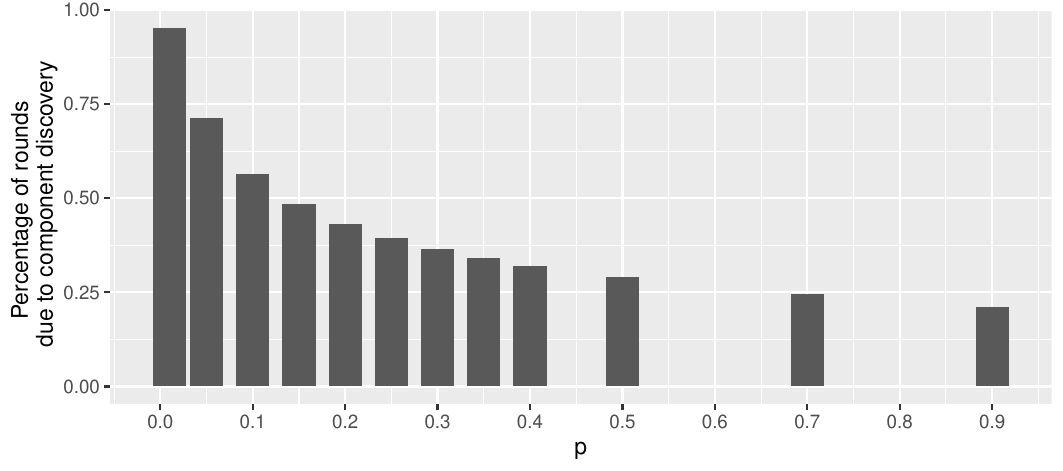}
\caption{\label{fig:erdos-renyi_component-discovery-percentage-by-p}The percentage of rounds the \(\FDMainHack\) algorithm spends in the component discovery phase on temporal Erdős-Renyi graphs by their density \(p\).}
\end{figure}

As this hypothesis is about the relationship between graph density and time spent on component discovery, we only analyze it on the Erdős-Renyi graphs, as the SNAP dataset does not exhibit significant differences in graph density.

In \Cref{fig:erdos-renyi_component-discovery-percentage-by-p}, we plot the percentage of time spent on component discovery dependent on the parameter \(p\) (which specifies the density of the graph).
We observe a clear and strong, inversely proportional relationship.
This leads us to accept \cref{hypo:component-discovery}.
This makes sense in the light of the theoretical analysis of the \(\FDMain\) algorithm, as we would expect denser graphs to have fewer \(\iphase\)-edge connected components, thus less need to discover new components.

\subsection{Hypothesis 3}
\label{sec:org13b004a}
\begin{figure}[t]
\begin{subfigure}{6.5cm}
    \centering
    \includegraphics[width=\linewidth]{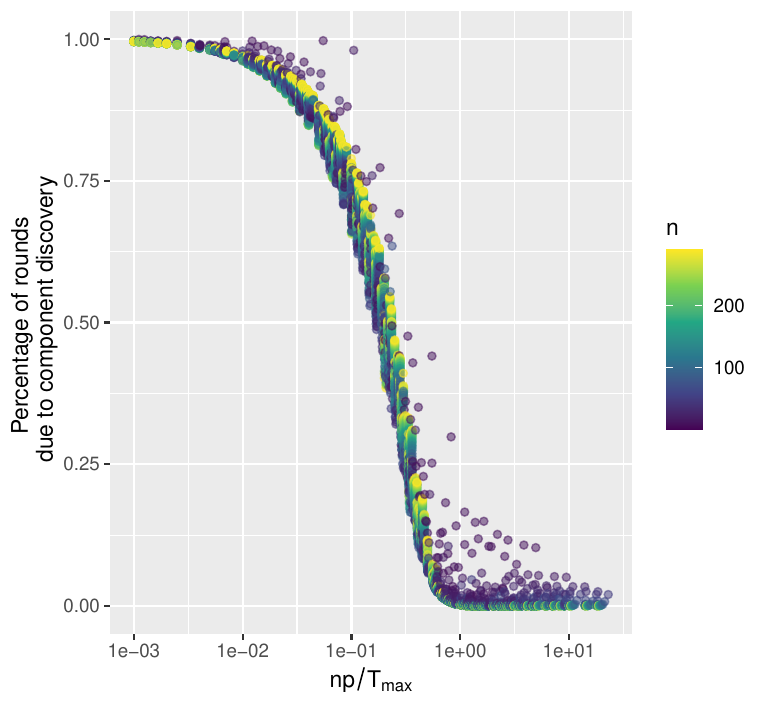}
    \caption{The percentage of rounds spent in the component discovery phase by \(np / \Tmax\).
    \label{fig:erdos-renyi_component-discovery-percentage-by-n-times-p-over-Tmax}}
\end{subfigure}%
\hfill%
\begin{subfigure}{6.5cm}
    \centering
    \includegraphics[width=\linewidth]{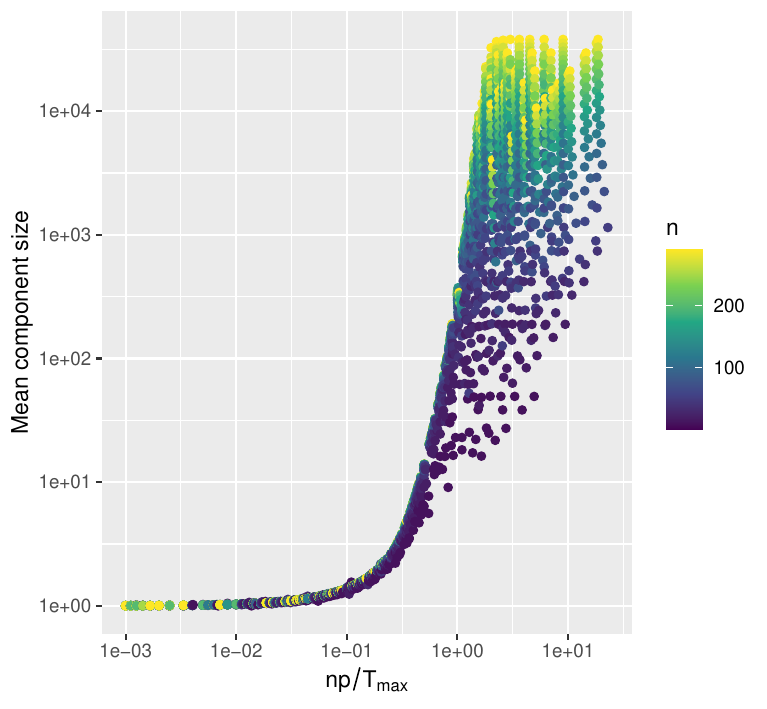}
    \caption{Mean size of the \(\iphase\)-edge connected components by \(np / \Tmax\).
    \label{fig:erdos-renyi_component_mean_size-by-n-times-p-over-Tmax}}
\end{subfigure}
\caption{Threshold behavior in the cost of component discovery in the \(\FDMainHack \) algorithm. Note that both graphs use logarithmic axes, and the color denotes the number of nodes.}
\label{fig:threshold}
\end{figure}

Examining \cref{fig:erdos-renyi_component-discovery-percentage-by-n-times-p-over-Tmax}, we can see a clear threshold between \(np/\Tmax = 0.01\) and \(np/\Tmax = 1.00\) where we move from spending only a small amount of rounds on component discovery to spending close to all of our rounds on component discovery.
Our hypothesis, that this is mediated by the size of the \(\iphase\)-edge connected components, is supported by the results in \cref{fig:erdos-renyi_component_mean_size-by-n-times-p-over-Tmax}.
Here, the additional differentiation by color for larger values of \(np / \Tmax\) can be explained because the size of a component is constrained by the size of the graph.

Note that big \(\iphase\)-edge connected components imply that we can discover many edges in the component exploration phase for a single node explored in the component discovery phase.
That means the plotted percentage shrinks as the average size of \(\iphase\)-edge connected components increases.
This behavior is similar to that seen in static Erdős-Renyi graphs, where if \(np=1\), there almost surely is a largest connected component (in the classical sense) of order \(n^{2/3}\).
And if \(np\) approaches a constant larger than 1, there asymptotically almost surely is a so-called giant component containing linearly many nodes \cite{erdos1960evolution}.
This motivates us to conjecture the following similar behavior.
\begin{conjecture}
Let \(p \in [0,1], \Tmax, n \in \N\) and \(\G \coloneqq ERT(n, p, \Tmax)\). Then
\begin{itemize}[topsep=5pt]
\item if \(\Tmax/(np) \xrightarrow{n \to \infty} c\) where \(c < 1\), then \(\decc{\G} / \edgesetsize{\G} \xrightarrow{n \to \infty} 0\), and
\item if \(\Tmax/(np) > 1\), then almost surely, all \(\iphase\)-edge connected components have size at most \(\BigO(\log n^2) = \BigO(\log n)\).
\end{itemize}
 \ifpaper \else {\vspace{-1em}}\fi
\label{conj:erdos}
\end{conjecture}

Investigating the connectivity behavior of random temporal graphs in a way that respects their inherent temporal aspects is an interesting and little understood problem \cite{casteigts_threshold,casteigts2024simple}.
Casteigts et al.~\cite{casteigts_threshold},
theoretically study sharp thresholds for connectivity in temporal Erdős-Renyi graphs, and also conjecture about a threshold for the emergence of a node-based giant (i.e.,~linear size) connected component.
This conjecture can be seen as capturing the equivalent behavior under waiting time constraints (i.e.,~where edges are only considered if their edge labels do not differ too much).
\chapter{Summary \& Future Work}
\label{sec:org9527ff4}
We give a comprehensive theoretical and empirical study of graph discovery in temporal graphs.
Our lower bounds and algorithms give a clear insight into what is possible in this problem space and what is not.
The \(\FDMain\) algorithm provides an efficient graph discovery strategy requiring only \(6\edgesetsize{\G} + \decc{\G} \left\lceil \Tmax/\iphase \right \rceil\) rounds.
Crucially, our lower bound in \Cref{thm:omega-m-witness-complexity} proves that any algorithm must spend at least \(\Omega(\edgesetsize{\G})\) rounds, showing the \(\FDMain\) algorithm is close to optimal.
The analysis of a number of settings not only broadens the applicability of our results, but also gives us insights into which properties make the problem hard.
When considering infection times-only–feedback or multiedges, the \(\FDMain\) algorithm still works.
For the variations with multilabels or unknown static graphs on the other hand, we prove there can be no algorithm whose running time only depends on the number of edges.

Our empirical analysis highlights the relevance of our theoretical results for practical applications and gives rise to interesting insights of its own.
We see that on Erdős-Renyi graphs, the observed performance of the \(\FDMain\) algorithm matched our theoretical analysis.
On real-world data from the SNAP collection, the algorithm even slightly outperforms our predictions.
Finally, we observe a close link between the parameters of the temporal Erdős-Renyi model, the temporal connectivity structure of the resulting graphs, and our algorithmic performance, creating a bridge between our theoretical insights on \(\iphase\)-connected components and their empirical behavior.

Yet, the study of spreading processes in temporal graphs is young, and the study of graph discovery on temporal graphs is especially so, thus it is natural that this first work gives rise to a number of interesting open questions.
In particular, future work could aim to completely tighten the lower bound from \Cref{thm:temporally-connected-graph-discovery-lower-bound} and give a lower bound that is tight in \(k\) and \(\iphase\),
Another promising avenue is to investigate \(\defaultdecc\) further. In particular, to prove or disprove the observed threshold in Erdős-Renyi graphs (\Cref{conj:erdos}) and investigate an analog to the static \(\ln n / n\) threshold. Especially, since this is essentially the waiting-time analog of the sharp connectivity thresholds
Casteigs et al.~proved
\cite{casteigts_threshold}.
Finally, future research could explore more variations of graph discovery, such as finding specific nodes or checking for structural properties instead of discovering the whole graph.

\FloatBarrier

\part{Source Detection \label{part:source-detection}}
\chapter{The Source Detection Game \label{sec:game-def}}
\label{game-def}
In our game, two players interact with each other. The \emph{Discoverer} attempts to find a fixed source of infections (i.e.,~a troll in a social network) and the \emph{Adversary} controls the environment as well as the position of that source and attempts to conceal the source from the Discoverer.
The goal of the Discoverer is to find the source while allowing as few successful infections until then as possible.
We refer to this number of infections as the \emph{price of detection}.

\begin{figure}[tbph]
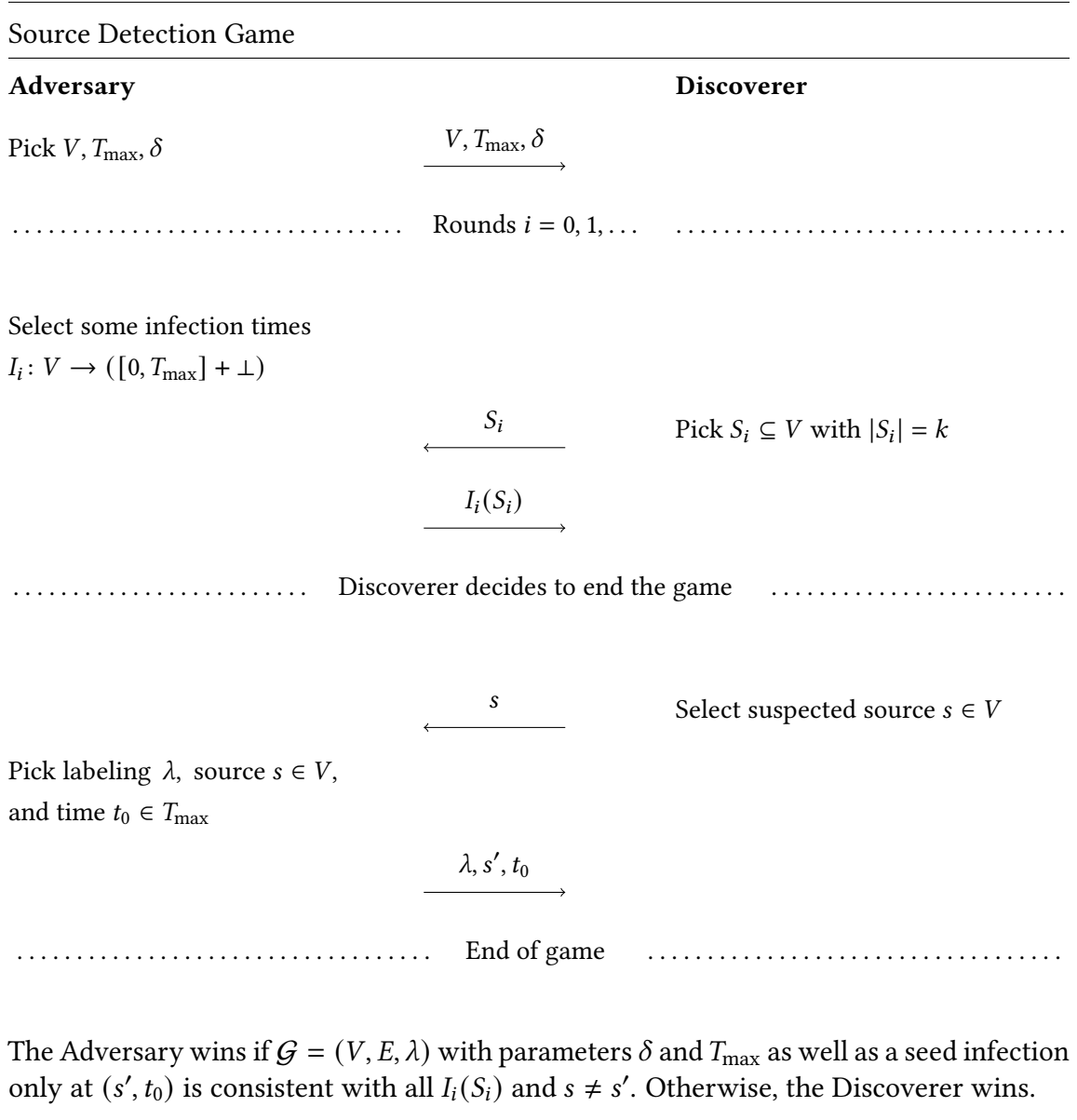

\rule{\textwidth}{0.4pt}\vspace{-0.75\baselineskip} %
\procedureblock[width=\textwidth]{Source Detection Game}{%
\textbf{Adversary} \< \< \textbf{Discoverer} \\
\text{Pick } V, \Tmax, \iphase \< \sendmessageright*{V, \Tmax, \iphase} \< \pclb
\pcintertext[dotted]{\quad Rounds \(i=0,1,\dots\) \quad}  \\
\text{Select some infection times} \\ I_i \colon V \to ([0,\Tmax] + \bot)  \< \< \\
\< \sendmessageleft*{S_i} \< \text{Pick } S_i \subseteq V \text{ with  } |S_i| = k \\
  \< \sendmessageright*{I_i(S_i)} \< \pclb
\pcintertext[dotted]{\quad Discoverer decides to end the game \quad}  \\
\< \sendmessageleft*{s} \< \text{Select suspected source } s \in V \\
\text{Pick labeling } \edgeLabelOp, \text{ source } s \in V, \\ \text{and time } t_0 \in \Tmax \< \< \\
\< \sendmessageright*{\edgeLabelOp, s', t_0} \< \pclb
\pcintertext[dotted]{\quad End of game \quad}
}
The Adversary wins if \(\G=(V, E, \edgeLabelOp)\) with parameters \(\iphase\) and \(\Tmax\) as well as a seed infection only at \((s', t_0)\) is consistent with all \(I_i(S_i)\) and \(s \ne s'\). Otherwise, the Discoverer wins.

\vspace{-0.25\baselineskip} \rule{\textwidth}{0.4pt}
\caption{The source detection game in the variation with consistent source behavior and known static graph. \label{game:source-detection}}
\end{figure}
\todo[color=pink]{Fix that the Discoverer learns the infection source}

We will investigate different restrictions upon the knowledge and power of the Discoverer and Adversary.
The dimensions whose impact on the price of detection we explore are:
\begin{itemize}
\item \emph{Infection behavior}: Does the source have to infect in the same fashion each round, or can it change its behavior between rounds? In any case, the Adversary is oblivious to which nodes the Discoverer watches.
\item \emph{Watched nodes}: How many nodes may the Discoverer watch each round? We investigate the cases where the Discoverer watches one or two nodes.
\item \emph{Discoverer knowledge}: Which information about the graph is available to the Discoverer? We investigate the case where the Discoverer knows the underlying static graph and where it only knows the nodes of the graph but not the edges.
\item \emph{Graph class}: We examine how the structure of the underlying static graph affects the price of detection. Concretely, we look at general graphs, trees, and graphs of bounded treewidth.
\end{itemize}

Furthermore, for many of these problems, we investigate their price of detection where the Discoverer is a randomized algorithm.
There, we aim for bounds that hold with constant probability.
See \Cref{tbl:results} for an overview of our results  \ifpaper \else {for the source detection problem}\fi.
\chapter{Randomized Source Detection \label{sec:rand}}
\label{rand}
Randomizing the behavior of the Discoverer enables us to evade the worst cases and yield good results with high probability.\todo[color=pink]{@Supervisors: Should I include a very bad lower bound on worst-case performance to justify this?}
In fact, in this section, we even show that there is a randomized Discoverer that wins the source detection game with consistent source behavior on unknown static graphs within \(O(n \sqrt n)\) infections and with constant success probability.
See the definition of \cref{alg:aY1}.

\begin{algorithm}[tbhp]
        \DontPrintSemicolon
        \SetKwFunction{FMain}{UnknownSourceDetection}
        \SetKwProg{Fn}{fun}{:}{}
        \Fn{\FMain}{
           1. Pick $\sqrt{n}$ nodes uniformly at random and watch them each for one round.

           2. Now let $a$ be the node among those picked that is infected the earliest.

           3. Recursively, pick the next $a$ as the node that infected the previous $a$.

           4. If there is no neighbor left that gets infected earlier, we must have found the root.
        }
    \caption{The randomized source detection algorithm for unknown static graphs and consistent source behavior.}
    \label{alg:aY1}
\end{algorithm}

Notice that this algorithm always finds the source and has a constant probability of terminating within \(2 \sqrt n\) rounds. Modifying step 3 such that it aborts after \(\sqrt{n}\) rounds yields an algorithm that always terminates within \(2\sqrt n\) rounds and has a constant probability of finding the source.

\begin{theorem}
\cref{alg:aY1} solves the source detection problem with consistent source behavior on unknown static graphs within \(O(n \sqrt n)\) infections with constant probability.
\label{thm:aY1}
\end{theorem}

\begin{proof}
The proof consists of two parts: first, we show that if we pick a good node as \(a\) in step 2, then steps 3 and 4 take \(O(\sqrt{n})\) rounds. Second, we show that picking a good node is sufficiently likely.

Let \(T\) be the tree of the infection behavior.
Observe that when node-labeled with the infection times, \(T\) is a min-heap.
For a node, call the distance from the root its \emph{infection distance}.
For part one, assume that \(a\) has an infection distance at most \(\sqrt{n}\).
Clearly, then step 3 recurses at most \(\sqrt n\) times.

For part two, set \(t\) to be the smallest time step at which any node \(v\) of infection distance greater than \(\sqrt n\)  gets infected, and let \(T'\) be the induced subgraph on \(T\) of the nodes that are infected before \(t\). Since \(T\) is a min-heap with regard to infection times, \(T'\) is also a tree (and not a forest).
By definition, there are at least \(\sqrt{n}\) nodes in \(T'\) (examine the path on which \(v\) lays).
Also observe that if the algorithm picks a node from \(T'\) in step 1, then it also picks a node from \(T'\) in step 2, giving us the result for the first part.
Lastly, let us examine how likely it is that a node from \(T'\) is picked in step 1.
For that, see that the probability of the complementary event is
\begin{align*}
\left(\frac{n - \sqrt{n}}{n} \right)^{\sqrt{n}},
\end{align*}
since there are \(n - \sqrt n\) nodes not in \(T'\) and we pick \(\sqrt n\) nodes. It’s well known that this term approaches \(1/e\) as \(n\)  approaches \(n\).
\end{proof}

\begin{theorem}
Let \(2 < \funnyconstant < 3\). For any algorithm solving the source detection problem with consistent source behavior, unknown static graph and watching a single node with success probability \(p >0\) there is an infinite family of graphs (which are all trees) such that the price of detection under that algorithm is at least \(\sqrt{\ln\left(\frac{1}{(1-p/\funnyconstant)^2}\right)} n \sqrt{n}\).
\label{thm:aY1-lb}
\end{theorem}

Interestingly, this also gives us a bound on the trade-off between the success probability and the price of detection.
For a more intuitive (though worse bound), obverse that by simple analytic tools, if \(2 < \funnyconstant < 3\), then we have that \(\sqrt{\ln\left(\frac{1}{(1-p/\funnyconstant)^2}\right)} > p\) for \(p \in (0,1)\).
Meaning, if we settle for a decrease in the success probability, we get less than a proportional improvement in the number of tolerated infections.
Note that we do not know of an algorithm exploiting the limited slack afforded by this trade-off.

\begin{proof}
First, we describe an infinite family of graphs, then assume there in an algorithm winning the game on these instances in \(o(n \sqrt{n})\) infections with probability at least \(p\). We then derive an algorithm that must solve the problem strictly faster and show that, for all small enough \(c\) (which we will give later), there are problem instances where this improved algorithm has probability less than \(p\) to require less than \(c\sqrt{n}\) rounds.

First of all, let \(n \in \N\). Then set \(P_n\) the path with \(n\) nodes and let the time labeling be \(\edgeLabel{i,i+1} = i\) for \(i \in [n]\). Now let \(A\) be any algorithm that solves the source detection problem on the set \(\{P_n\}_{n \in \N}\). Without loss of generality, we may assume that there is a round in which \(A\) watches the root node (if there is no such round, we can modify \(A\) such that in a single extra round before submitting its answer, \(A\) watches the node it will claim is the source).

From this, we construct an algorithm \(B\) which has strictly more operations available to it. Concretely, \(B\) may watch an arbitrary number of nodes. In any round, \(B\) chooses to watch the same node \(A\) watches, as well as all nodes that \(A\) has observed being infected. \(B\) then terminates once it has watched the source (the algorithm can tell because the node gets infected but not via an edge). Since we assumed \(A\) watches the node it believes to be the source before terminating, if \(A\) wins the source detection game, so does \(B\). Also observe that \(B\) requires at most as many rounds as \(A\).

Observe that, since the players do not know the underlying graph, \(A\) can either select a node about which it has information or one about which it does not have information.
If \(A\) has information about a node, that means it has been picked before or is the source of an infection detected at another node.
Since the behavior of the source is always the same, the first of these two options yields no extra information, and thus we can assume \(A\) never makes such a choice.

By assumption, \(A\) requires less than \(c \sqrt{n}\) rounds (with probability at least \(p\) and for large enough \(n\)). Thus, in order for \(B\) to have picked the source at the end, \(A\) must have picked a node with distance at most \(c \sqrt{n }\) from the source at some time. In order for that to happen, \(A\)  must have picked a node about which it had no information and which has distance at most \(2 c \sqrt{n}\) from the source. Since \(\funnyconstant > 2\), The probability of that is smaller than
$$
1-\left( \frac{n-\funnyconstant c\sqrt{n}}{n}\right)^{c \sqrt{n}}.
$$
If we now pick  \(c< \sqrt{\ln\left(\frac{1}{(1-p/\funnyconstant)^2}\right)}\) and \(n\) large enough such that \(c\sqrt{n} \ge 1\), we have that this success probability is less than p. This follows since, by choice, \(-2c^2 > \ln(1-p/\funnyconstant)\) and thus \(e^{-2c^2} = \left(e^{2c} \right)^{c} \ge 1-p/\funnyconstant\). As \(\left(1-\frac{\funnyconstant c}{\sqrt{n}}\right)^{\sqrt{n}}\) approaches \(e^{\funnyconstant c}\) and \(\frac{n-\funnyconstant c\sqrt{n}}{n} = 1- \frac{\funnyconstant c}{\sqrt{n}}\), we have the desired result \(1-\left( \frac{n-\funnyconstant c\sqrt{n}}{n}\right)^{c \sqrt{n}} < p\) for large enough \(n\). This contradicts the assumption that \(A\) had success probability at least \(p\), proving that no such algorithm may exist.
\end{proof}

Also note that since all the graphs in the described family are trees, this result also holds in the same model but on trees. Similarly, observe that the algorithm \(B\) can simply be extended to support an algorithm \(A\) which may watch multiple nodes at the same time. Now, if \(A\) may watch \(k\) nodes, then \(B\) has a success probability of at most
$$
1-\left( \frac{n-2c\sqrt{n}}{n}\right)^{kc \sqrt{n}},
$$
where a similar result holds.

We now give a useful primitive to use in the construction of more complex algorithms.
We essentially prove that the source can, in expectation, not hide too many infections from the Discoverer.
Or looked at the other way round: if there is a linear number of infections, a Discoverer employing this strategy likely learns of at least one infection.
Interestingly, this result even holds in our settings that are most difficult for the Discoverer.

\begin{lemma}
Consider an instance of the source detection game (either with known or unknown graph and either consistent or obliviously dynamic source behavior).
Then there is a strategy for the Discoverer such that, after termination of the strategy, the Discoverer has observed a node in a round in which it gets infected.
This strategy succeeds after tolerating at most \(3n/2\) infections in expectation.
\label{lem:linear-discovery}
\end{lemma}

\begin{proof}
In each round, the strategy simply picks a node to watch uniformly at random.
It terminates when it has observed a that which gets infected while being watched.
Let \(T\) be the random variable that takes the number of the round in which this happens.
Let \(a_1, \dots\) be the nodes infected in the respective rounds.
Then the expected number of tolerated infections is
\begin{align*}
&\E{\sum_{i=1}^T a_i}
= \sum_{i = 0}^\infty a_i \cdot \Prob{i \le T}
= \sum_{i = 0}^\infty a_i \cdot \prod_{j=1}^{i-1} \left(\frac{n-a_j}{n}\right),
\intertext{where the equalities follow from the alternative definition of the expectation and since we require at least \(i\) rounds iff the first \(i-1\), rounds are unsuccessful. Now, substitute \(p_i = a_i / n\) for all \(i\). Then \(p_i\) is the probability of finishing in round \(i\) assuming it is reached. We may also pull \(n\) outside the sum to view the number of tolerated infections in each round as a multiple of \(n\). This leaves us with}
=~ &n \sum_{i = 0}^\infty p_i \cdot \prod_{j=1}^{i-1} (1-p_j).
\intertext{Now group the rounds of our strategy and thus the \(p_i\) by taking consecutive elements until their sum is at least \(1/2\) and then start the next group. Formally, let \(\{k_j\}_{j \in \N}\) be the (possibly finite) sequence such that \(k_j\) is the smallest integer such that \(\sum_{i = k_{j}}^{k_{j+1}-1} p_i \ge 1/2\). Then, by union bound, the probability of finishing within a given group, assuming we reach it, is at least \(1/2\). Also, we have that \(\sum_{i = k_j}^{k_{j+1}-1} \le 3/2\) since all but the last elements together are less than \(1/2\) and the last element is at most 1. This is a bound on the tolerated infections (as a multiple of \(n\)). Therefore, we can bound the above expectation from above by the geometric process that tolerates \(3n/2\) infections to perform the rounds in a given group and has probability \(1/2\) of finishing within that block. Thus, we can bound above by}
\le~ &n \sum_{j=1}^\infty \left( \sum_{i=k_j}^{i=k_{j+1}-1} p_i \right) \cdot \left(\prod_{\ell=1}^{j-1}\left( 1 - \sum_{i=k_\ell}^{i=k_{\ell+1}-1} p_i \right) \right)  \le n \sum_{j=1}^\infty 3/2 \cdot 1/2^j = 3n/2.
\end{align*}
\end{proof}

By Markov's inequality, the probability that this strategy takes at most \(3n\) infections is at most \(1/2\)—which famously is a constant.
\chapter{When Knowing the Static Graph Helps \label{sec:known}}
\label{known}
We would expect that knowing the static graph where the infections take place puts the Discoverer at a significant advantage.
Surprisingly, we prove that generally, this is not the case.
Motivated by this negative result, we investigate which knowledge about the static graph has value for the Discoverer.
We show that if the static graph has treewidth \(\tw{}\), the price of detection is in \(\BigO(\tw\cdot n \log n)\).
Finally, we argue that the central property that allows this is that the graph then recursively has small separators if it has small treewidth.

\begin{theorem}
Let \(A\) be a Discoverer algorithm for the source detection game with consistent source behavior on known static graphs. Then there is an algorithm \(A^u\) for the source detection game with consistent source behavior on unknown static graphs such that if \(A\) wins the game within \(O(f(n))\) infections with probability \(g(n)\), then so does \(A^u\) for its game.
\label{thm:consistent-known-to-unknown}
\end{theorem}

\begin{figure}
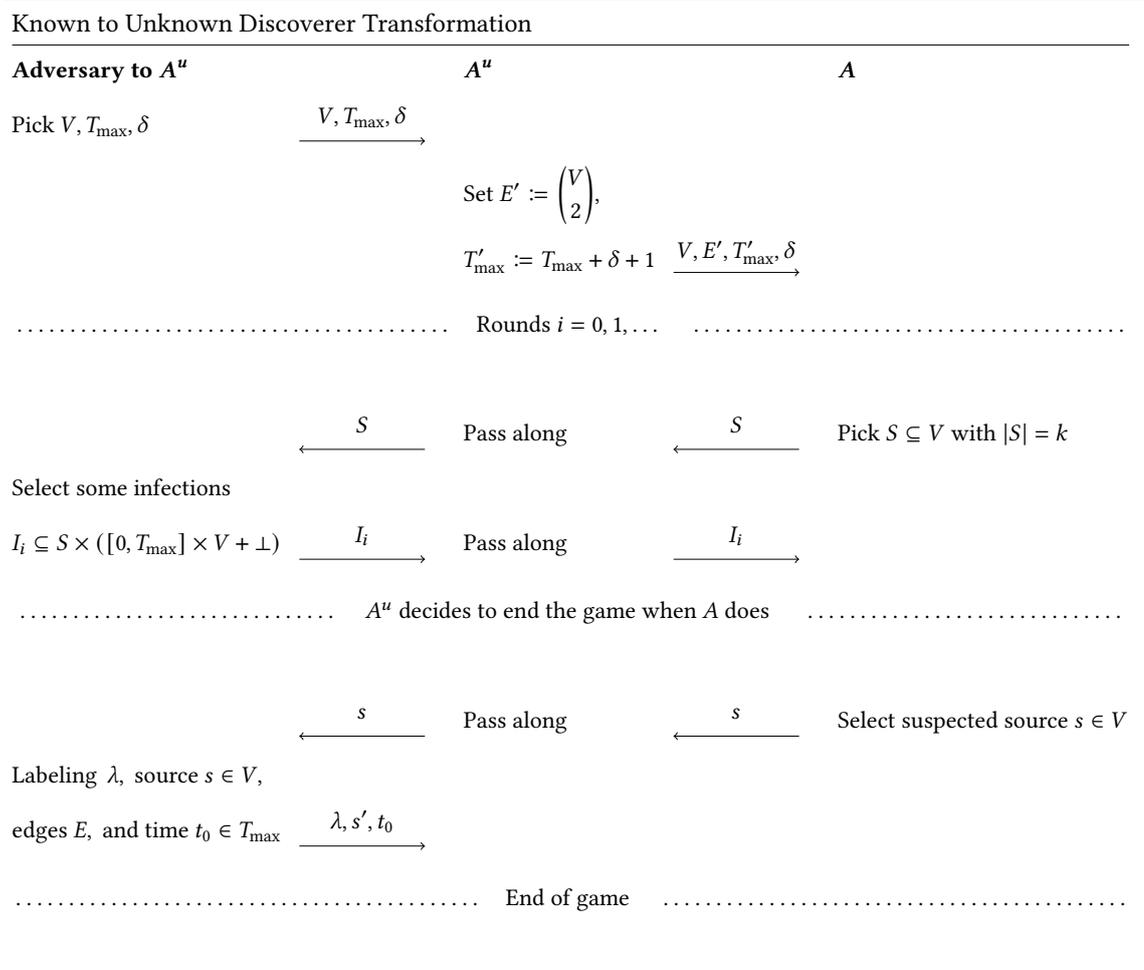

\rule{\textwidth}{0.4pt}
\resizebox{\textwidth}{!}{%
\begin{minipage}{1.2037\textwidth}
\procedureblock{Known to Unknown Discoverer Transformation}{%
\textbf{Adversary to \(A^u\)} \> \> \bm{A^u} \> \> \bm{A} \\
\text{Pick } V, \Tmax, \iphase \> \sendmessageright*{V, \Tmax, \iphase} \>  \> \> \\
\> \> \text{Set } E' \coloneqq \binom{V}{2}, \\
\> \> \Tmax' \coloneqq \Tmax + \iphase + 1 \> \sendmessage{->}{top={$V, E', \Tmax', \iphase$}} \> \> \pclb
\pcintertext[dotted]{\quad Rounds \(i=0,1,\dots\) \quad}  \\
\> \sendmessageleft*{S} \> \text{Pass along} \> \sendmessageleft*{S} \> \text{Pick } S \subseteq V \text{ with  } |S| = k \\
\text{Select some infections} \\
I_i \subseteq S \times ([0,\Tmax] \times V + \bot)  \> \sendmessageright*{I_i} \> \text{Pass along } \> \sendmessageright*{I_i} \> \pclb
\pcintertext[dotted]{\quad \(A^u\) decides to end the game when \(A\) does \quad}  \\
\> \sendmessageleft*{s} \> \text{Pass along} \> \sendmessageleft*{s} \> \text{Select suspected source } s \in V \\
\text{Labeling } \edgeLabelOp, \text{ source } s \in V, \\
\text{edges } E, \text{ and time } t_0 \in \Tmax \> \sendmessageright*{\edgeLabelOp, s', t_0} \> \> \> \pclb
\pcintertext[dotted]{\quad End of game \quad}%
}
\end{minipage}
}
\vspace{-0.25\baselineskip} \rule{\textwidth}{0.4pt}
\caption{Deriving an algorithm \(A^u\) for unknown static graphs from an algorithm \(A\) for known static graphs. \(A^u\) acts as the adversary to \(A\) and uses its behavior to determine its behavior towards its own Adversary. \label{alg:known-to-unknown}}
\end{figure}

\begin{proof}
We construct \(A^u\) from \(A\) by simulating the Adversary for \(A\) and using \(A\)'s responses to talk to \(A^u\)'s Adversary.
Let \(V, \Tmax, \iphase\) be the parameters that \(A^u\) receives from the Adversary in the first step of the game.
In the beginning, we pick \(E' \coloneqq \binom{V}{2}\), that is, we report the complete graph to \(A\).
Similarly, we pick \(\Tmax' \coloneqq \Tmax = \iphase + 1\).
We then report these modified parameters to \(A\).
In the rounds phase of the game, we simply relay watching queries and responses between the Adversary and \(A\).
Finally, we simply take \(A\)'s guess of the source and use it as \(A^u\)'s guess.
See \Cref{alg:known-to-unknown} for a diagram of this construction.

Observe that there is a one-to-one correspondence between the infections in the two games.
Thus, if \(A\) finishes within \(O(f(|V|))\) infections for some function \(f \colon \N \to \N^+\), so must \(A^u\).

Define
\begin{align*}
\edgeLabelOp' \coloneqq u,v \mapsto \begin{cases}
\edgeLabel{u,v}, &\text{if } uv \in E, \\
\Tmax', &\text{otherwise.}
\end{cases}
\end{align*}
Consider \Cref{fig:consistent-known-to-unknown} for an illustration of this construction.
Observe that if \(E, \edgeLabelOp, s', t_0\) is consistent with all \(I_i\), then so is \(E', \edgeLabelOp, s', t_0\).
Thus, if the Adversary wins against \(A^u\), the simulated Adversary wins against \(A\).
By contraposition, we may assume that if \(A\) wins against the simulated Adversary, so does \(A^u\) win against its Adversary, yielding the desired translation of success probabilities from \(A\) to \(A^u\).
\end{proof}

\begin{figure}[t]
\centering
\includegraphics[width=5cm]{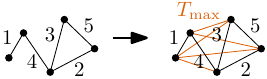}
\caption{\label{fig:consistent-known-to-unknown}The construction from the proof of \Cref{thm:consistent-known-to-unknown}. Discoverer for known graphs \(A\) is used on the completed graph on the right to win the game in the setting of unknown graphs.}
\end{figure}

After this negative result, we explore when knowing the static graph does help.
In particular, this also gives us a structural insight into which kinds of static graphs are easy to discover the time labels on.

\begin{theorem}
There is an algorithm that wins the source detection game with consistent source behavior on known graphs while only tolerating \(\BigO(\tw \cdot n \log n)\) infections with constant probability, where \(\tw\) is the treewidth of the graph.
\label{thm:aX1-tw}
\end{theorem}

\begin{figure}[t]
\centering
\includegraphics[width=4cm]{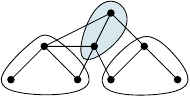}
\caption{\label{fig:aX1-tw}A graph with a balanced separator of size 2 (blue). An infection chain that includes nodes from both the left and right partitions must pass through the separator. Thus, in these cases, watching the separator reveals which of the partitions the source is in.}
\end{figure}

This result is achieved by \Cref{alg:aX1-tw}.
It crucially depends on the existence of small separators.
Specifically, in a static graph \(G\) with treewidth \(\tw{}\), there is a \(1/2\)-balanced separator \(S\) (i.e.,~every connected component in \(G-S\) has size at most \(1/2 \nodesetsize{G}\)) of size at most \(\tw + 1\) \cite{cygan2015parameterized}.
See \Cref{fig:aX1-tw} for an illustration.
This result also holds if the nodes are weighted, and we ask for a separator to ensure each component has at most half of the weight.
Therefore, if we assign either \(0\) or \(1\) as weights, we may pick which nodes we want to evenly distribute on the sides of the separator.

.
\begin{algorithm}[tbhp]
    \DontPrintSemicolon
    \SetKwFunction{FMain}{TreewidthSourceDetection}
    \SetKwProg{Fn}{fun}{:}{}
    \Fn{\FMain}{
       Maintain a set of nodes that could still be the source.

       Until the watched node is the source, repeat:

       1. Compute a balanced \( (\tw+1)\)-size separator of the candidate nodes.

       2. Watch each of the nodes in the separator for one round.

       3. If none of them get infected, watch one of the candidate nodes uniformly at random until one is infected. \label{alg:aX1-tw:random-search}

       4. Update the candidate nodes to only include nodes on the component induced by removing the separator that first infected a node in the separator.

    }

\caption{The source detection algorithm for static graphs with bounded treewidth and consistent source behavior.}
\label{alg:aX1-tw}
\end{algorithm}

\begin{proof}
The claimed properties of the algorithm follow, since the following three properties hold after each iteration of the main loop:
\begin{enumerate}
\item We correctly track the candidate nodes (i.e.,~the source must always be one of the remaining candidates).
\item After each iteration of the main loop of the algorithm, the number of nodes halves.
\item With constant probability, we only tolerate a linear number of infections until detection.
\end{enumerate}

The first property holds inductively.
At the beginning, clearly, the source is one of the candidates since all nodes are candidates.
Then, after each loop, exactly one of three things must have happened: (a) the source was one of the separator nodes; (b) the source is not part of the separator but infected at least one node in the separator; or (c) the source infected no node in the separator, and we found an infected node via the search in step 3.
In case (a), we will have found the source, and the Discoverer wins the game.
In case (b), the source must be in the component of the separated graph from which the infection first entered the separator.
In case (c), the source must be in the same component as the infected node, since if it were not, the infection chain must have included at least one node of the separator.
Thus, the source is always one of the candidate nodes.

The second property simply follows since the separator balances the candidate nodes.
That is, every component induced by removing the separator has at most half of the candidate nodes.
As we have seen, after each round, we select one of these components to restrict the candidate nodes to.

The third property follows from applying \Cref{lem:linear-discovery} to each iteration of the loop, then summing the costs, applying the linearity of expectation and finally Markov's inequality.

From these three properties, we see that the main loop runs at most \(\log n\) times and that each of its iterations incurs \(O(\tw \cdot n)\) infections wcp.
\end{proof}

Intuitively, this dependence on the existence of separators makes sense.
Looking for the source means that in each round, we have to decide where to look next based on the information gleamed so far.
Together, \Cref{thm:aY1} and \Cref{thm:aY1-lb} show that in the case of unknown graphs, we cannot do much better than testing a few nodes and then retracing their infection path node-by-node.
In some cases, we can do better by performing a binary search for the source, but in order to decide which half of our candidate nodes infections stem from, we have to separate them.
Such separators may not be too large since we need to spend rounds on watching them, thus this technique is only useful if small separators exist recursively.
This explains the connection between small treewidth and a smaller price of detection.
\chapter{The Special Case of Trees \label{sec:trees}}
\label{trees}
First, observe that our algorithm exploiting the treewidth of a known graph (see \Cref{thm:aX1-tw}) directly translates to trees.

\begin{corollary}
There is an algorithm that wins the source detection game with consistent source behavior on known trees while only tolerating \(\BigO(n \log n)\) infections with constant probability.
\label{cor:aX1-tree}
\end{corollary}

\begin{proof}
This follows from the \Cref{thm:aX1-tw} since trees have treewidth 1.
\end{proof}

In the case of trees, we can also show a strong matching lower bound.

\begin{theorem}
There is an infinite family of trees \(\{P_i\}_{i \in \N}\) such that, for any algorithm solving the source detection problem with consistent source behavior and known static graph on trees graphs with success probability \(p >0\),  the price of detection under that algorithm approaches \(n \log n\) as \(n\) approaches \(\infty\).
\label{thm:aX1-lb-tree}
\end{theorem}

Note that this implies that there is no algorithm winning this game setup in \(o(n \log n)\) with constant probability.
Again, the proof can be adjusted to instead hold for algorithms that always win the game and have at least a constant probability to finish with less than \(o(n \log n)\) infections.

Surprisingly, the concrete success probability does not play a role for the result (as long as it is positive and constant). Meaning that, for large enough graphs, we cannot trade a lower success probability for a multiplicative improvement over the \(n \log n\) number of tolerated infections.
Compare that to the weaker lower bound from  \Cref{thm:aY1-lb}, where this door is left open (though there is no known algorithm to exploit it).

\begin{proof}
First of all, let \(n \in \N\). Then set \(P_n\) as the path with \(n\) nodes. Set \(\iphase = n\). Now let \(A\) be any algorithm that wins the source detection game  restriction with consistent source behavior on known trees in \(o(n \log n)\) rounds with a constant success probability \(p\). The adversary picks the source node \(s\) uniformly at random. Then, let the adversary set the edge labels for all \(i < s\) as \(\edgeLabel{i,i+1} = n - i\) and for all \(i \ge s\) as \(\edgeLabel{i,i+1} = i\).

Note that we require the algorithm to solve the problem on any tree and with any source with constant probability. We now model all possible randomized algorithms as a distribution over what we call execution trees. These execution trees encode the reactions the algorithm may have to all possible responses by the adversary.  An instance of the game is then equivalent to picking an execution tree at the beginning and evaluating the responses of the adversary against the execution tree, yielding a path through the execution tree.

The behavior of the Discoverer can be described as the series of nodes it picks to watch. We may assume without loss of generality that the last node the Discoverer algorithm picks is the one it believes is the root (this is similar to the proof of \Cref{thm:aY1-lb}). Now, by the definition of the source detection game, in each round, the Discoverer submits a node to watch and receives information about when and from which direction it was infected (if at all). By definition of our instances, every node is infected in every round. Thus, the only information the Discoverer receives is from which direction a node was infected and at which time step.

Observe that the edge labels only encode whether the edge is on the left or right side of the source, which the Discoverer also learns because it learns the direction along which the infection travels over the edge.
Thus, we may disregard the effect that this knowledge has on the behavior of the Discoverer.

We construct this (binary) tree as follows. As the Discoverer algorithm is randomized, which node it picks in the next rounds, is dependent on previous information and random choices. We can thus assume the Discoverer makes all its random choices at the start of the execution and then builds a tree of possible outcomes. Each tree node is labeled with the node the Discoverer will watch in a certain round. Then, the left and right children are the choices, the Discoverer algorithm will make if the watched node is infected by its left and right neighbor respectively.

Thus, after picking an execution tree, the behavior of the Discoverer algorithm is deterministic. To complete the result, we show that for any possible execution tree with height in \(o(\log n)\), there are a sufficient amount of source nodes for which the algorithm would not win.

Clearly, an execution of the game is now associated with the execution tree the Discoverer algorithm picks at the start and the path through that tree. Since we assumed the algorithm always watches the source node in at least one round, an execution is only winning if the path taken includes the source nodes.

Assume the algorithm terminates after \(r\) rounds and let \(c \in \R^+\) such that \(r = c \log n\).
We may assume that the execution path has length at most \(r\).
Since the Discoverer can now only win if the source is in the \(r\) first layers of the execution tree (in which there are \(2^r\) nodes) and the source was picked uniformly at random, we have that \(2^r / n \ge p\).
Thus, \(2^{c \log n} / n \ge p\), which rearranges to \(n^c/n \ge p\).
Applying the logarithm and rearranging leaves us with
\begin{align*}
c \ge 1 + \log p / \log n.
\end{align*}

Observe that the right-hand side approaches \(1\) as \(n\) approaches \(\infty\).
Thus, \(r\) approaches \(\log n\) and since on our family of graphs, all nodes are infected in each round, the number of tolerated infections approaches \(n \log n\).
\end{proof}
\chapter{Dynamic Infection Behavior \label{sec:dyn}}
\label{dyn}
In this variation of the game, the Adversary may vary the time at which the source is seed-infected, but not the source itself.
The Discoverer still learns when the watched node becomes infected each round.
Note that the Adversary does not learn which nodes the Discoverer watches, and so its behavior may not depend on the choices of the Discoverer.
We will see that for this variation, the simple deterministic brute-force algorithm is asymptotically optimal.

\begin{theorem}
There is a Discoverer algorithm that wins any instance of the source detection game on unknown graphs and with obliviously dynamic source behavior that tolerates at most \(n^2\) infections.
\label{thm:dY1}
\end{theorem}

Notice that this result is not randomized; this algorithm is always successful.
The game is also the most difficult (for the Discoverer) we discuss.
Thus, this result directly translates to all the other game variations we looked at.

\begin{proof}
The algorithm simply spends one round watching each node.
Since watching the source reveals it, we definitely find the source.
In each round, at most \(n\) nodes become infected, and we require \(n\) rounds.
Thus, we tolerate at most \(n^2\) infections.
\end{proof}

On this difficult model, the result is indeed tight, as we see in the following theorem.
In fact, we cannot even improve upon the quadratic cost of detection if we only ask for an algorithm with a constant probability of success.

\begin{theorem}
For any algorithm solving the source detection problem with obliviously dynamic behavior, known static graph, and watching a single node with success probability \(p >0\), there is an infinite family of graphs such that the price of detection under that algorithm is at least \(\min(1/3, p/2) n^2\).
\label{thm:dX1-lb}
\end{theorem}

\begin{figure}[t]
\centering
\includegraphics[width=4cm]{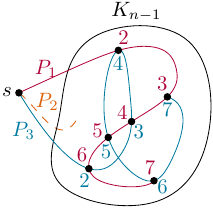}
\caption{\label{fig:dX1-lb}The construction from the proof of \Cref{thm:dX1-lb}. \(s\) is the source of the graph and \(P_1, P_2\) and \(P_3\) are Hamiltonian paths. For simplicity, \(P_2\) is omitted. Nodes are labeled by their position on \(P_1\) and \(P_3\).}
\end{figure}

\begin{proof}
Let \(n \in 2\N+1\).
Then set \(\G\) to be the complete graph on \(n\) vertices.
Let \(P_1, \dots, P_{(n-1)/2}\) be edge disjoint Hamiltonian paths on \(\G\), all starting at some fixed node \(s\).

We may construct them as follows: take \((n-1)/2\) Hamiltonian paths on the \(K_{n-1}\) (which exist by \cite{axiotis_approx}), add a node \(s\) and connect \(s\) to all other nodes.
Take that node as the first and last node on every path to form \((n-1)/2\) Hamiltonian cycles.
Turn them into paths starting at \(s\) arbitrarily.

Order the nodes and the Hamiltonian paths uniformly at random.
Note, these cover all edges in \(\G\).
Write \(p \colon \edgeset{\G} \to [n/2]\) be the function that maps an edge to the Hamiltonian path it belongs to and write \(q \colon \edgeset{\G} \to [n]\) for the function that assigns an edge its index on that path.
Then let the Adversary act according to \(\edgeLabelOp \coloneqq e \mapsto p(e)n + q(e),\) and pick \(\iphase = 1\).
See \Cref{fig:dX1-lb} for an illustration of this construction.
In each round \(i\), the Adversary lets the source infect the whole graph via \(P_{i \modOp (n-1)/2}\).
For the sake of this proof, we may assume the Discoverer knows the scheme according to which these paths are picked (but not in which order the nodes appear on a specific path).
Thus, if the Discoverer knows the label of an edge \(e\), it may deduce which path it belongs to (by calculating \(\edgeLabel{e} \divOp n\)) and what the index of the edge has on the path (by calculating \(\edgeLabel{e} \modOp n\)).

We only consider the first \((n-1)/2 - 1\) rounds, since after these, there will have been \(n((n-1)/2-1) \in \Omega(n^2)\) infections.
This is where the \(1/3\) in the minimum in this theorem comes from, since for large enough \(n\), we have \(n((n-1)/2-1) > 1/3 \cdot n^2\).
This edge case handled, we assume for the rest of this proof, that during our game there is exactly one infection chain per path.

Intuitively, the Discoverer has information about some of the edges for each path, and the goal of the Discoverer is to find the start of any of the paths.
Formally, we prove that the Discoverer must have observed an infection involving the source (i.e., the first node on all paths), that is, it must have watched the first or second node on a path.
By construction, an observed infection on one path reveals very limited information about the other paths.
An observed infection reveals the positions of the two involved nodes on the current path.
This only tells the Discoverer that they may not be adjacent on any other path.
Within the first \((n-1)/2-1\) rounds, the Discoverer learns precisely \((n-1)/2 - 1\) of these pieces of information, thus for each node, there are at least \(n/2\) positions on the current path that may be consistent with this information.
Also, after \((n-1)/2-1\) rounds, there are at least two nodes on which the Discoverer has no information.
Thus, if the Discoverer never observes the first or second node on a path, the adversary may rearrange the paths in such a way, that one of the unobserved nodes becomes the source (it, of course, picks whichever node the Discoverer does not claim is the source).

Since the positions are distributed uniformly at random, no strategy that the Discoverer may use has a probability larger than \(2/n\) of picking the first or second node on the current path, thus advancing the requirements outlined above.
Therefore, any algorithm that has success probability at least \(p\), must have at least probability \(p\) of picking the first or second node on a path.
Thus, by union bound, the Discoverer must have spent at least \(np/2\) rounds, that is, tolerated \(n^2 p/2\) infections.
\end{proof}

\begin{theorem}
There is a Discoverer algorithm that wins any instance of the source detection game on known tree graphs under obliviously dynamic source behavior while watching at most two nodes per round. This algorithm tolerates \(\BigO(n \log n)\) infections with constant probability.
\label{thm:dX2}
\end{theorem}

\begin{proof}
Similar to \Cref{alg:aX1-tw}, we introduce an algorithm that relies on the existence of small spanners.
Concretely, it is a well known result that a tree has a \emph{centroid decomposition} \cite{centroids}.
A centroid is a node such that its removal disconnects the tree into two components at most half the size of the original tree.
Since, any tree has a centroid and the two components are again trees, we can decompose a tree into a series of these centroids and the respective split trees.

The algorithm now proceeds as follows: \todo[color=pink]{@supervisiors: move to figure?}
\begin{enumerate}
\item Maintain a subtree of candidate nodes, and always compute a centroid as a balanced separator. Start with the whole graph as the candidates.
\item In each round: watch the current separator and one node in the current subtree picked uniformly at random.
\item If you receive information about the subtree to go into (either because the separator or the other node was infected), do so.
\begin{enumerate}
\item If the separator is infected, recurse into the side of the separation from which the infection originated.
\item If the randomly picked node is infected but not the separator, recurse into the side of the separation this node is a part of.
\end{enumerate}
\end{enumerate}

For the correctness, note that infection chains always form a directed, connected subtree of the entire graph.
Thus, if the infection stems from one side of the separation, then the source must be part of that subgraph.
Similarly, if the separator is not infected but a node in one of the subtrees is, then the source must be in that subtree.

For the price of discovery, observe that we have to recurse \(\log \nodesetsize{\G}\) times as the centroid is a balanced separator, and thus each time we do so, the subtree containing the remaining candidate nodes halves in size.
Also, by \Cref{lem:linear-discovery}, at most \(n\) nodes are infected in expectation until we observe an infection at the randomly picked node.
Then we may recurse in step 4.

Thus, in expectation, we tolerate \(\BigO(n \log n)\) infections.
Apply Markov's inequality to get the desired result.
\end{proof}

Note that this is the only result in this paper that depends on the Discoverer's ability to watch more than one node at a time.
As we will see in the next theorem, this can only be necessary if the source has obliviously dynamic behavior.
We currently do not know of a lower bound proving that watching more than one node in these scenarios provides an asymptotic advantage.
Intuitively, the above algorithm exploits this capability to mitigate one of the problems with dynamic source behavior: we cannot trace back along an infection chain, as it is unclear if in a given round we see similar behavior to the previous one.
By watching both the next and previous node to be tested, we can extenuate this issue.

\begin{theorem}
Let \(A\) be a Discoverer algorithm for the source detection game with consistent source behavior (with either known or unknown static graph) while watching \(k\) nodes per round.
Then there is a Discoverer algorithm \(A^1\) such that if \(A\) tolerates at most \(x\) successful infections in some instance of the game, \(A^1\) tolerates at most \(k x\) infections in the same instance.
\label{thm:aXYk-aXY1}
\end{theorem}

This proves that, asymptotically, the ability to watch two nodes does not help the Discoverer in the setting with consistent source behavior.
Note, this theorem and its proof can trivially be extended to randomized bounds.

\begin{proof}
\(A^1\) simply simulates \(A\) in the following fashion: if \(A\) watches \(\ell \le k\) nodes \(v_1, \dots, v_\ell\) in some round \(i\), then \(A^1\) spends \(\ell\) rounds watching \(v_1, \dots, v_\ell\) individually.
Both the correctness and the price of detection bound follow since the behavior of the source is consistent.
Because of the consistency, \(A^1\) gains precisely the same information in the individual rounds for \(v_1, \dots, v_\ell\) as \(A\) does in the one round for \(v_1, \dots, v_\ell\).
Also, because the source behaves consistently, it infects the same number of nodes in each round.
Since \(A^1\) splits each round of \(A\) into at most \(k\) rounds, the number of tolerated infections increases at most by a factor of \(k\).
\end{proof}

We do not know if this natural relationship holds when the source behavior is obliviously dynamic.
Intuitively, our argument does not translate, since we cannot rely on the argument that splitting one round with many watched nodes into many rounds with a single watched node yields the same information and tolerates the same number of infections.
\chapter{Summary \& Future Work}
\label{sec:org570ff46}
Our work provides an extensive theoretical study of source detection under the SIR model in temporal graphs.
We give a precise definition of source detection in this setting via an interactive two-player game.
Using randomization, we overcome the challenging worst-case behavior and offer efficient algorithms for many settings, such as on trees and general graphs, under consistent and obliviously dynamic source behavior as well as for known and unknown static graphs.
For all but one of the many settings we investigated, we have matching lower bounds proving our algorithms are asymptotically optimal among all algorithms winning the source detection game with constant probability, allowing us to precisely characterize the respective difficulty of these settings.

Our work could naturally be extended by studying source detection under the susceptible-infected-susceptible model.
As many of our lower bounds rely on the resistance of recovered nodes, seeing which results translate promises to be insightful.
Closing the remaining gap in \Cref{tbl:results} would require investigating source detection on unknown graphs with obliviously dynamic source behavior where the underlying static graph is guaranteed to be a tree.
And finally, the most compelling unanswered question is whether allowing the Discoverer to watch multiple nodes in the setting with obliviously dynamic source behavior allows for asymptotically fast algorithms.
We know this is not the case in the setting with consistent source behavior by \Cref{thm:aXYk-aXY1}.
Yet, while we use the additional capability in \Cref{thm:dX2}, we have no proof it is fundamentally necessary to achieve the better running time.

\FloatBarrier

\chapter{Conclusion \label{sec:thesis-conclusion}}
\label{sec:orgd6943b4}
While we give the first formal models for infections in temporal graphs and thoroughly investigate their complexity in a number of different settings, the field has many compelling open questions to offer.
In addition to the research questions outlined at the end of \cref{part:graph-discovery,part:source-detection}, there are a number of more general, but just as exciting, questions for the study of spreading processes in temporal graphs.

Extending our work to other infection models (such as the linear threshold or SIS models) is sure to yield interesting insights into how these different spreading behaviors interact with changing graphs.
In particular, seeing which of our results translate to different models and which do not promises to shed light on the complex interaction between infections and their underlying network.

In \cref{part:source-detection} we have seen the power of randomized algorithms to mitigate worst-case behavior.
Whether similarly faster randomized algorithms can be obtained for temporal graph discovery is an open question.

Finally, our model of infections itself is deterministic, that is, the spread of an infection in a temporal graph is entirely determined by the seed nodes and times.
This simplification is not fully realistic, which leads to the natural question of how our graph discovery and source detection approaches could be adapted to a model with probabilistic infections.
Such probabilistic models have been studied extensively for static graphs.
Here, the independent cascade \cite{independent-cascade} is particularly influential.

Temporal graphs promise to be a fruitful research framework for studying spreading processes, allowing us to more accurately model and analyze these complex real-world phenomena. Theoretical insights in this area stand to benefit applications from the prediction of diseases to the combating of misinformation.

\iffalse
\bibliography{papers}
\fi

    % Following are the files and commands for the bibliography and the author’s publications.
    \pagestyle{plain}

    \renewcommand*{\bibfont}{\small}
    \printbibheading
    \addcontentsline{toc}{chapter}{Bibliography}
    \printbibliography[heading = none]

    \addchap{Declaration of Authorship}
    I hereby declare that this thesis is my own unaided work. All direct or indirect sources used are acknowledged as references.\\[6 ex]

\begin{flushleft}
    Potsdam, \today
    \hspace*{2 em}
    \raisebox{-0.9\baselineskip}
    {
        \begin{tabular}{p{5 cm}}
            \hline
            \centering\footnotesize\printAuthor
        \end{tabular}
    }
\end{flushleft}

\end{document}